\documentclass[a4paper,11pt]{article} 

\usepackage[english]{babel}
\usepackage[a4paper]{geometry}
\usepackage{enumerate}  
\usepackage{amsmath}
\usepackage{dsfont}
\usepackage{graphicx}
\usepackage{color}
\usepackage{amsthm}
\usepackage{amssymb}

\usepackage{hyperref}    

\newcommand{\hc}{\mbox{h.c.}}

\renewcommand{\le}{\leqslant}
\renewcommand{\leq}{\leqslant}
\renewcommand{\ge}{\geqslant}
\renewcommand{\geq}{\geqslant}

     \let\g=\gamma     \let\d=\delta     
             \let\l=\lambda
                  \let\p=\pi        \let\r=\rho
\let\s=\sigma          \let\ph=\varphi   
       
        \let\L=\Lambda

\def\cE{{\cal E}}\def\cM{{\cal M}}\def\cV{{\cal V}}
\def\cC{{\cal C}}\def\cF{{\cal F}}\def\cH{{\cal H}}\def\cW{{\cal W}}
\def\cT{{\cal T}}\def\cN{{\cal N}}\def\cB{{\cal B}}
\def\cL{{\cal L}}\def\cJ{{\cal J}}\def\cQ{{\cal Q}}
\def\cD{{\cal D}}\def\cG{{\cal G}}
\def\cO{{\cal O}}\def\cK{{\cal K}}

  \def\v0{{\vec 0}}

\def\bal{{\bar \l}}

\def\bC{\mathbb{C}}
\def\bR{\mathbb{R}}

\def\cV{\mathcal{V}}
\def\cF{\mathcal{F}}
\def\cG{\mathcal{G}}
\def\cL{\mathcal{L}}
\def\cN{\mathcal{N}}
\def\cE{\mathcal{E}}
\def\cK{\mathcal{K}}
\def\cH{\mathcal{H}}
\def\eps{\varepsilon}
\def\ph{\varphi}

\def\bN{\mathbb{N}}  
\def\bZ{\mathbb{Z}} 
\def\bR{\mathbb{R}}

\def\indic{\hbox{\raise-2pt \hbox{\indbf 1}}}

\let\==\equiv

\let\0=\noindent

\def\*{{\hfill\break\null\hfill\break}}

\def\tende#1{\,\vtop{\ialign{##\crcr\rightarrowfill\crcr
			\noalign{\kern-1pt\nointerlineskip}
			\hskip3.pt${\scriptstyle #1}$\hskip3.pt\crcr}}\,}
\def\otto{\,{\kern-1.truept\leftarrow\kern-5.truept\to\kern-1.truept}\,}

\def\wt#1{\widetilde{#1}}
\def\wh#1{\widehat{#1}}
\def\sqt[#1]#2{\root #1\of {#2}}

\def\hc{{\rm h.c.}\,}

\def\wt{\widetilde}

\def\be{\begin{equation}}
	\def\ee{\end{equation}}
\def\bea{\begin{eqnarray}}\def\eea{\end{eqnarray}}
\def\bean{\begin{eqnarray*}}\def\eean{\end{eqnarray*}}
\def\bfr{\begin{flushright}}\def\efr{\end{flushright}}
\def\bc{\begin{center}}\def\ec{\end{center}}
\def\bal{\begin{align}} 
	\def\eal{\end{align}}

\def\spl#1\spl{\[ \begin{split}#1\end{split} \]}

\def\bd{\begin{description}}\def\ed{\end{description}}
\def\bv{\begin{verbatim}}\def\ev{\end{verbatim}}

\def\Halmos{\hfill\vrule height10pt width4pt depth2pt \par\hbox to \hsize{}}

\newtheorem{theorem}{Theorem}[section]
\newtheorem{prop}[theorem]{Proposition}
\newtheorem{lemma}[theorem]{Lemma} 
\theoremstyle{remark}

\numberwithin{equation}{section}
\title{Third order corrections to the ground state energy \\ of a Bose gas in the Gross-Pitaevskii regime}

\date{\today}

\begin{document}


\author{
	Cristina Caraci$^1$\footnote{Electronic mail: cristina.caraci@unige.ch}\;,
	Alessandro Olgiati$^2$\footnote{Electronic mail: alessandro.olgiati@polimi.it}\;,  Diane Saint Aubin $^3$\footnote{Electronic mail: diane.saintaubin@math.uzh.ch}\;,  Benjamin Schlein $^4$ \footnote{Electronic mail: benjamin.schlein@math.uzh.ch}\\[0.2cm]
	{\footnotesize $^{1}$Section of Mathematics, University of Geneva, rue du Conseil-G\'en\'eral 7-9, 1205 Geneva}\\{\footnotesize $^{2}$ Dipartimento di Matematica, Politecnico di Milano, Piazza Leonardo da Vinci 32, 20133 Milano }\\ {\footnotesize$^{3,4}$Institute of Mathematics, University of Zurich, Winterthurerstrasse 190, 8057 Zurich.}\\
}

	\maketitle

\begin{abstract} 
For a translation invariant system of $N$ bosons in the Gross-Pitaevskii regime, we establish a precise bound for the ground state energy $E_N$. While the leading, order $N$, contribution to $E_N$ has been known since \cite{LY,LSY} and the second order corrections (of order one) have been first determined in \cite{BBCS}, our estimate also resolves the next term in the asymptotic expansion of $E_N$, which is of the order $(\log N) / N$. 
\end{abstract}

\section{Introduction}

In the Gross-Pitaevskii regime, we consider a gas of $N$ bosons moving on the unit torus $\Lambda \simeq [0;1]^3$, interacting through a repulsive potential with scattering length of the order $1/N$. The Hamilton operator of such a system has the form 
\begin{equation}\label{eq:ham0}
H_N = \sum_{j=1}^N -\Delta_{x_j} + \sum_{i<j}^N N^2 V (N (x_i - x_j)) 
\end{equation} 
and acts, according to the bosonic statistics, on $L^2_s (\Lambda^N)$, the subspace of $L^2 (\Lambda^N)$ consisting of functions that are symmetric w.r.t. permutations. Here, we are going to assume that $V \in L^3 (\bR^3)$ is non-negative, compactly supported and spherically symmetric. We denote by $\mathfrak{a}$ its scattering length, which is defined through the solution $f$ of the zero energy scattering equation
\[ \left[ -\Delta + \frac{1}{2} V \right] f = 0 \]
with the boundary condition $f (x) \to 1$, as $|x| \to \infty$, by requiring that 
\[ f(x) = 1 - \frac{\mathfrak{a}}{|x|} \]
outside the support of $V$. By scaling, the scattering length of the potential $N^2 V (N\cdot)$ appearing in (\ref{eq:ham0}) is then given by $\frak{a}/N$. Observe that, after rescaling $x \to Nx$, the Gross-Pitaevskii regime equivalently describes a gas of particles interacting through the fixed potential $V$, at density $\rho = N^{-2}$.

As first established in \cite{LY,LSY}, the ground state energy $E_N$ of (\ref{eq:ham0}) is given, to leading order, by 
\begin{equation*}\label{eq:LY}  E_N = 4\pi \frak{a} N + o (N) \end{equation*} 
in the limit $N \to \infty$. In \cite{LS1,LS2,NRS}, it was also shown that the corresponding ground state vectors exhibit complete Bose-Einstein condensation; all particles, up to a fraction vanishing in the limit $N \to \infty$, can be described by the zero-momentum one-particle orbital $\ph_0 \in L^2 (\Lambda)$, defined by $\ph_0 (x) = 1$ for all $x \in \Lambda$. In the last years, more precise bounds on the rate of condensation have been derived. For $V \in L^3 (\bR^3)$, it was shown in \cite{BBCS4} that, for any normalized  sequence $\psi_N \in L^2_s (\Lambda^N)$ of approximate ground state vectors, satisfying 
\[ \langle \psi_N, H_N \psi_N \rangle \leq E_N + K ,\] 
the number of particles that are orthogonal to $\ph_0$ remains bounded by a constant proportional to $K$ (but independent of $N$), in the limit $N \to \infty$ (a simplified proof of condensation has been recently proposed in \cite{BBCO}, using the approach developed in \cite{B}). 

This optimal estimate on the rate of condensation was used in \cite{BBCS4} as input for a rigorous version of Bogoliubov theory \cite{Bo}, showing that the ground state energy of (\ref{eq:ham0}) satisfies 
\begin{equation}\label{eq:LHY-GP} E_N = 4\pi \frak{a} (N-1) + e_\Lambda \frak{a}^2 - \frac{1}{2} \sum_{p \in 2\pi \bZ^3 \backslash \{ 0 \}} \Big[ p^2 + 8\pi \frak{a} - \sqrt{|p|^4 + 16 \pi \frak{a} p^2} - \frac{(8\pi \frak{a})^2}{2p^2} \Big] + \cO (N^{-1/4}) \end{equation} 
and that the spectrum of $H_N - E_N$ below a threshold $\zeta > 0$ consists of eigenvalues having the form 
\begin{equation}\label{eq:excite} \sum_{p \in 2\pi \bZ^3 \backslash \{ 0 \}} n_p \sqrt{|p|^4 + 16 \pi \frak{a} p^2} + \cO ((1+\zeta^3) N^{-1/4}) \end{equation} 
where $n_p \in \bN$ for all $p \in 2\pi \bZ^3 \backslash \{ 0 \}$. In (\ref{eq:LHY-GP}), we defined  
\begin{equation}\label{eq:eLambda} e_\Lambda = 2 - \lim_{M \to \infty} \sum_{\substack{p \in \bZ^3 \backslash \{ 0 \} : \\ |p_1|, |p_2| , |p_3| \leq M}} \frac{\cos |p|}{p^2} \end{equation} 

The optimal bound on the condensation rate and the estimates (\ref{eq:LHY-GP}), (\ref{eq:excite}) on the low-energy spectrum of (\ref{eq:ham0}) have been later extended to Bose gases trapped by an external potentials in \cite{NNRT,BSS1,NT,BSS2}, to bosons moving in a box with Neumann boundary conditions in \cite{BS}, 
to systems interacting through a potential with scattering length of the order $N^{-1+\kappa}$, for sufficiently small $\kappa > 0$, in \cite{ABS,BCaS} and to Bose gases in the two-dimensional Gross-Pitaevskii regime in  \cite{CCS1,CCS2}. Recently, an upper bound matching (\ref{eq:LHY-GP}) was proven in \cite{BCOPS}, for particles interacting through a non-integrable, hard-sphere potential. New and simpler proofs of (\ref{eq:LHY-GP}), (\ref{eq:excite}) have been obtained in \cite{HST} and, very recently, in \cite{B} (also beyond the Gross-Pitaevskii regime, for $\kappa > 0$ small enough). Some rigorous bounds are also available at positive temperatures; to leading order, the free energy in the Gross-Pitaevskii regime was determined in \cite{DS}, up to temperatures comparable with the critical temperature for condensation. Upper bounds for the free energy capturing also the next order corrections have been obtained in \cite{BDM,CD}. 

Bogoliubov theory has been recently also used to determine equilibrium properties of Bose gases in the thermodynamic limit, where we consider $N$ particles moving in the box $[0;L]^3$, with periodic boundary conditions, letting $N, L \to \infty$ keeping the density $\rho = N/L^3$ fixed. In \cite{LHY}, Lee-Huang-Yang derived a formula for the asymptotic behavior of the ground state energy per particle, in the dilute regime, to leading- and next-to-leading order. Their result was then improved by Wu in \cite{W}, by Hugenholtz-Pines in \cite{HP} and by Sawada in \cite{Sa}, who predicted that 
\begin{equation}\label{eq:Wu} e(\rho) = \lim_{\substack{N,L \to \infty \\ N/L^3 = \rho}} \frac{E_{N,L}}{N} = 4\pi \frak{a} \rho \left[ 1 + \frac{128}{15 \sqrt{\pi}} (\rho \frak{a}^3)^{1/2} + 8 \Big( \frac{4}{3} \pi - \sqrt{3} \Big) \rho \frak{a}^3 \log (12 \pi \rho \frak{a}^3) + \dots \right] \end{equation} 
up to lower order corrections, in the limit $\rho \frak{a}^3 \to 0$. The validity of the first term on the r.h.s. of (\ref{eq:Wu}) has been known since \cite{Dy} (upper bound) and \cite{LY} (matching lower bound). As for the second term in (\ref{eq:Wu}) (the Lee-Huang-Yang correction), a lower bound was proven in \cite{FS1} and, for more general interaction potentials (including hard-sphere interactions), in \cite{FS2}. Recently, an optimal lower bound was also derived in \cite{HHNST} for the free energy at positive temperature (chosen so that the energy of thermal excitations is comparable to the Lee-Huang-Yang correction). As an upper bound, the first two terms in (\ref{eq:Wu}) were first established in \cite{YY}. A simpler proof, which applies to more general potentials (but not to hard-sphere interactions) was obtained in \cite{BCS}. For the hard-sphere potential, on the other hand, the derivation of an upper bound matching (\ref{eq:Wu}) to second order is still an open problem. However, an upper bound establishing the validity of the first term on the r.h.s. of (\ref{eq:Wu}), with an error of the Lee-Huang-Yang order (but with the wrong constant) was recently proven in \cite{BCGOPS}. There is still no rigorous result about the third term on the r.h.s. of (\ref{eq:Wu}), neither as a lower nor as an upper bound to the ground state energy per particle. 

In this paper, we improve (\ref{eq:LHY-GP}), establishing the next contribution to the ground state energy in the Gross-Pitaevskii regime, which turns out to be of the order $(\log N)/N$. The next theorem is our main result. 
\begin{theorem} \label{thm:main}
 Let $V\in L^3(\Lambda)$ be non-negative, spherically symmetric, and compactly supported. Let $\Lambda^*_+ = 2\pi \bZ^3 \backslash \{ 0 \}$. Then the ground state energy $E_N$ of the Hamiltonian $H_N$ from (\ref{eq:ham0}) satisfies
 \begin{equation} \label{eq:main}
 \begin{split}
E_N=\;& 4\pi \mathfrak{a}(N-1)+e_\Lambda\mathfrak{a}^2\\
&-\frac{1}{2}\sum_{p \in \Lambda_+^*} \bigg[ p^2+8\pi\mathfrak{a}-\sqrt{|p|^4+16\pi\mathfrak{a}p^2}-\frac{(8\pi\mathfrak{a})^2}{2p^2}\bigg] \\ &- 64 \pi \Big( \frac{4}{3} \pi -\sqrt{3} \Big) \frak{a}^4  \frac{(\log N)}{N}  + \mathcal{O} ((\log N)^{1/2} /N)
\end{split}
 \end{equation}
 as $N\to\infty$, with $e_\Lambda$ defined in \eqref{eq:eLambda}.
\end{theorem}

{\it Remarks.} 
\begin{itemize}
\item[1)] Let $\frak{a}_N = \frak{a}/N$ be the scattering length of the potential in (\ref{eq:ham0}). With $\rho = N$, we observe that $\rho \frak{a}_N^3 = \frak{a}^3/N^2$. We conclude that third term on the r.h.s. of (\ref{eq:main}) is consistent with the prediction (\ref{eq:Wu}) for the third term in the asymptotic expansion of ground state energy per particle in the thermodynamic limit. 
\item[2)] With our analysis, we could also improve the estimate (\ref{eq:excite}) for the low-energy spectrum of $H_N - E_N$, showing that, below a threshold $\zeta > 0$, it consists of eigenvalues having the form 
\[ \sum_{p \in 2\pi \bZ^3 \backslash \{ 0 \}} n_p \sqrt{|p|^4 + 16 \pi \frak{a} p^2} + \cO (C_\zeta (\log N)^{1/2} /N ) \]
for an appropriate constant $C_\zeta > 0$, depending polynomially on $\zeta$.
\item[3)] Expansions of the ground state energy of Bose gases beyond second order have been previously obtained in the mean-field limit \cite{P3,BPS,NN}. Moreover, for systems of $N$ particles interacting through a potential of the form $N^{3\beta-1} V (N^\beta \cdot)$, for a $\beta \in (0;1)$, the ground state energy was recently resolved to order $N^{-1+\beta}$ in \cite{BLPR}.
\item[4)] At the expense of a slightly longer proof, with our techniques we could prove that the error term in \eqref{eq:main} is $\mathcal{O}(N^{-1})$.
\item[5)] Heuristically, the appearance of a term of order $(\log N) / N$ on the r.h.s. of (\ref{eq:main})  can be understood by perturbation theory. The bounds (\ref{eq:LHY-GP}), (\ref{eq:excite}) are proven in \cite{BBCS4} showing that, after appropriate unitary transformations, $H_N$ can be approximated by a Fock space Hamiltonian of the form 
\[ H_\text{quadr}  =  E_{N}^{(0)}  + \sum_{p \in 2\pi \bZ^3 \backslash \{ 0 \} } \sqrt{|p|^4 + 16 \pi \frak{a} p^2} \, a_p^* a_p \, \] 
quadratic in creation and annihilation operators. Here, $E_N^{(0)}$ denotes the approximation to the ground state energy of $H_N$ appearing on the r.h.s. of (\ref{eq:LHY-GP}). The ground state of $H_\text{quadr}$ is the Fock space vacuum $\Omega = \{ 1, 0 , 0, \dots \}$. The main correction to $H_\text{quadr}$ is given by cubic terms in creation and annihilation operators, whose expectation vanishes in the state $\Omega$. By second order perturbation theory, we obtain therefore 
\begin{equation}\label{eq:pert} E_N \simeq E_N^{(0)} + \langle \Omega , W (H_\text{quadr} - E^{(0)}_N)^{-1}  W \Omega \rangle \end{equation} 
where $W$ has the form 
\[ W = \frac{1}{\sqrt{N}} \sum_{p,r \in 2\pi \bZ^3 \backslash \{ 0 \}} \ph_{p,r} \big[ a_{p+r}^* a_{-p}^* a_{-r}^* + \text{h.c.} \big] \]
for some appropriate coefficients $\ph_{p,r}$. The operator $W$ includes all cubic terms which do not vanish when acting on $\Omega$. Using the canonical commutation relations $[a_p, a_q^* ] = \delta_{pq}$, $[a_p, a_q] = 0$, we find 
\[ E_N = E_N^{(0)} + \frac{1}{N} \sum_{p,r}  \frac{\ph_{p,r} \big[ \ph_{p,r}+ \ph_{p+r,r} + \ph_{r,p} + \ph_{r,p+r} + \ph_{p+r,p} + \ph_{p, p+r}\big]}{\eps_p + \eps_r + \eps_{p+r}} \] 
with the dispersion $\eps_p = \sqrt{|p|^4 + 16 \pi \frak{a} p^2}$. Determining the precise form of the coefficients $\ph_{p,r}$ is not trivial (it requires understanding precisely which corrections to the quadratic Hamiltonian $H_\text{quadr}$ are important); it turns out that $\ph_{p,r}$ scales as momentum to the power $-4$, for $|p|, |r| \leq C N$. Taking into account that $\eps_p \simeq p^2$, for large $p$, this produces a correction to $E_N^{(0)}$, exactly of the order $(\log N) / N$. Although this argument could be probably also made into a rigorous proof, our approach is different, since it resolves the correct energy through two unitary conjugations, one with the exponential of a quadratic and, respectively, a cubic expression in (modified) creation and annihilation operators (our unitary conjugations implement the perturbative expansion leading to (\ref{eq:pert})).  
\end{itemize}

{\it Acknowledgments.} We would like to thank the two anonymous referees for carefully reading our paper and suggesting several improvements. We gratefully acknowledge partial support from the Swiss National Science Foundation through the Grant ``Dynamical and energetic properties of Bose-Einstein condensates'', from the NCCR SwissMAP and from the European Research Council through the ERC-AdG CLaQS. C.C. and A.O. acknowledge support from the GNFM Gruppo Nazionale per la Fisica Matematica - INDAM. A.O. acknowledges support from MUR
through the grant “Dipartimento di Eccellenza 2023-2027” of Dipartimento di Matematica, Politecnico di Milano.
\medskip

{ \small {\it Data availability statement.} This manuscript has no associated data.}
\medskip

{\small{\it Conflict of interest. } All authors declare that they have no conflict of interest.}

\section{Excitation Hamiltonians and proof of Theorem \ref{thm:main}} 

In order to determine the low-energy spectrum of the Hamilton operator (\ref{eq:ham0}), it is convenient, first of all, to factor out the Bose-Einstein condensate, focussing on its orthogonal excitations. To this end, we observe that an arbitrary wave function $\psi_N \in L^2_s (\Lambda^N)$ can be uniquely decomposed as  
\[ \psi_N = \alpha_0 \ph_0^{\otimes N} + \alpha_1 \otimes_s \ph_0^{\otimes (N-1)} + \dots + \alpha_N \]
with $\alpha_j \in L^2_\perp (\Lambda)^{\otimes_s j}$, where $L^2_\perp (\Lambda)$ denotes the orthogonal complement of the condensate wave function $\ph_0$, defined by $\ph_0 (x) = 1$ for all $x \in \Lambda$ (and where $\otimes_s$ indicates the symmetric tensor product). This observation allows us to define a unitary operator $U_N : L^2_s (\L^N) \to \cF_\perp^{\leq N}$ mapping the original Hilbert space $L^2_s (\Lambda^N)$ into the truncated Fock space \[ \cF_\perp^{\leq N} = \bigoplus_{n=0}^N L^2_{\perp \ph_0} (\Lambda)^{\otimes_s n}, \] 
 setting $U_N \psi_N = \{ \alpha_0 , \dots , \alpha_N \}$. The map $U_N$ is characterized by its action on number of particle-preserving products of creation and annihilation operators, given by 
 \begin{equation} \label{eq:rules} 
\begin{split}
    U_N a^*_0 a_0 U_N^*=\;&N-\mathcal{N}_+\\
    U_N a^*_p a_0 U_N^*=\;&\sqrt{N} b^*_p\\
    U_N a^*_0 a_p U_N^* =\;&\sqrt{N} b_p\\
    U_N a^*_p a_q U_N^* =\;& a^*_p a_q
\end{split}
\end{equation}
for momenta $p,q \in \Lambda^*_+ = \Lambda^* \backslash \{ 0 \}$ (where $\Lambda^* =  2\pi  \bZ^3$ is the dual lattice to $\Lambda$). Here $a^*_p = a (\ph_p)$ and $a_p = a^* (\ph_p)$ are creation and annihilation operators creating and, respectively, annihilating a particle with momentum $p \in \Lambda^*$, described by the plane wave $\ph_p (x) = e^{-i p \cdot x}$. Furthermore, $\cN_+ = \sum_{p \in \Lambda_+^*} a_p^* a_p$ denotes the number of particles operator on $\cF_{\perp}^{\leq N}$ and, for $p\ \in \Lambda_+^*$, we introduced the modified creation and annihilation operators
\begin{equation*}
    b_p:=\sqrt{\frac{N-\mathcal{N}_+}{N}}a_p, \qquad b^*_p := a^*_p \sqrt{\frac{N-\mathcal{N}_+}{N}}.
\end{equation*}
These operators act on $\cF^{\leq N}_\perp$, they are bounded by the square root of $\cN_+$, in the sense that
\begin{equation*}
\begin{split}
\|b_p\xi\|\le\;&  \,\|\mathcal{N}_+^{1/2}\xi\|\\
\|b^*_p\xi\|\le\;& \,\|(\mathcal{N}_++1)^{1/2}\xi\|.
\end{split}
\end{equation*}
and they satisfy the commutation relations 
\begin{equation} \label{eq:CCR_b}
\begin{split}
[b_p,b^*_q]=\;&\left( 1-\frac{\mathcal{N}_+}{N}\right)\delta_{p,q}-\frac{1}{N}a^*_q a_p\\
[b_p,b_q]=\;&[b_p^*,b_q^*]=0,
\end{split}
\end{equation}
and 
\begin{equation}\label{eq:baa} [b_p,a^*_q a_r]=\delta_{p,q}b_r , \qquad [b^*_p,a^*_q a_r] = - \delta_{p,r} b^*_q. \end{equation} 

Rewriting the Hamilton operator (\ref{eq:ham0}) in momentum space, using the language of second quantization, we find 
\begin{equation}\label{eq:ham1} H_N = \sum_{p \in \L^*} p^2 a_p^* a_p +  \frac{1}{2N} \sum_{p,q,r \in \L^*} \widehat{V} (r/N) a_{p+r}^* a_q^* a_{q+r} a_p.
\end{equation} 
This expression allows us to compute the excitation Hamiltonian $\cL_N = U_N H_N U_N^*$, defined on the excitation space $\cF_\perp^{\leq N}$, using the rules (\ref{eq:rules}). We find 
\begin{equation} \label{eq:cLN} 
\mathcal{L}_N:= U_N H_N U_N^*= \mathcal{L}_N^{(0)} + \mathcal{K}+\mathcal{L}_N^{(2,V)}+\mathcal{L}_N^{(3)}+\mathcal{V}_N,
\end{equation}
where we introduced the kinetic and potential energy operators 
\begin{equation} \label{eq:K,V}
\mathcal{K}= \sum_{p \in \Lambda_+^*} p^2 a^*_p a_p, \qquad \mathcal{V}_N=\frac{1}{2N} \sum_{\substack{r \in \Lambda^*, p,q \in \Lambda^*_+\\ r\ne -p,-q}} \widehat V(r/N) a^*_{p+r} a^*_q a_p a_{q+r}
\end{equation}
and we set 
\begin{equation} \label{eq:def_L}
\begin{split}
\mathcal{L}_N^{(0)}=\;&\frac{\widehat V(0)}{2}(N-1)\\
\mathcal{L}_N^{(2,V)}=\;&\sum_{p \in\Lambda_+^*} \widehat V(p/N)\left( b^*_p b_p-\frac{1}{N}a^*_p a_p \right) \\
&+\frac{1}{2}\sum_{p\in \Lambda^*_+} \widehat V(p/N) (b^*_p b^*_{-p}+b_pb_{-p})- \frac{\widehat V(0)}{2N}\mathcal{N}_+(\mathcal{N}_+-1) \\
\mathcal{L}_N^{(3)}=\;&\frac{1}{\sqrt{N}} \sum_{\substack{p,q \in\Lambda_+^*\\p+q\ne0}}\widehat V(p/N)\left( b^*_{p+q}a^*_{-p}a_q +a^*_q a_{-p} b_{p+q}\right) .
\end{split}
\end{equation}

After conjugation with $U_N$, the vacuum vector $\Omega \in \cF^{\leq N}_\perp$ corresponds to the 
factorized wave function $\ph_0^{\otimes N}$, which is still very far, energetically, from the ground state 
of (\ref{eq:ham1}). In the next step, we are going to renormalize the excitation Hamiltonian (\ref{eq:cLN}), factoring out the microscopic correlation structure characterizing its low-energy states. To describe correlations, we fix $\ell > 0$ and consider the ground state solution of the Neumann problem 
\begin{equation*} \label{eq:Neumann_unscaled}
\begin{split} 
\Big[ -\Delta+\frac{1}{2}V\Big] f_\ell &=  \lambda_\ell f_\ell(x) \qquad |x| \le N\ell \\
\partial_r f_\ell(x) &=0 \qquad  \hspace{1cm} |x|=N\ell 
\end{split} 
\end{equation*}
on the ball $|x| \leq N \ell$, normalized so that $f_\ell (x) = 1$ for $|x| = N\ell$. We extend $f_\ell (x) = 1$ for all $|x| > N\ell$ and we set $w_\ell (x) = 1 - f_\ell (x)$. We denote $f_{N,\ell} (x) = f_\ell (Nx)$ the solution of the rescaled Neumann problem 
\begin{equation} \label{eq:Neumann_scaled}
\begin{split}
\Big[ -\Delta+\frac{N^2}{2}V (N\cdot) \Big] f_{N,\ell} &=  \lambda_\ell f_{N,\ell} \qquad  |x| \le \ell\\
\partial_r f_{N,\ell} (x) &= 0 \qquad \hspace{.8cm} |x|= \ell 
\end{split}
\end{equation}
on the ball $|x| \leq \ell$, with $f_{N,\ell} (x) = 1$ for all $|x| \geq \ell$. As above, we set $w_{N,\ell} (x) = 1- f_{N,\ell} (x)$. With a slight abuse of notation we use the same notation for $f_{N,\ell}$ and for its periodisation on the torus $\Lambda$. The next Lemma, whose proof can be found in \cite[Appendix B]{BBCS}, collects important bounds for the functions $f_\ell$ and $w_\ell$, and for the eigenvalue $\lambda_\ell$.

\begin{lemma} \label{lemma:scattering}
Let $V$ be as in the assumptions of Theorem \ref{thm:main}.
\begin{itemize}
\item[(i)] The eigenvalue $\lambda_\ell$ appearing in \eqref{eq:Neumann_scaled} satisfies
\begin{equation*}
\lambda_\ell=\frac{3\mathfrak{a}}{(N\ell)^3}\left[ 1+\frac{9}{5}\frac{\mathfrak{a}}{ N\ell }+\mathcal{O}\left(\frac{\mathfrak{a}^2}{(N\ell)^2}\right) \right].
\end{equation*}
\item[(ii)] There exists a constant $C>0$ such that
\begin{equation}\label{eq:Vf-8pia}
\left| \int_{\bR^3} V(x) f_\ell(x) dx-8\pi\mathfrak{a}\left(1+\frac{3}{2}\frac{\mathfrak{a}}{N \ell}\right)\right|\le \frac{C\mathfrak{a}^3}{(N\ell)^2}
\end{equation}
for $\ell\in(0,1/2)$.
\item[(iii)] There exists $C> 0$ with 
\begin{equation} \label{eq:decay_w}
w_\ell (x) \leq \frac{C}{|x|+1}\,, \qquad |\nabla w_\ell (x)| \leq \frac{C}{x^2 + 1}\,. \end{equation}
\end{itemize}
\end{lemma}

The Fourier coefficients of $f_{N,\ell}$ are given by 
\[ \wh{f}_{N,\ell} (p) = \int_\Lambda f_{N,\ell} (x) e^{-ip\cdot x} dx = \delta_{p,0} - N^{-3} \wh{w}_\ell (p/N) \]
where we defined
\[ \wh{w}_\ell (q) = \int_{\bR^3} w_{N,\ell} (x) e^{-i q \cdot x} dx \]
for all $q \in \bR^3$. For $p \in \Lambda^* = 2\pi \bZ^3$, we consider the coefficients 
\begin{equation}\label{eq:defeta} \eta_p = - N^{-2} \wh{w}_\ell (p/N) \end{equation}  
By (\ref{eq:Neumann_scaled}), they satisfy the equation 
\begin{equation}
\label{scattering}
p^2 \eta_p + \frac{1}{2} \widehat{V}(p/N) +\frac{1}{2N}\sum_{q\in\Lambda^*}\widehat{V} ((p-q)/N) \eta_q = N^3 \lambda_\ell \widehat{\chi}_\ell (p) + N^2 \lambda_\ell \sum_{q \in \Lambda^*}\widehat{\chi}_\ell (p-q) \eta_q \, .
\end{equation}
Here $\chi_\ell$ is the characteristic function of the ball of radius $\ell$, centered at the origin. Through the coefficients $\eta_p$ we define the antisymmetric operator 
 \begin{equation*}\label{eq:Beta} B_\eta = \frac{1}{2} \sum_{p \in \L^*_+} \eta_p \big( b_p^* b_{-p}^* - b_p b_{-p} \big).
 \end{equation*} 
Conjugating (\ref{eq:cLN}) with the generalized Bogoliubov transformation $e^{B_\eta}$, we define the renormalized excitation Hamiltonian $\wt{\cG}_N = e^{-B_\eta} \cL_N e^{B_\eta}$. As shown in \cite{BBCS}, $\wt{\cG}_N$ has the form 
\[ \wt{\cG}_N \simeq C_{\wt{\cG}_N} + \cQ_{\wt{\cG}_N} + \cC_{\wt{\cG}_N} + \cV_N \]
up to small corrections. Here $C_{\wt{\cG}_N}$ is a constant, while $\cQ_{\wt{\cG_N}}, \cC_{\wt{\cG}_N}$ are quadratic and, respectively, cubic contributions in creation and annihilation operators. As discussed in \cite{BBCS}, this form of the excitation Hamiltonian is still not enough to determine its spectrum (not even up to errors of order one, in $N$), because the cubic term $\cC_{\wt{\cG}_N}$ is not negligible. A second renormalization, this time with a unitary transformation given by the exponential of a cubic expression in creation and annihilation operators, must be used to get rid of $\cC_{\wt{\cG}_N}$. The resulting twice renormalized excitation Hamiltonian has the form 
\begin{equation} \label{eq:wtJN} \wt{\cJ}_N \simeq C_{\wt{\cJ}_N} + \cQ_{\wt{\cJ}_N} + \cV_N \end{equation}
again up to small corrections. At this point the quadratic part of $\wt{\cJ}_N$ has the form 
\begin{equation}\label{eq:wtQN} \cQ_{\wt{\cJ}_N} = \sum_{p \in \L^*_+} \Big[ \widetilde{F}_p b_p^* b_p + \frac{1}{2} \widetilde{G}_p \big(b^*_p b^*_{-p} + b_p b_{-p} \big) \Big] \end{equation} 
with the coefficients 
\begin{equation}\label{eq:FpGp} \begin{split} 
\widetilde{F}_p &
= \Big[ p^2 +  \big( \widehat{V} (\cdot /N) * \widehat{f}_{N,\ell} \big)_p \Big] \cosh (2\eta_p) +  \big( \widehat{V} (\cdot /N) * \widehat{f}_{N,\ell} \big)_p \sinh (2\eta_p) \\
 \widetilde{G}_p &
 = \Big[ p^2 +  \big( \widehat{V} (\cdot /N) * \widehat{f}_{N,\ell} \big)_p \Big] \sinh (2\eta_p) +  \big( \widehat{V} (\cdot /N) * \widehat{f}_{N,\ell} \big)_p \cosh (2\eta_p). \end{split} \end{equation} 
To compute the spectrum of (\ref{eq:wtJN}), it is convenient to diagonalize (\ref{eq:wtQN}), conjugating it with another generalized Bogoliubov transformation.
As shown in \cite[Lemma 5.1]{BBCS}, the coefficients (\ref{eq:FpGp}) satisfy the bounds
\begin{equation}\label{eq:wtFG}
    \frac{p^2}{2}\le \widetilde{F}_p \le C(1+p^2),\qquad|\widetilde{G}_p|\le \frac{C}{p^2},\qquad|\widetilde{G}_p|\le \widetilde{F}_p
\end{equation}
for all $p\in\Lambda_+^*$. As a consequence, we can define coefficients $\tau_p$ requiring that 
\begin{equation}\label{eq:deftau}
	\begin{split} \tanh (2\tau_p) &= - \frac{\widetilde{G}_p}{\widetilde{F}_p} = - \frac{\big[ p^2 +  \big( \widehat{V} (\cdot /N) * \widehat{f}_{N,\ell} \big)_p \big] \sinh (2\eta_p) +  \big( \widehat{V} (\cdot /N) * \widehat{f}_{N,\ell} \big)_p \cosh (2\eta_p)}{\big[ p^2 +  \big( \widehat{V} (\cdot /N) * \widehat{f}_{N,\ell} \big)_p \big] \cosh (2\eta_p) +  \big( \widehat{V} (\cdot /N) * \widehat{f}_{N,\ell} \big)_p \sinh (2\eta_p)}.
\end{split}
\end{equation} 
With this choice of the coefficients $\tau_p$, it was proven in \cite{BBCS} that \[ \wt{\cM}_N  = e^{-B_\tau} \wt{\cJ}_N e^{B_\tau} \simeq C_{\wt{\cM}_N} + \sum_{p\in \L^*_+} \eps (p) a_p^* a_p + \cV_N \] for appropriate constant $C_{\cM_N}$ and dispersion $\eps (p)$. It is then easy to determine the low-energy spectrum of $\wt{\cM}_N$ (using the positivity of the potential energy operator $\cV_N$, and its smallness on states with few low-momentum excitations); see \cite{BBCS} for the details. 

In the present work, to improve the energy resolution up to errors smaller than $(\log N)/N$, we find it more convenient to combine $e^{B_\eta}$ and $e^{B_\tau}$ into a single generalized Bogoliubov transformation. To this end, we define the coefficients 
\begin{equation}\label{eq:defmu} \mu_p = \eta_p + \tau_p \end{equation} 
for all $p \in \L^* = 2\pi \bZ^3$, with $\eta_p$ as in (\ref{eq:defeta}) and $\tau_p$ as introduced in (\ref{eq:deftau}). In the next lemma, we collect important properties of the coefficients $\eta_p$, $\tau_p$, $\mu_p$, and
\begin{equation} \label{eq:def_gs}
    \gamma_p=\cosh \mu_p,\quad \sigma_p=\sinh\mu_p.
\end{equation}
We will systematically use such properties in estimates throughout the paper.

\begin{lemma} \label{lemma:coefficients}
Let $V$ be as in the assumptions of Theorem \ref{thm:main}.
\begin{itemize}
\item[(i)] There exists a constant $C>0$ such that
\begin{equation*} \label{eq:decay_eta}
|\eta_p|\le C|p|^{-2}.
\end{equation*}
Moreover, 
\begin{equation} \label{eq:norms_eta}
|\eta_0| \leq C, \qquad \|\eta\|_2 \le C, \qquad \sum_{p \in \L^*} p^2 |\eta_p|^2 \leq C N, \qquad \| \eta \|_q \leq C \max \{ 1 , N^{3-2q} \}   
\end{equation}
for all $q \geq 1$.
\item[(ii)] There exists a constant $C>0$ such that
\begin{equation}\label{eq:tau-bds}
|\tau_p| \leq C |p|^{-4}, \qquad  \sum_{p \in \L^*_+} \tau_p^2 + \sum_{p \in \L^*_+} p^2 \tau_p^2 < C.
\end{equation} 
\item[(iii)] As a consequence of $(i)$ and $(ii)$ and of the definitions of $\mu_p,\gamma_p,\sigma_p$, we have
\begin{equation} \label{eq:decay_mu,g,s}
    |\mu_p|+|\sigma_p|\le C |p|^{-2},\qquad |\gamma_p|\le C
\end{equation}
and
\begin{equation} \label{eq:norms_mu,g,s}
    \|\mu\|_2+\|\sigma\|_2\le C,\qquad \sum_{p\in\Lambda_+^*}p^2\big(\mu_p^2+ \sigma_p^2)\le C N
\end{equation}
for a suitable $C>0$.
\end{itemize}
\end{lemma}

\begin{proof}
We only prove the last bound in (\ref{eq:norms_eta}) since the other estimates in $(i)$ are shown in \cite[Appendix B]{BBCS}, the estimates in $(ii)$ follow from the definition of $\tau_p$ together with \eqref{scattering} and \eqref{eq:wtFG}, and $(iii)$ is an immediate consequence of $(i)$ and $(ii)$. We have 
\[ \begin{split}  \| \eta \|_q^q &= \sum_{p \in \Lambda^*_+} | \eta_p |^q = \sum_{p : |p| \leq N} |p|^{-2q} + \sum_{p : |p| > N} |\eta_p|^q\\ & \leq C \max \{ N^{3-2q} , 1 \} +  \sum_{p : |p| > N}  \frac{\big| \big( \widehat{V} (\cdot / N) * \widehat{f}_{N,\ell}\big)_p \big|^q}{|p|^{2q}} + C \sum_{p : |p| > N}  \frac{\big| \big( \widehat{\chi}_\ell * \widehat{f}_{N,\ell} \big)_p \big|^q}{|p|^{2q}}
 \end{split} \]
With H\"older's inequality, we find 
\[ \begin{split} \sum_{p : |p| > N}  \frac{\big| \big( \widehat{V} (\cdot / N) * \widehat{f}_{N,\ell}\big)_p \big|^q}{|p|^{2q}} &\leq \big\| \widehat{V} (\cdot / N) * \widehat{f}_{N,\ell} \big\|_2^q  \,  \Big\| \frac{\chi (|.| > N)}{|.|^{2q}} \Big\|_{2/(2-q)} \\ &\leq C N^{3q/2} N^{-(7q-6)/2} \leq C N^{3-2q}. \end{split} \] 
\end{proof} 

Using the coefficients $\mu_p$ we define the antisymmetric operator
\begin{equation}\label{eq:Bmu} B_\mu = B_\eta + B_\tau = \frac{1}{2} \sum_{p \in \L^*_+} \mu_p (b_p^* b_{-p}^* - \text{h.c.} ). \end{equation} 
With the corresponding generalized Bogoliubov transformation $e^{B_\mu}$, we define the renormalized excitation Hamiltonian 
\begin{equation}\label{eq:cGN-def} \cG_N = e^{-B_\mu} \cL_N e^{B_\mu}. \end{equation} 
The advantage that we have, when working with $\cG_N$ rather than with $\cL_N$, is that, after removing the microscopic correlation structure through $e^{-B_\mu}$, low-energy states of $\cG_N$ have only few excitations, with bounded energy. More precisely, we obtain the following a-priori estimates on products of the energy operator $\cH_N = \cK+ \cV_N$ with arbitrary powers of the number of particles operator $\cN_+$.  
\begin{prop} \label{prop:apri} 
Let $\psi_N \in L^2_s (\Lambda^N)$ be a normalized sequence of approximate ground state vectors of the Hamilton operator (\ref{eq:ham0}), satisfying $\psi_N = \chi (H_N \leq E_N +K) \psi_N$, where $E_N$ is the ground state energy of (\ref{eq:ham0}) and $K > 0$ is fixed. Let $\xi_N = e^{-B_\mu} U_N \psi_N$ be the corresponding normalized sequence of approximate ground state vectors of the renormalized excitation Hamiltonian (\ref{eq:cGN-def}). Let $k \in \bN$. Then, there exists $C > 0$ (depending on $K$ and $k$) such that 
\begin{equation}\label{eq:apri-xi} \langle \xi_N, (\cH_N + 1) (\cN_+ + 1)^k \xi_N \rangle \leq C \end{equation} 
for all $N \in \bN$.
\end{prop}  
A-priori bounds of the form (\ref{eq:apri-xi}) have been established in \cite[Prop. 4.1]{BBCS}, for a sequence $\xi'_N = e^{-B_\eta} U_N \psi_N$, defined in terms of the generalized Bogoliubov transformation generated by $B_\eta$, rather than in terms of that generated by $B_\mu$. From (\ref{eq:Bmu}), the difference $B_\mu - B_\eta = B_\tau$ is associated with the kernel $\tau$ exhibiting, by (\ref{eq:tau-bds}), fast decay in momentum space. For this reason, the estimate in Prop. \ref{prop:apri} follows from the bounds for $\xi'_N$ in \cite[Prop. 4.1]{BBCS}. This is shown in Appendix \ref{app:apri}.

In the next theorem, whose proof is deferred to Section \ref{sec:cGN}, we determine the operator $\cG_N$, up to very small errors (which can be controlled through the a-priori estimates in Prop. \ref{prop:apri}).  
\begin{theorem} \label{thm:quadratic} 
Let $V\in L^3(\mathbb{R}^3)$ be non-negative, compactly supported, and spherically symmetric. Let $\mathcal{G}_N$ be defined as in \eqref{eq:cGN-def} with parameter $\ell\in(0,\frac{1}{2})$ small enough. For $p \in \Lambda^*_+$, let $\gamma_p = \cosh \mu_p$ and $\sigma_p = \sinh \mu_p$. Moreover, define the constant 
\begin{equation}
    \label{eq:C_G_N}
    \begin{split}
  C_{\mathcal{G}_N}  
  = \; & \frac{(N-1)}{2}\widehat V(0)  + \sum_{p \in \L^*_+} \Big[ p^2\s_p^2  + \widehat V (p/N) (\s_p\g_p +\s_p^2)\Big] \\
 & + \frac{1}{2N} \sum_{p \in \L^*_+} \widehat V ((p-q)/N) \s_q\g_q\s_p\g_p + \frac1N \sum_{p \in \L^*_+}\Big[ p^2\eta_p^2 + \frac1{2N}\big( \widehat V (\cdot / N) *\eta\big)_p \eta_p\Big] \\
 & - \frac1N \sum_{p \in \L^*,  q\in \L^*_+} \widehat V (p/N) \eta_p\, \s_q^2 
 \end{split}
\end{equation}
and the cubic operator
\begin{equation} \label{eq:cC_G_N}
    \begin{split}
        \cC_{\mathcal{G}_N} = \frac{1}{\sqrt{N}}\sum_{\substack{p,q\in \Lambda_+^*\\p+q\ne0}}\widehat{V} (p/N) \,b^*_{p+q}b^*_{-p}(\gamma_q b_q+\sigma_q b^*_{-q})+\mathrm{h.c.}.
    \end{split}
\end{equation}
Furthermore, let $\cT_{\cG_N} = \cT^{(2)}_{\cG_N} + \cT^{(4)}_{\cG_N}$, with 
\begin{equation}\label{eq:cT_G_N} \begin{split} 
\cT^{(2)}_{\cG_N} = \; & \frac{1}{2N^2} \sum_{p\in\Lambda_+^*} \big(\widehat{V} (\cdot/ N)* \eta \big)_p  ( b_p b_{-p}+b^*_p  b^*_{-p}) \big(1+2 \| \s \|_2^2) \\
        &+\frac{1}{2N} \sum_{q \in \L^*_+} \Big(2\sigma_q^2+\frac{\gamma_q\sigma_q}{\mu_q}-1\Big) \sum_{p \in \Lambda_+^*}  p^2 \eta_p (b_pb_{-p}+b^*_pb^*_{-p})\\
               \mathcal{T}_{\mathcal{G}_N}^{(4)}=\;&\frac{1}{2N}\sum_{p, q \in \L^*_+, r \in \L^*}\widehat{V}(r/N) \sigma_{p}\sigma_{q+r}b^*_{p+r}b^*_{q}b^*_{-p}b^*_{-q-r} + \hc\\
 & + \frac{1}{N^2} \sum_{p,q\in \Lambda^*_+} \big(\widehat{V}(\cdot / N)*\eta\big)_p \gamma_q \s_q b^*_pb^*_{-p} b^*_q b^*_{-q} +\mathrm{h.c.} \\
 &- \frac{1}{2N} \sum_{p \in \L^*, q \in \L^*_+} \big(\widehat V (\cdot / N) * \widehat{f}_{N,\ell} \big)_p \Big(\gamma_q\sigma_q- \frac{\eta_q}{2} - \frac{\s_q^2}{2\mu_q}\Big) b^*_pb^*_{-p} b^*_qb^*_{-q} +\mathrm{h.c.}. \end{split} \end{equation} 
Then
\begin{equation} \label{eq:expansion_G_N}
 \begin{split}   \mathcal{G}_N= \; &C_{\mathcal{G}_N} + \sum_{p \in \L_+^*} \sqrt{|p|^4 + 2 (\widehat{V} (\cdot /N) * \widehat{f}_{N,\ell} )_p p^2} \; a_p^* a_p \\ &- \frac{1}{N} \sum_{p,q \in \L^*_+} \big( \widehat{V} (q/N) + \widehat{V} ((q+p)/N) \big) \eta_q \big[ (\g_p^2 + \s_p^2) b_p^* b_p + \g_p \s_p (b_p^* b_{-p}^* + b_p b_{-p} ) \big]  \\ &+\mathcal{C}_{\mathcal{G}_N} + \cV_N + \mathcal{T}_{\mathcal{G}_N} + \mathcal{E}_{\mathcal{G}_N}
\end{split} \end{equation}
where, for every $0< \eps < 1$, we have 
\begin{equation}\label{eq:cE-bds} 
    \pm \mathcal{E}_{\mathcal{G}_N} \le  \eps \cN_+ + \frac{C}{\eps N} (\cH_N + \cN_+^3 + 1) (\cN_+ + 1)  
\end{equation}
with $\cH_N = \cK + \cV_N$. 
\end{theorem}

{\it Remark.}  In the representation \eqref{eq:expansion_G_N} of the renormalized excitation Hamiltonian, we distinguish three types of error terms (terms which will not contribute to the energy of the Hamiltonian, up to order $(\log N) / N$). First of all, in $\cE_{\cG_N}$ we absorb several contributions that are controlled by the second term on the r.h.s. of (\ref{eq:cE-bds}). With the a-priori bounds in Prop. \ref{prop:apri}, these terms are small, of order $N^{-1}$, in the limit $N \to \infty$. Other contributions to the error $\cE_{\cG_N}$ are bounded by $\eps \cN_+$, for an arbitrary small $\eps > 0$. Since the cubic conjugation only increases $\cN_+$ by $\cO (N^{-1})$, we will control these terms using a little bit of kinetic energy. Terms in $\cT_{\cG_N}$, on the other hand, could only be controlled, at this point, by \[ | \langle \xi, \cT_{\cG_N} \xi \rangle | \leq C N^{-1/2} \| \cK^{1/2} \xi \| \| (\cN_+ + 1) \xi \|  \] 
or by
\[ | \langle \xi, \cT_{\cG_N} \xi \rangle | \leq C N^{-1/2} \| \cK^{1/2} \cN_+^{1/2} \xi \| \| (\cN_+ + 1)^{1/2} \xi \| \,. \] 
Notice, in the two bounds, the presence of the operator $(\cN_+ + 1)$, instead of just $\cN_+$, due to the fact that contributions in $\cT_{\cG_N}$ only contain creation or annihilation operators, but never both. Since moreover the cubic conjugation changes the expectation of $\cK$ by order one, this estimate does not yet allow us conclude that $\cT_{\cG_N}$ is negligible (instead, we first have to apply the cubic conjugation to $\cT_{\cG_N}$ and only afterwards we will be able to show that it can be dropped).  

\medskip

To get rid of the cubic term $\cC_{\cG_N}$ in \eqref{eq:expansion_G_N}, we conjugate $\cG_N$ with a second unitary transformation, given by the exponential of a cubic expression in (modified) creation and annihilation operators. We define  
\begin{equation} \label{eq:def_A}
A=\frac{1}{\sqrt{N}}\sum_{\substack{r,v \in \Lambda_+^*\\r+v\neq 0}}b^*_{r+v}b^*_{-r}(\eta_r\gamma_vb_v+\nu_{r,v}b^*_{-v})-\mathrm{h.c.}=A_{\gamma}+A_{\nu},
\end{equation}
where 
\begin{equation}\label{eq:def_wtsigma}
\nu_{r,v}=  \frac{2r^2\eta_r \sigma_v}{|r+v|^2+r^2+v^2},
\end{equation}
and $\eta$, $\gamma$, and $\sigma$ were defined in \eqref{eq:defeta} and in \eqref{eq:def_gs} (recall from Lemma \ref{lemma:coefficients} that $|\eta_r| , |\sigma_r| \leq C |r|^{-2}$, while $|\gamma_r| \leq C$, for an appropriate constant $C> 0$). The choice of the coefficients in $A$ guarantees that the commutator $[\cH_N , A]$ produces a contribution cancelling the cubic term $\cC_{\cG_N}$ in (\ref{eq:cC_G_N}). Compared with the cubic phase used in \cite{BBCS} and in later works, where the coefficient in front of the operator $b_{r+v}^* b_{-r}^* b_{-v}^*$ was simply given by $\eta_r \sigma_v$, we modify here the choice of $\nu_{r,v}$ to eliminate certain terms arising from the commutator $[\cK, A]$ which would not be negligible at the level of accuracy required to show Theorem \ref{thm:main}; see the remark after Lemma \ref{lemma:K,V,A} (notice that in \cite{BBCS}, the operator $A$ was defined summing only over momenta $r,v$ with $|r| \gg |v|$; in this region, $\nu_{r,v} \simeq \eta_r \sigma_v$).

For our analysis, it is very important to control the growth of number and energy of excitations, w.r.t. conjugation by $e^A$. 
\begin{prop} \label{prop:a_priori_A}
Let $A$ be defined as in \eqref{eq:def_A}, $\mathcal{N}_+$ be the number of particles operator on $\cF_\perp^{\leq N}$ and $\mathcal{H}_N = \cK + \cV_N$ as be defined as in Theorem \ref{thm:quadratic}. Then, for any $k \in \bN$, for any  $s \in [0;1]$, there exists $C > 0$ (depending on $k$) such that 
\begin{equation}
    \label{eq:growthN}
    e^{-sA}(\cN_++1)^k e^{sA} \leq C (\cN_++1)^k\,.
\end{equation}
Furthermore, there is $C > 0$ such that  
\begin{equation}
    \label{eq:growthNimproved}
    \langle \xi, e^{-sA}\cN_+ e^{sA} \xi\rangle \leq \langle \xi, \cN_+\xi\rangle  + \frac{C}N \langle \xi, (\cN_++1)^2\xi\rangle\,.
\end{equation}
Moreover, for every $k \in \bN$ and every $s \in [0;1]$, there exists $C > 0$ such that  
\begin{equation}\label{eq:growthHN}
\begin{split}
    e^{-sA}(\mathcal{H}_N+1) (\mathcal{N}_++1)^k e^{sA}\le C(\mathcal{H}_N+1)  (\mathcal{N}_++1)^k + C(\cN_++1)^{k+2} \, .
\end{split}
\end{equation} 
\end{prop}
The proof of Proposition \ref{prop:a_priori_A} will be given in Section \ref{sec:gron-A} (there, the choice of the coefficients (\ref{eq:def_wtsigma}) will become clear). 

With $A$ defined as in (\ref{eq:def_A}), we introduce the cubically renormalized excitation Hamiltonian 
\begin{equation} \label{eq:def_mathcal_J}
    \mathcal{J}_{N}=e^{-A} \mathcal{G}_N e^{A}.
\end{equation}
In the next theorem, we describe the operator $\cJ_N$. 
\begin{theorem}\label{thm:cubic} 
Let $V\in L^3(\mathbb{R}^3)$ be non-negative, compactly supported, and spherically symmetric. Let $\mathcal{J}_N$ be defined in \eqref{eq:def_mathcal_J}. Let 
\begin{equation} \label{eq:def_C_J}
    \begin{split}
        C_{\mathcal{J}_N} = & \; 4\pi \frak{a} (N-1) + e_\Lambda \frak{a}^2 - \frac{1}{2} \sum_{p \in \Lambda^*_+} \Big[ p^2 + 8\pi \frak{a} - \sqrt{|p|^4 + 16 \pi \frak{a} p^2} - \frac{(8\pi \frak{a})^2}{2p^2} \Big] \\
        &+ \frac{1}{N}\sum_{\substack{p,q\in\Lambda_+^*\\p+q\ne0}}\bigg[\big(\widehat{V}(\cdot/N)\ast\widehat f_{N,\ell}\big)_p +\big(\widehat{V}(\cdot/N)\ast\widehat f_{N,\ell}\big)_{p+q}\bigg]\\
            &\hspace{4cm}\times\eta_p\,\eta_q\frac{2\eta_{q+p}(q+p)^2-2\eta_q(p\cdot q)}{p^2+q^2+(p+q)^2}\\
    \end{split}
\end{equation}
with $e_\Lambda$ defined in (\ref{eq:eLambda}).
Then we have 
\begin{equation} \label{eq:cJN-ex}
\begin{split} 
 \mathcal{J}_N= \; &C_{\mathcal{J}_N} + \sum_{p \in \L_+^*} \sqrt{|p|^4 + 16 \pi \frak{a}  p^2} \; a_p^* a_p  + \cV_N + \mathcal{E}_{\mathcal{J}_N}
\end{split} \end{equation}
where, for any $\eps > 0$ ($\eps$ can also depend on $N$, provided $\eps > C (\log N)/N$), 
\[ \pm \cE_{\cJ_N} \leq \eps \cK + \frac{C}{N} \Big[ (\log N)^{1/2} + \eps^{-1} \Big]  (\cH_N+1) ( \cN_+ + 1)^4 \] 
 \end{theorem}
The proof of Theorem \ref{thm:cubic} will be given below, in Section \ref{sec:cubic}. 

We can now apply Theorem~\ref{thm:cubic}, to show our main result. 

\begin{proof}[Proof of Theorem \ref{thm:main}]  
We claim that the ground state energy $E_N$ of  (\ref{eq:ham0}) is such that 
\begin{equation*}
    | E_N-C_{\cJ_N} | \le \frac{C \, (\log N)^{1/2}}{N}.
\end{equation*}

{\it Upper bound}. From (\ref{eq:cJN-ex}), taking $\eps>0$ a constant, we obtain
\[ \cJ_N \leq C_{\cJ_N} + C  \cK  + \cV_N + \frac{C (\log N)^{1/2}}{N} (\cH_N+ 1) (\cN_+ + 1)^4. \]
We conclude, taking expectation in the vacuum, that 
\[ E_N \leq \langle \Omega, \cJ_N \Omega \rangle \leq C_{\cJ_N}+ \frac{C (\log N)^{1/2}}{N}. \]

{\it Lower bound.} From (\ref{eq:cJN-ex}), taking $\eps = (\log N)^{-1/2}$, and using the positivity of $\cV_N$, we find 
\[ \cJ_N \geq C_{\cJ_N} + \Big[ 1- \frac{1}{(\log N)^{1/2}} \Big] \cK - \frac{C (\log N)^{1/2}}{N} (\cH_N + 1 ) ( \cN_+ + 1)^4. \]
Let $\theta_N \in \cF^{\leq N}_\perp$ denote a normalized ground state vector for $\cJ_N$. Then $\xi_N = e^{A} \theta_N$ is a normalized ground state vector of the excitation Hamiltonian $\cG_N$ defined in (\ref{eq:cGN-def}). Combining Prop. \ref{prop:apri} with Prop. \ref{prop:a_priori_A}, we conclude that 
\[ \langle \theta_N, (\cH_N + 1) (\cN_+ + 1)^4 \theta_N \rangle \leq C \langle \xi_N, (\cH_N + 1) (\cN_+ + 1)^5 \xi_N \rangle \leq C \] 
and therefore that 
\[ E_N = \langle \theta_N, \cJ_N \theta_N \rangle \geq C_{\cJ_N} - \frac{C (\log N)^{1/2}}{N} .\]

To conclude the proof of Theorem \ref{thm:main}, we still need to evaluate the constant $C_{\cJ_N}$. To this end, we write (\ref{eq:def_C_J}) as 
\[ C_{\cJ_N}  = 4\pi \frak{a} (N-1) + e_\Lambda \frak{a}^2 - \frac{1}{2} \sum_{p \in \Lambda^*_+} \Big[ p^2 + 8\pi \frak{a} - \sqrt{|p|^4 + 16 \pi \frak{a} p^2} - \frac{(8\pi \frak{a})^2}{2p^2} \Big] + \wt{C}_{\cJ_N} \]
with 
\begin{equation}\label{eq:wtCJN}  
\begin{split} 
\wt{C}_{\cJ_N} =  \; &= \frac{1}{N}\sum_{\substack{r,v\in\Lambda_+^*\\r+v\ne0}}\bigg[\big(\widehat{V}(\cdot/N)\ast\widehat f_{N,\ell}\big)_r +\big(\widehat{V}(\cdot/N)\ast\widehat f_{N,\ell}\big)_{r+v}\bigg]\\
            &\hspace{4cm}\times\eta_r\,\eta_v\frac{2\eta_{r+v}(r+v)^2-2\eta_v(r\cdot v)}{r^2+v^2+(r+v)^2}\\
    \end{split}
\end{equation} 
We first show that \begin{equation*}\label{eq:final-claim} \wt{C}_{\cJ_N} = \frac{1024 \pi^4 \frak{a}^4}{N} \sum_{r,v \in \Lambda^*_+ : |r|, |v| \leq N} \frac{r\cdot v - v^2}{r^2 + v^2 + r \cdot v} \frac{1}{r^2 v^4}  + \cO (N^{-1}). \end{equation*} 
To reach this goal, we apply the scattering equation \eqref{scattering} to the second line of (\ref{eq:wtCJN}). Noticing that the contribution arising from the r.h.s. of \eqref{scattering} is negligible (the r.h.s decays faster, it makes the term of order $N^{-1}$), and combining with symmetry we arrive at 
\begin{equation}
\label{CS_3}
    \begin{split}
        &\widetilde{C}_{\cJ_N} \\ & = \frac{1}{8N}\sum_{\substack{r,v\in\Lambda_+^*\\r+v\ne0}} \frac{(r\cdot v)-v^2}{r^2+v^2+r\cdot v}\frac{1}{r^2v^4}\big( \widehat{V} (\cdot / N)*\widehat{f}_{N,\ell}\big)_v^2 \big( \widehat{V} (\cdot / N)*\widehat{f}_{N,\ell}\big)_r \big( \widehat{V} (\cdot / N)*\widehat{f}_{N,\ell}\big)_{r+v}\\
        & \; \; -\frac{1}{8N} \sum_{\substack{r,v\in\Lambda_+^*\\r+v\ne0}}\frac{1}{r^2+v^2+r\cdot v}\frac{1}{r^2v^2}\big( \widehat{V} (\cdot / N)*\widehat{f}_{N,\ell}\big)_{r+v}^2 \big( \widehat{V} (\cdot / N)*\widehat{f}_{N,\ell}\big)_r \big( \widehat{V} (\cdot / N)*\widehat{f}_{N,\ell}\big)_v\\& \; \;+\frac{1}{8N} \sum_{\substack{r,v\in\Lambda_+^*\\r+v\ne0}}\frac{(r\cdot v)}{r^2+v^2+r\cdot v}\frac{1}{r^2v^4} \big( \widehat{V} (\cdot / N)*\widehat{f}_{N,\ell}\big)_r^2 \big( \widehat{V} (\cdot / N)*\widehat{f}_{N,\ell}\big)_v^2+\cO (N^{-1}) .
    \end{split}
\end{equation}
In the next step, we restrict all sums to $|r|, |v| \leq N$. For the second term on the r.h.s. of the last equation, it is easy to check that the corresponding error is negligible, of order $N^{-1}$. In fact,
\begin{equation} \label{eq:CS31} \begin{split} \Big| \frac{1}{N} &\sum_{r,v : |r| > N} \frac{1}{r^2+v^2+r\cdot v}\frac{1}{r^2v^2}\big( \widehat{V} (\cdot / N)*\widehat{f}_{N,\ell}\big)_{r+v}^2 \big( \widehat{V} (\cdot / N)*\widehat{f}_{N,\ell}\big)_r \big( \widehat{V} (\cdot / N)*\widehat{f}_{N,\ell}\big)_v \Big| \\  \leq \; & 
\frac{C}{N} \sum_{r,v : |r| > N}  \frac{1}{(r^2 + v^2) v^2 r^2} \big| \big( \widehat{V} (\cdot / N)*\widehat{f}_{N,\ell}\big)_r \big| \leq \frac{C}{N} \sum_{|r| > N} \frac{1}{|r|^3} \big| \big( \widehat{V} (\cdot / N)*\widehat{f}_{N,\ell}\big)_r \big| \leq  \frac{C}{N} 
\end{split} \end{equation} 
using the bound $\| \widehat{V} (\cdot / N) * \widehat{f}_{N,\ell} \|_2 = \| N^3 V (N\cdot) f_{N,\ell} \|_2 \leq C N^{3/2}$. 
Similarly, one can also bound the contribution to this term arising from the region $|v| > N$. Let us now consider the last term on the r.h.s. of (\ref{CS_3}). Observing that 
\[ \begin{split} \Big| &\frac{1}{N} \sum_{r,v : |v| > N} \frac{(r\cdot v)}{r^2+v^2+r\cdot v}\frac{1}{r^2v^4}  \big( \widehat{V} (\cdot / N)*\widehat{f}_{N,\ell}\big)_r^2 \big( \widehat{V} (\cdot / N)*\widehat{f}_{N,\ell}\big)_v^2 \Big| \\\ &\leq \frac{C}{N^2} \sum_{r,v : |v| > N} \frac{1}{r^2+v^2}\frac{1}{|r| v^2} \big|  \big( \widehat{V} (\cdot / N)*\widehat{f}_{N,\ell}\big)_r \big|^2 \leq \frac{C}{N^2}  \sum_{r \in \Lambda^*_+} \frac{1}{r^2} \big| \big( \widehat{V} (\cdot / N)*\widehat{f}_{N,\ell}\big)_r \big|^2 \leq \frac{C}{N}, \end{split} \] 
we can restrict the sum to $|v| \le N$. To restrict it also to $|r| \leq N$, we estimate, using a change of variable $v \to - v$,  

\begin{equation}\label{eq:CS32} \begin{split} \Big| \frac{1}{N} &\sum_{|v| \leq N, |r| > N} \frac{(r\cdot v)}{r^2+v^2+r\cdot v}\frac{1}{r^2v^4}  \big( \widehat{V} (\cdot / N)*\widehat{f}_{N,\ell}\big)_r^2 \big( \widehat{V} (\cdot / N)*\widehat{f}_{N,\ell}\big)_v^2 \Big| \\ = \; & \frac{1}{2} \sum_{|v| \leq N , |r| > N} \frac{(r\cdot v)^2}{\big( r^2 + v^2 + r \cdot v \big) \big( r^2 + v^2 - r \cdot v \big)} \frac{1}{r^2 v^4} \big( \widehat{V} (\cdot / N)*\widehat{f}_{N,\ell}\big)_r^2 \big( \widehat{V} (\cdot / N)*\widehat{f}_{N,\ell}\big)_v^2 \\ \leq \; &C \sum_{|v| \leq N , |r| > N} \frac{1}{\big( r^2 + v^2 \big) v^2} \frac{1}{r^2} \big( \widehat{V} (\cdot / N)*\widehat{f}_{N,\ell}\big)_r^2 \big( \widehat{V} (\cdot / N)*\widehat{f}_{N,\ell}\big)_v^2 \\ \leq \; &\frac{1}{N} \sum_{|r| > N} \frac{1}{|r|^3} \big( \widehat{V} (\cdot / N)*\widehat{f}_{N,\ell}\big)_r^2 \leq \frac{C}{N}.
\end{split} \end{equation} 
The term in the first line of (\ref{CS_3}) can be handled similarly. In fact, we control the contribution proportional to $-v^2$ as in (\ref{eq:CS31}). As for the contribution proportional to $(r\cdot v)$, we proceed analogously to (\ref{eq:CS32}). There is here  an additional term arising from the change of variable $v \to -v$, due to the potential $(\widehat{V} (\cdot / N) * \widehat{f}_{N,\ell})_{r+v}$, which can be bounded using that 
\begin{equation}\label{eq:reno-V}  \big| (\widehat{V} (\cdot / N) * \widehat{f}_{N,\ell})_{r+v} - (\widehat{V} (\cdot / N) * \widehat{f}_{N,\ell})_{r-v} \big| \leq C |v|/N \, . \end{equation} 
Finally, we replace all renormalized potentials $(\widehat{V} (\cdot / N) * \widehat{f}_{N,\ell})_p$ with factors of $(\widehat{V} (\cdot / N) * \widehat{f}_{N,\ell})_0$, and then, using \eqref{eq:Vf-8pia}, by factors of $8\pi \frak{a}$. To bound the corresponding errors, we rely again on estimates of the form (\ref{eq:reno-V}). Considering an example among the terms arising from the contribution on the second line of (\ref{CS_3}), we can bound
\[ \begin{split} \Big| \frac{1}{N} \sum_{|r|,|v| \leq N} &\frac{1}{r^2 +v^2 + r\cdot v} \frac{1}{r^2 v^2} \big( \widehat{V} (\cdot / N)*\widehat{f}_{N,\ell}\big)_{r+v}^2  \big( \widehat{V} (\cdot / N)*\widehat{f}_{N,\ell}\big)_v \\ &\hspace{3cm} \times \big[ \big( \widehat{V} (\cdot / N)*\widehat{f}_{N,\ell}\big)_r - \big( \widehat{V} (\cdot / N)*\widehat{f}_{N,\ell}\big)_0 \big] \Big| \\ &\hspace{1.5cm} \leq  \frac{C}{N^2} \sum_{|r|, |v| \leq N} \frac{1}{r^2 + v^2} \frac{1}{|r| v^2} \leq \frac{C}{N^2} \sum_{|r|, |v| \leq N} \frac{1}{|r|^{5/2} |v|^{5/2}} \leq \frac{C}{N} .  \end{split} \]
To estimate the errors arising from the third line of (\ref{CS_3}), on the other hand, we proceed similarly to (\ref{eq:CS32}), performing a change of variable $v \to -v$, to bound 
\[ \begin{split} \Big| \frac{1}{N} &\sum_{|r|,|v| \leq N} \frac{(r\cdot v)}{r^2 + v^2 + r \cdot v} \frac{1}{r^2 v^4} \big( \widehat{V} (\cdot / N)*\widehat{f}_{N,\ell}\big)_v \big[ \big( \widehat{V} (\cdot / N)*\widehat{f}_{N,\ell}\big)_r - \big( \widehat{V} (\cdot / N)*\widehat{f}_{N,\ell}\big)_0 \big] \Big| \\ &= \Big| \frac{1}{2N} \sum_{|r|,|v| \leq N} \frac{(r\cdot v)^2}{(r^2 + v^2 + r \cdot v)(r^2+v^2 - r\cdot v)} \frac{1}{r^2 v^4} \big( \widehat{V} (\cdot / N)*\widehat{f}_{N,\ell}\big)_v \\ &\hspace{6cm} \times  \big[ \big( \widehat{V} (\cdot / N)*\widehat{f}_{N,\ell}\big)_r - \big( \widehat{V} (\cdot / N)*\widehat{f}_{N,\ell}\big)_0 \big] \Big| \\ &\leq \frac{C}{N^2} \sum_{|r|,|v| \leq N} \frac{|r|}{v^2 (r^2 + v^2)^2} \leq \frac{C}{N^2} \sum_{|r|, |v| \leq N} \frac{1}{|r|^{5/2} |v|^{5/2}} \leq \frac{C}{N}. \end{split} \]
Also the errors from the first line of (\ref{CS_3}) can be bounded analogously (also here the change of variables $v \to -v$ needed to handle terms proportional to $(r\cdot v)$ will produce additional contributions, containing an additional difference $(\widehat{V} (\cdot / N) * \widehat{f}_{N,\ell})_{r+v} -  (\widehat{V} (\cdot / N) * \widehat{f}_{N,\ell})_{r-v}$, which can be estimated by $|v|/N$ and can be handled similarly as above). We conclude that 
\begin{equation} \label{eq:wtCRN_2} \wt{C}_{\cJ_N} = \frac{1024 \, \pi^4 \frak{a}^4}{N} \sum_{\substack{r,v\in2\pi\mathbb{Z}^3 \backslash \{ 0 \} \\|r|,|v| \leq N}} \frac{r\cdot v - v^2}{r^2 + v^2 + r\cdot v} \frac{1}{r^2 |v|^4} + \cO (N^{-1}). \end{equation} 

At this point, we can approximate the sum with an integral. Consider first the contribution proportional to $-v^2$. For $\tilde{r} = (\tilde{r}_1, \tilde{r}_2, \tilde{r}_3), \tilde{v} = (\tilde{v}_1, \tilde{v}_2, \tilde{v}_3)  \in 2\pi \bZ^3$, with $2\pi \leq |\tilde{r}|, |\tilde{v}| \leq N$ and $r \in B_{\tilde{r}} = [\tilde{r}_1 -\pi ; \tilde{r}_1 + \pi] \times [\tilde{r}_2 -\pi , \tilde{r}_2 + \pi ] \times [\tilde{r}_3 -\pi ; \tilde{r}_3 + \pi ]$, $v \in B_{\tilde{v}} = [\tilde{v}_1 -\pi ; \tilde{v}_1 + \pi ] \times [\tilde{v}_2 -\pi ; \tilde{v}_2 + \pi ] \times [\tilde{v}_3 -\pi ; \tilde{v}_3 + \pi ]$, we find
\[ \Big| \frac{1}{r^2 + v^2 + r \cdot v} \frac{1}{r^2 v^2} - \frac{1}{\tilde{r}^2 + \tilde{v}^2 + \tilde{r} \cdot \tilde{v}} \frac{1}{\tilde{r}^2 \tilde{v}^2} \Big| \leq C \frac{1}{r^2 + v^2} \frac{1}{r^2 v^2}\big(|r|^{-1}+|v|^{-1}\big). \]
Setting 
\[ U = \bigcup_{\substack{\tilde{r} , \tilde{v} \in 2\pi \bZ^3 : \\ 2\pi \leq |\tilde{r}| , |\tilde{v}| \leq N}}  B_{\tilde{r}} \times B_{\tilde{v}} \]
this implies that 
\[ \begin{split} \Big| \sum_{\substack{ r,v \in 2\pi \bZ^3 \\  2\pi\le |r|, |v| \leq N}} \frac{1}{r^2 + v^2 + r \cdot v} \frac{1}{r^2 v^2} \;-\; & \frac{1}{(2\pi)^6} \int_U \frac{1}{r^2 + v^2 + r \cdot v} \frac{1}{r^2 v^2}  dr dv \Big| \\ & \leq C\int_{|r|, |v| \ge \pi}  \frac{1}{r^2 + v^2} \frac{1}{r^2 |v|^3} dr dv \leq C. \end{split} \]
Observing that 
\[ \{ (r,v) \in \bR^6 : 2\pi \leq |r|, |v| \leq N - C \} \subset U \subset \{ (r,v) \in \bR^6 : \pi  \leq |r|, |v| \leq N +C \} ,  \]
and estimating
\[ \begin{split} \int_{\{ \pi \leq |r|  \leq 2\pi\} \cup \{ N-C \leq |r| \leq N+C \}} dr  \int_{\pi \leq |v| \leq N+C} dv \frac{1}{r^2 + v^2 + r \cdot v} \frac{1}{r^2 v^2} \leq C \end{split} \]
we conclude that 
\[ \begin{split} \Big| \sum_{\substack{ r,v \in 2\pi \bZ^3 \\  2\pi\le |r|, |v| \leq N}} \frac{1}{r^2 + v^2 + r \cdot v} \frac{1}{r^2 v^2} \;-\; & \frac{1}{(2\pi)^6} \int_{2\pi \leq |r|, |v| \leq N} \frac{1}{r^2 + v^2 + r \cdot v} \frac{1}{r^2 v^2}  dr dv \Big|  \leq C. \end{split} \]
Using the identity
\begin{equation*}
\sum_{|r|, |v| \leq N} \frac{r \cdot v}{r^2 + v^2 + r \cdot v} \frac{1}{r^2 |v|^4} = - \sum_{|r|,|v| \leq N}  \frac{(r\cdot v)^2}{\big[ r^2 + v^2 + r \cdot v \big] \big[r^2 + v^2 - r\cdot v \big]}  \frac{1}{r^2 |v|^4} \end{equation*} 
we can obtain a similar bound also for the contribution associated with the factor $r\cdot v$ appearing on the r.h.s. of (\ref{eq:wtCRN_2}). Thus 
\[ \wt{C}_{\cJ_N} =\frac{1024 \pi^4 \frak{a}^4}{N} \frac{1}{(2\pi)^6} \int_{2\pi \leq |r|, |v| \leq N} \frac{r\cdot v - v^2}{r^2 + v^2 + r \cdot v} \frac{1}{r^2 |v|^4 }dr dv + \cO (N^{-1}). \]
By explicit computation (for fixed $r$, we first integrate over $v$ using spherical coordinates $(|v|, \theta, \ph)$, with $r \cdot v = |r| |v| \cos \theta$; the result of the $v$ integral is a radial function of $r$, which can be integrated using spherical coordinates for $r$),  we find 
\[ \wt{C}_{\cJ_N} = - 64 \pi \Big( \frac{4}{3} \pi -\sqrt{3} \Big) \frak{a}^4  \frac{(\log N)}{N} + \cO (N^{-1}), \]
which concludes the proof of (\ref{eq:main}). 
\end{proof}

\section{Quadratic renormalization: proof of Theorem \ref{thm:quadratic}} 
\label{sec:cGN} 

In order to show Theorem \ref{thm:quadratic} we will rely on bounds controlling the growth of the number of excitations and of their energy w.r.t. the action of the generalized Bogoliubov transformation $e^{B_\mu}$. Recall that, from Lemma \eqref{eq:norms_mu,g,s}, $\| \mu \|_2 \le C$ uniformly in $N$. As shown for example in \cite[Lemma 3.1]{BS}, this implies that for every $j \geq 0$ there exists $C > 0$ such that 
\begin{equation}
\label{eq:control-eBmu}  
e^{-B_\mu} (\cN_+ + 1)^j e^{B_\mu} \leq C(  \cN_+ + 1)^j . 
\end{equation}
From (\ref{eq:norms_eta}), we also obtain rough, non-uniform estimates on the growth of the energy. The proof of the following lemma can be found in \cite[Lemma 7.1]{BBCS} (the coefficients $\mu_p$ and $\eta_p$ satisfy the same bounds). 
\begin{lemma} \label{lemma:a_priori_quadratic}
Let $\cK, \cV_N$ be defined as in \eqref{eq:K,V}. Under the same assumptions as in Theorem \ref{thm:quadratic}, for every $j\in\mathbb{N}$ there exists $C>0$ such that
    \begin{equation} \label{eq:a_priori_quadratic}
        \begin{split}
            e^{-B_\mu} \mathcal{K}(\mathcal{N}_++1)^j e^{B_\mu} \le\;& C \mathcal{K}(\mathcal{N}_++1)^j+ CN(\mathcal{N}_++1)^{j+1}\\
            e^{-B_\mu} \mathcal{V}_N (\mathcal{N}_++1)^j e^{B_\mu} \le\;& C \mathcal{V}_N(\mathcal{N}_++1)^j+C N(\mathcal{N}_++1)^j.
        \end{split}
    \end{equation}
\end{lemma}

To prove Theorem \ref{thm:quadratic}, we will also need to compute the action of the generalized Bogoliubov transformation $e^{B_\mu}$ on creation and annihilation operators more precisely. To this end, we introduce the notation $\gamma_p = \cosh \mu_p$, $\s_p = \sinh \mu_p$ and we define operators $d_p, d^*_p$, for $p \in \L^*_+$, through the identities  
\begin{equation}
e^{-B_\mu} b_p e^{B_\mu} = \gamma_pb_p+\sigma_p b^*_{-p} + d_p,\qquad e^{-B_\mu} b_p^* e^{B_\mu} =\gamma_pb_p^*+\sigma_p b_{-p} + d_p^*. \label{eq:def_d}
\end{equation}
In position space, we similarly introduce operator-valued distributions $\check{d}_x, \check{d}_x^*$, requiring that 
\begin{equation*}
e^{-B_\mu}\check{b}_xe^{B_\mu} = b(\check{\gamma}_x)+b^*(\check{\sigma}_x)+\check{d}_x,\qquad e^{-B_\mu}\check{b}_x^*e^{B_\mu} = b(\check{\gamma}_x)^*+b(\check{\sigma}_y)+\check{d}_x^*.
\end{equation*}
Here $\check{\s}_x (y) = \check{\s} (y;x) = \sum_{p \in \L_+^*} \s_p \exp (i p \cdot (x-y))$ and similarly for the distribution $\check{\gamma}_x$. 

On states with few excitations, ie $\cN_+ \ll N$, the operators $b^*, b$ are close to the standard creation and annihilation operators $a^*, a$. Thus, we expect $B_\mu$ to act almost as a Bogoliubov transformation; equivalently, we expect $d^* , d$ to be small. To prove bounds on the fields $d^* , d$, we will use the integral representation proven in the following lemma. 
\begin{lemma}
    For $s\in[0,1]$, let $\gamma_p^{(s)}=\cosh(s \mu_p)$ and $\sigma_p^{(s)}=\sinh(s \mu_p)$. Then
    \begin{equation} \label{eq:d_detailed_expansion}
        \begin{split}
            d_p=\;&-\frac{1}{N} \int_0^1 ds\, e^{-(1-s) B_\mu}\bigg[\gamma_p^{(s)} \Big( \mu_p \mathcal{N}_+ b^*_{-p}+\sum_{q\in\Lambda_+^*} \mu_q b^*_q a^*_{-q} a_p \Big)\\
            &\qquad\qquad\qquad\qquad\qquad+\sigma_p^{(s)} \Big( \mu_p \mathcal{N}_+ b_p+\sum_{q\in\Lambda_+^*} \mu_q a^*_{-p} a_{-q} b_q \Big) \bigg] e^{(1-s)B_\mu}.
        \end{split}
    \end{equation}
\end{lemma}

\begin{proof}
    The commutators
    \begin{equation*}
        \begin{split}
            [b_p,B]=\mu_p\big(1-\tfrac{\mathcal{N}_+}{N}\big) b^*_{-p}-\frac{1}{N}\sum_{q\in\Lambda_+^*} \mu_q b^*_q a^*_{-q} a_p\\
            [b^*_{-p},B]=\mu_p b_p \big(1-\tfrac{\mathcal{N}_+}{N}\big)-\frac{1}{N} \sum_{q\in\Lambda_+^*} \mu_q a^*_{-p} a_{-q} b_q
        \end{split}
    \end{equation*}
    imply the identity
    \begin{equation*}
    \begin{split}
        \frac{d}{ds}\bigg( e^{sB_\mu} \Big( \gamma_p^{(s)} b_p+\sigma_p^{(s)} b^*_{-p} \Big) e^{-sB_\mu}\bigg)=\;&\frac{1}{N}e^{sB_\mu} \bigg[ \gamma_p^{(s)}\Big( \mu_p \mathcal{N}_+ b^*_{-p}+\sum_{q\in\Lambda_+^*}\mu_q b^*_q a^*_{-q} a_p  \Big) \\
        &+\sigma_p^{(s)}\Big( \mu_p b_p \mathcal{N}_+ +\sum_{q\in\Lambda_+^*} \mu_q a^*_{-p} a_{-q} b_q \Big)\bigg]e^{-sB_\mu}.
    \end{split}
    \end{equation*}
    Integrating both sides from $s=0$ to $s=1$ and comparing with \eqref{eq:def_d} concludes the proof.
\end{proof}

In the next lemma, we collect bounds for the operators $d^*, d$ that will be used throughout the proof of Theorem \ref{thm:quadratic}. Similar estimates have been shown in \cite{BBCS}, but only under the assumption that the $\ell^2$-norm of the coefficients in the Bogoliubov transformation is small enough (which is not satisfied here, since we included $\tau$ in (\ref{eq:defmu})). With the help of the representation (\ref{eq:d_detailed_expansion}), we relax this assumption. 

\begin{lemma} \label{lemma:d}
Let $\mu\in \ell^2(\Lambda_+^*)$, and let $n\in \mathbb{N}$. Then there exists $C>0$ such that
\begin{equation} \label{eq:decay_d}
\begin{split}
\|(\mathcal{N}_++1)^{n/2}d_p\xi\|\le\;&\frac{C}{N}\Big[ |\mu_p|\|(\mathcal{N}_++1)^{(n+3)/2}\xi\|+\|b_p(\mathcal{N}_++1)^{(n+2)/2}\xi\| \Big]\\
\|(\mathcal{N}_++1)^{n/2}d_p^*\xi\|\le\;&\frac{C}{N}\|(\mathcal{N}_++1)^{(n+3)/2}\xi\|\\
\|(\mathcal{N}_++1)^{n/2}a_qd_p\xi\| \leq & \; \frac{C}{N}\Big[|\mu_q|\|(\cN_++1)^{(n+3)/2}a_p\xi\| + |\mu_p|\|(\cN_++1)^{(n+3)/2}a_q\xi\|\\
   & \qquad+ \d_{p,-q} |\mu_q|\|(\cN_++1)^{(n+3)/2}\xi\| +\|(\cN_++1)^{(n+2)/2}a_pa_q\xi\|\\
   &\qquad+|\mu_p||\mu_q|\|(\cN_++1)^{(n+4)/2}\xi\|\Big]\\
\|(\mathcal{N}_++1)^{n/2}a_qd_p^*\xi\| \leq &\; \frac{C}{N}\Big[|\mu_q|\|(\cN_++1)^{(n+4)/2}\xi\| + \|(\cN_++1)^{(n+3)/2}a_q\xi\|\\
   &\qquad+ \d_{p,-q}\|(\cN_++1)^{(n+2)/2}\xi\|\Big]\,,
\end{split}
\end{equation}
for all $p\in\Lambda_+^*$ and $\xi\in\mathcal{F}_\perp^{\le N}$. Moreover, as distributions,
\begin{equation}\label{eq:decay_d_position}
\begin{split}  \| (\mathcal{N}_+ + 1)^{n/2} \check{d}_x \xi \| \le\;&\frac{C}{N} \Big[ \| (\mathcal{N}_+ + 1)^{(n+3)/2} \xi \| + \| \check{a}_x (\mathcal{N}_+ + 1)^{(n+2)/2} \xi \| \Big] \\
\| (\mathcal{N}_+ + 1)^{n/2} \check{a}_y \check{d}_x \xi \| \leq \;& \frac{C}{N} \Big[ \| \check{a}_x (\mathcal{N}_+ + 1)^{(n+1)/2} \xi \| + \big(1 + |\check{\mu} (x-y)|\big) \| (\mathcal{N}_+ + 1)^{(n+2)/2} \xi \| \\ 
&\qquad + \| \check{a}_y (\mathcal{N}_+ + 1)^{(n+3)/2} \xi \| + \| \check{a}_x \check{a}_y (\mathcal{N}_+ + 1)^{(n+2)/2} \xi \| 
\Big] \\ 
\| (\mathcal{N}_+ + 1)^{n/2} \check{d}_x \check{d}_y \xi \| \leq \; &\frac{C}{N^2} \Big[ \| (\mathcal{N}_++ 1)^{(n+6)/2} \xi \| + |\check{\mu} (x-y)| \| (\mathcal{N}_+ + 1)^{(n+4)/2}  \xi \| \\
&\qquad + \| \check{a}_x (\mathcal{N}_+ + 1)^{(n+5)/2} \xi \| + \| \check{a}_y (\mathcal{N}_+ + 1)^{(n+5)/2} \xi \| \\
&\qquad + \| \check{a}_x \check{a}_y (\mathcal{N}_+ +  1)^{(n+4)/2} \xi \| \Big] \,.\end{split} \end{equation}
\end{lemma}
A proof of Lemma \ref{lemma:d} is given in Appendix \ref{sec:d-ops}. The main message from Lemma \ref{lemma:d} is the fact that $d$-operators, defined through (\ref{eq:def_d}), are small (and therefore the action of $e^{-B_\mu}$ is close to the action of a Bogoliubov transformation) on states with few excitations. For $d_p$ (but not for $d^*_p$), these bounds also allow us to extract some decay for large momenta $p \in \Lambda^*_+$.

We proceed now with the computation of the renormalized excitation Hamiltonian $\cG_N$, which will lead to the proof of Theorem \ref{thm:quadratic}. From (\ref{eq:cLN}), we find 
\[ \cG_N = \frac{\wh{V} (0)}{2} (N-1) + \cG_N^{(2,\cK)} + \cG_N^{(2,V)} + \cG_N^{(3)} + \cG_N^{(4)} \]
with \[  \begin{split}  \mathcal{G}_N^{(2,\mathcal{K})} &= e^{-B_\mu} \mathcal{K} e^{B_\mu}, \, \quad \mathcal{G}_N^{(2,V)}= e^{-B_\mu} \mathcal{L}^{(2,V)}_N e^{B_\mu}, \\  \mathcal{G}_N^{(3)} &= e^{-B_\mu} \mathcal{L}^{(3)}_N e^{B_\mu}, \quad \mathcal{G}_N^{(4)}= e^{-B_\mu} \mathcal{V}_N e^{B_\mu}.\end{split} \] The form of the operators $\cG^{(2,K)}_N, \cG^{(2,V)}_N, \cG_N^{(3)}, \cG_N^{(4)}$ will be determined in Props. \ref{prop:quad_on_K} - \ref{prop:quad_on_quartic}, up to negligible errors. To this end, we will argue similarly as in the proof of  \cite[Proposition 3.2]{BBCS}; however, to resolve the energy to lower order, we will need to keep several additional terms, which did not play an important role in \cite{BBCS}.

First of all, we establish the form of the operator $\cG_N^{(2,K)} = e^{-B_\mu} \mathcal{K} e^{B_\mu}$.
\begin{prop} \label{prop:quad_on_K} Under the same assumptions of Theorem \ref{thm:quadratic} we have
\begin{equation} \label{eq:quadratic_on_K}
\begin{split}
\mathcal{G}^{(2,K)}_N=\;& \mathcal{K} + \sum_{p \in\Lambda_+^* } p^2 \Big[\sigma^2_p + 2 \s_p^2 b_p^* b_p + \gamma_p\sigma_p(b_pb_{-p}+b^*_p b^*_{-p})\Big] +\frac{1}{N} \sum_{p\in\Lambda_+^*} p^2 \eta_p^2(1-\mathcal{N}_+) \\
&+\frac{1}{N} \sum_{p,q\in\Lambda_+^*} p^2 \eta_p^2\Big[ \sigma_q^2+\big(\gamma_q^2+\sigma_q^2\big) b^*_q b_q+\gamma_q\sigma_q\big(b_qb_{-q}+b^*_q b^*_{-q}\big) \Big] \\
&+\frac{1}{2N} \sum_{p,q \in \Lambda_+^*} p^2 \eta_p\big(b_pb_{-p}+b^*_pb^*_{-p}\big)\bigg[\sigma_q^2- \frac{\gamma_q\sigma_q}{\eta_q} + 1\bigg]\\
&+\frac{1}{2N} \sum_{p,q\in \Lambda_+^*} p^2 \eta_p b^*_p b^*_{-p}b^*_qb^*_{-q}\bigg[\gamma_q\sigma_q-\frac{\sigma_q^2}{\eta_q}\bigg]+\mathrm{h.c.}\\
&+\sum_{p\in\Lambda_+^*} p^2 \eta_p \left(b_{-p}d_p + d^*_p b_{-p}^*\right)+\mathcal{E}_{\cG_N^K}
\end{split}
\end{equation}
with
\begin{equation}
    \pm\mathcal{E}_{\cG_N^K} \le  \eps \cN_+ +\frac{C}{\eps N} (\mathcal{K} + \cN_+^2 + 1)(\mathcal{N}_++1) .
\end{equation}
\end{prop}

{\it Remark.} Compared with the corresponding result in \cite{BBCS}, we additionally need to keep track of the terms appearing in the third and fourth line of \eqref{eq:quadratic_on_K}. These terms will be included in the operator $\mathcal{T}_{\mathcal{G}_N}$ defined in \eqref{eq:cT_G_N} and eventually will be proven to be negligible after cubic conjugation.

\begin{proof}[Proof of Proposition \ref{prop:quad_on_K}]
    Writing
    \begin{equation*}
        \mathcal{K}=\frac{N-1}{N} \sum_{p\in\Lambda_+^*}p^2 b^*_p b_p+\sum_{p\in \Lambda_+^*}p^2 b^*_p b_p\frac{\mathcal{N}_+}{N}+\mathcal{K}\frac{(\mathcal{N}_+-1)^2}{N^2} \, ,
    \end{equation*}
    and using \eqref{eq:a_priori_quadratic} to get rid of the last term, we find
    \begin{equation} \label{eq:first_decomp_quad_K}
        e^{-B_\mu} \mathcal{K} e^{B_\mu}= \sum_{p\in \Lambda_+^*}p^2 e^{-B_\mu} b^*_p b_pe^{B_\mu}+\frac{1}{N}\sum_{p,q\in \Lambda_+^*}p^2 e^{-B_\mu} b^*_p b^*_q b_p b_q e^{B_\mu} +{\mathcal{E}_{\mathcal{K},1}}
    \end{equation}
    with
    \begin{equation*}
        \pm{\mathcal{E}_{\mathcal{K},1}} \le CN^{-1} (\mathcal{K}+ \cN_+^2 + 1)( \cN_+ + 1) \,.
    \end{equation*}

We decompose the first term on the r.h.s. of (\ref{eq:first_decomp_quad_K}) as  
    \begin{equation*}
        \sum_{p\in \Lambda_+^*}p^2 e^{-B_\mu} b^*_p b_pe^{B_\mu}=E_1+E_2+E_3,
    \end{equation*}
    with
    \begin{equation*}
        \begin{split}
            E_1=\;&\sum_{p\in\Lambda_+^*}p^2(\gamma_p b^*_{-p}+\sigma_p b_{-p})(\gamma_p b_p+\sigma_p b^*_{-p})\\
            E_2=\;&\sum_{p\in\Lambda_+^*}p^2 (\gamma_p b^*_{p}+\sigma_p b_{-p}) d_p + \mathrm{h.c.} \\
            E_3=\;&\sum_{p\in\Lambda_+^*}p^2 d^*_pd_p.
        \end{split}
    \end{equation*}
    With the commutation relations (\ref{eq:CCR_b}) (and using $|\mu_p - \eta_p| = |\tau_p| \leq C / |p|^4$) we obtain  
    \begin{equation}\label{eq:E1}
        E_1=\mathcal{K}+\sum_{p\in\Lambda_+^*} p^2\bigg(\sigma_p^2-\eta_p^2\frac{\mathcal{N}_+}{N}\bigg)+\sum_{p\in\Lambda_+^*}\Big(p^2 \gamma_p \sigma_p(b_pb_{-p}+b^*_p b^*_{-p})+2p^2\sigma_p^2b^*_p b_p\Big)+\mathcal{E}_{\mathcal{K},2},
    \end{equation}
    with
 $\pm \mathcal{E}_{\mathcal{K},2} \le CN^{-1} \mathcal{K}(\mathcal{N}_++1)$. Applying \eqref{eq:decay_d}, we find (see \cite[Equation above (7.16)]{BBCS}) 

    \begin{equation}\label{eq:E3}
        \pm E_3\le C N^{-1} (\mathcal{K}+ \cN^2_+ +1)(\mathcal{N}_++1) \, . 
    \end{equation}
As for $E_2$, using \eqref{eq:decay_d}, $|\gamma_p - 1| \leq C/|p|^4$, $|\sigma_p - \mu_p| \leq C / |p|^6$ and again $|\mu_p - \eta_p| = |\tau_p| \leq C / |p|^4$, we obtain  
   \begin{equation*}
    \begin{split}
        E_2=\;&\sum_{p\in\Lambda_+^*}p^2 \eta_p( b_{-p} d_p+\mathrm{h.c.}) +\sum_{p\in \Lambda_+^*}p^2 (b^*_pd_p+\mathrm{h.c.})+\mathcal{E}_{\mathcal{K},3},
    \end{split}
    \end{equation*}
    with $\pm\mathcal{E}_{\mathcal{K},3} \le C N^{-1} (\mathcal{N}_++1)^2$. The first term contributes to the r.h.s. of \eqref{eq:quadratic_on_K}. As for the second term, we expand $d_p$ using \eqref{eq:d_detailed_expansion}. In the resulting expression, we observe that the coefficients $\gamma_p^{(s)}$ and $\s_p^{(s)}$ can be replaced by $1$ and $s \mu_p$, respectively, up to negligible errors. As an example, consider 
         \begin{equation} \label{eq:example_bound_E_2}
        \begin{split}
            \pm \frac{1}{N}\sum_{p,q\in \Lambda_+^*}&p^2\mu_q  \int_0^1ds\,(\gamma_p^{(s)}-1) \Big( \big\langle \xi,b^*_pe^{(s-1)B_\mu} b^*_q a^*_{-q} a_pe^{-(s-1)B_\mu}\xi\big\rangle+\mathrm{h.c.}\Big)\\
            \le\;&\frac{C}{N}\int_0^1ds\sum_{p,q\in\Lambda_+^*}\frac{1}{|q|^2} \big\|b_qe^{-(s-1)B_\mu}b_p\xi\big\|\,\big\|a_p\mathcal{N}_+^{1/2}e^{-(s-1)B_\mu}\xi\big\|\\
            \le\;&C N^{-1} \langle \xi, (\mathcal{N}_++1)^3\xi\rangle,
        \end{split}
    \end{equation}
We obtain (replacing $a, a^*$ with $b,b^*$ only produces small corrections) 
\begin{equation*} \label{eq:ugly_expansion_b^*d}
        \begin{split}
            \sum_{p\in\Lambda_+^*}& p^2(b^*_p d_p+\mathrm{h.c.})\\
            =\;&-\frac{1}{N} \sum_{p\in\Lambda_+^*}p^2 \mu_p \int_0^1 ds\,b^*_p \,e^{(s-1)B_\mu}b^*_{-p}\Big(\sum_{q\in\Lambda_+^*}b^*_q b_q+1\Big)e^{-(s-1)B_\mu}+\mathrm{h.c.}\\
            &-\frac{1}{N} \sum_{p,q\in\Lambda_+^*} p^2\mu_q \int_0^1ds \,b^*_p\,e^{(s-1)B_\mu} b^*_q b^*_{-q} b_pe^{-(s-1)B_\mu}+\mathrm{h.c.}\\&-\frac{1}{N} \sum_{p,q,\in\Lambda_+^*} p^2\mu_p\mu_q \int_0^1 ds\,s\, b^*_p e^{(s-1)B_\mu} b^*_{-p} b_q b_{-q}e^{-(s-1)B_\mu}+\mathrm{h.c.}+\mathcal{E}_{\mathcal{K},4}
        \end{split}
    \end{equation*}
    with $\pm \mathcal{E}_{\mathcal{K},4}\le C N^{-1} (\mathcal{N}_++1)^3$. With (\ref{eq:def_d}) we observe that $e^{(s-1) B_\mu} b_{-p}^* e^{-(s-1) B_\mu} \simeq b_{-p}^*$ (in the first and third line) and that $e^{(s-1) B_\mu} b_{-p}^* e^{-(s-1) B_\mu} \simeq \mu_p b_{p}$ (in second line), up to contributions that can be bounded similarly as in (\ref{eq:example_bound_E_2}). Setting $t = 1-s$, we obtain 
         \begin{equation*}
        \begin{split}
       \sum_{p\in\Lambda_+^*}& p^2(b^*_p d_p+\mathrm{h.c.})\\
            =\;&-\frac{1}{N} \sum_{p,q\in\Lambda_+^*} p^2\mu_p\,b^*_p b^*_{-p} \int_0^1dt \,e^{-tB_\mu}\Big( b^*_q b_q+t\,\mu_q\, b^*_q b^*_{-q}+(1-t)\mu_q\,b_qb_{-q}\Big) e^{tB_\mu}+\mathrm{h.c.}\\
            &-\frac{1}{N} \sum_{p\in\Lambda_+^*} p^2 \mu_p\, \big(b^*_{p} b^*_{-p}+\mathrm{h.c.}\big) +\mathcal{E}_{\mathcal{K},5}      
        \end{split}
    \end{equation*}
    with $\pm \mathcal{E}_{\mathcal{K},5}\le CN^{-1}(\mathcal{N}_++1)^3$. Applying again (\ref{eq:def_d}), noticing that contributions involving $d,d^*$ operators are negligible as a consequence of \eqref{eq:decay_d}, rearranging terms in normal order, computing the integrals over $t$ explicitly and using that $|\mu_p - \eta_p| \leq C / |p|^4$, we arrive at 
     \begin{equation}\label{eq:E2} 
        \begin{split}
            E_2=\;&\sum_{p\in\Lambda_+^*} p^2(\eta_p b_{-p} d_p+\mathrm{h.c.})\\
            &-\frac{1}{N} \sum_{p,q\in\Lambda_+^*} p^2\eta_p\,b^*_p b^*_{-p}\bigg[\bigg(\sigma_q^2+\frac{\gamma_q\sigma_q}{\mu_q}\bigg)b^*_q b_q+\frac{1}{2}\bigg(\frac{\sigma^2_q}{\mu_q}+\gamma_q\sigma_q\bigg) \big(b_qb_{-q}+b^*_q b^*_{-q}\big)\bigg]\\
            &-\frac{1}{N} \sum_{p\in\Lambda_+^*} p^2 \eta_p\,\big( b^*_{p} b^*_{-p}+\mathrm{h.c.}\big)\bigg[ 1+\frac{1}{2}\sum_{q\in\Lambda_+^*} \Big(\sigma_q^2+\frac{\gamma_q\sigma_q}{\mu_q}-1\Big)\bigg] +\mathcal{E}_{\mathcal{K},6},
        \end{split}
    \end{equation}
where $\pm\mathcal{E}_{\mathcal{K},6} \le C N^{-1} (\mathcal{N}_++1)^3$.
   
Finally, we consider the second term on the r.h.s. of (\ref{eq:first_decomp_quad_K}). Again we apply (\ref{eq:def_d}) and we observe that all contributions involving the operators $d,d^*$ are negligible, by (\ref{eq:decay_d}). Furthermore, we notice that the coefficients $\gamma_p$ and $\sigma_p$ can be replaced everywhere by one and, respectively, $\eta_p$, up to a negligible error (using again $|\mu_p - \eta_p| \leq C /|p|^4$). Finally, we remark that all terms proportional to $b_p^* b_p$ are also irrelevant (because the sum over $p$ can be controlled by $\cK$). Collecting all terms proportional to $b_p^* b_{-p}^*, b_p b_{-p}$ (and all commutators arising from normal ordering), we arrive at 
  \begin{equation}\label{eq:E2_2}
         \begin{split}
        \frac{1}{N}&\sum_{p,q\in\Lambda_+^*}p^2 e^{-B_\mu}b^*_pb^*_q b_q b_pe^{B_\mu}\\
        =\;&\frac{1}{N}\sum_{p\in\Lambda_+^*}p^2 \eta_p^2 \bigg(1+\sum_{q\in\Lambda_+^*}\sigma_q^2\bigg)+\frac{1}{N}\sum_{p,q\in\Lambda_+^*} p^2\eta_p^2\Big[(\gamma_q^2+\sigma_q^2)b^*_q b_q+\gamma_q\sigma_q\big(b_qb_{-q}+b^*_{q} b^*_{-q}\big)\Big]\\
            &+\frac{1}{N} \sum_{p,q\in\Lambda_+^*}p^2\eta_p b^*_p b^*_{-p} \Big( \big(\gamma_q^2+\sigma_q^2\big) b^*_q b_q+ \gamma_q \sigma_q \big(b_qb_{-q}+b^*_q b^*_{-q}\big)\Big)+\mathrm{h.c.}\\
            &+\frac{1}{N} \sum_{p\in\Lambda_+^*}p^2\eta_p b^*_p b^*_{-p}\bigg(1+\sum_{q\in\Lambda_+^*} \sigma_q^2\bigg)+\mathcal{E}_{\mathcal{K},7},
    \end{split}
    \end{equation}
    with $\pm\mathcal{E}_{\mathcal{K},7} \le C N^{-1} (\mathcal{K} + \cN_+^2 + 1) (\cN_+ + 1)$. Lastly, we notice that, in (\ref{eq:E2}) and in (\ref{eq:E2_2}), the quartic terms which contain both creation and annihilation operators can be treated as errors bounded by $\eps \cN_++C\eps^{-1}N^{-1}\cK(\cN_++1)^2$. As an example, consider \begin{equation} \label{eq:example_bound_E_2_2}
        \begin{split}
            \pm \frac{1}{N}\sum_{p,q\in \Lambda_+^*}&p^2\eta_p\,  \gamma_q\sigma_q\langle \xi,b^*_pb^*_{-p}b_qb_{-q}\xi \rangle +\mathrm{h.c.} \\
            \le\;& \frac{C}{N}\sum_{p,q\in \Lambda_+^*}p^2|\eta_p|\,  |\sigma_q| \|b^*_qb_{-p}b_p \xi\| \; \|b_{-q} \xi\|
            \le\; \eps \langle \xi, \cN_+\xi\rangle +\frac{C}{\eps N} \langle \xi, \mathcal{K}(\mathcal{N}_++1)^2\xi\rangle.
        \end{split}
    \end{equation} Combining this with (\ref{eq:E1}), (\ref{eq:E3}), (\ref{eq:E2}), we obtain \eqref{eq:quadratic_on_K}. 
\end{proof}

Next, we consider the operator $\mathcal{G}^{(2,V)}_N = e^{-B_\mu} \mathcal{L}^{(2,V)}_N e^{B_\mu}$.
\begin{prop} 
Under the same assumptions of Theorem \ref{thm:quadratic} we have
\begin{equation} \label{eq:quad_on_L^2}
\begin{split}
\mathcal{G}^{(2,V)}_N=\;&\sum_{p\in\Lambda_+^*}\widehat{V} (p/N) \left(\sigma_p^2+\gamma_p\sigma_p-\eta_p\frac{\cN_+}{N}\right)\\
&+\sum_{p\in\Lambda_+^*}\widehat{V} (p/N) (\gamma_p+\sigma_p)^2\left(b^*_pb_p+\frac{1}{2}(b_pb_{-p}+b^*_pb^*_{-p})\right)\\
&+\frac{1}{2}\sum_{p\in\Lambda_+^*}\widehat{V} (p/N) \Big[b_{-p}d_p+d_p (b_{-p}+\eta_pb^*_{p}) +\mathrm{h.c.} \Big] +\mathcal{E}_{\mathcal{G}_{N}^{V}},
\end{split}
\end{equation}
with 
\begin{equation*}
    \pm \mathcal{E}_{\mathcal{G}_N^{V}}\le C N^{-1}  (\mathcal{K} + \cN_+^3 + 1) (\mathcal{N}_++1) .
\end{equation*}
\end{prop}

\begin{proof}
    With \eqref{eq:def_L}, we write 
    \begin{equation*}
    \cG_N^{(2,V)} =  e^{-B_\mu}\mathcal{L}_N^{(2,V)}e^{B_\mu}=F_1+F_2+F_3+F_4
    \end{equation*}
    with
    \begin{equation*}
        \begin{split}
            F_1=\;& \sum_{p\in\Lambda_+^*}\widehat{V} (p/N) e^{-B_\mu}b^*_p b_pe^{B_\mu}\\
            F_2=\;&-\frac{1}{N}\sum_{p\in\Lambda_+^*}\widehat{V} (p/N)e^{-B_\mu}a^*_p a_pe^{B_\mu}\\
            F_3=\;&\frac{1}{2}\sum_{p\in\Lambda_+^*}\widehat{V} (p/N)e^{-B_\mu}\big(b^*_pb^*_{-p}+b_pb_{-p}\big)e^{B_\mu}\\
            F_4=\;&\frac{\widehat{V}(0)}{2N}e^{-B_\mu}\mathcal{N}_+(\mathcal{N}_+-1)e^{B_\mu}.
        \end{split}
    \end{equation*}
    Clearly $\pm(F_2+F_4) \le CN^{-1}(\mathcal{N}_++1)^2$. Proceeding as in \cite[Proposition 7.2]{BBCS} (using (\ref{eq:def_d}) and the bounds (\ref{eq:decay_d})) and normal ordering using (\ref{eq:CCR_b}), we find
    \begin{equation*} 
        F_1=\sum_{p\in\Lambda_+^*}\widehat{V}(p/N)\Big(\big(\gamma_p^2+\sigma_p^2)b^*_pb_p+\gamma_p\sigma_p \big(b^*_pb^*_{-p}+b_pb_{-p}\big)\Big)+\mathcal{E}_{V,1},
    \end{equation*}
    with $\pm \mathcal{E}_{V,1}\le CN^{-1}(\mathcal{N}_++1)^2$. As for $F_3$, we apply \eqref{eq:def_d} to decompose    \begin{equation*}
        \begin{split}
F_3=\;&\frac{1}{2}\sum_{p\in\Lambda_+^*}\widehat{V} (p/N) \Big(\big(\gamma_p^2+\sigma_p^2)(b^*_pb^*_{-p}+b_pb_{-p})+2\gamma_p\sigma_p(2b^*_pb_p+1)-2\eta_p\frac{\cN_+}{N}\Big)\\
            &+\frac{1}{2}\sum_{p\in\Lambda_+^*}\widehat{V} (p/N) \Big[b_{-p}d_p+d_p (b_{-p}+\eta_pb^*_{p}) \Big]+\mathrm{h.c.} +\mathcal{E}_{V,2},
        \end{split}
    \end{equation*}
    with 
    \begin{equation}
    \label{eq:E_V_2}
    \begin{split}
        \mathcal{E}_{V,2}=&\frac{1}{2}\sum_{p\in\Lambda_+^*}\widehat{V} (p/N) \Big[ 2(\eta_p-\gamma_p\sigma_p)\frac{\cN_+}{N}+(\gamma_p-1)(b_{-p}d_p+d_pb_{-p})\\
        &\qquad \qquad \qquad  +\sigma_pb^*_{p}d_p+(\sigma_p-\eta_p)d_pb^*_{-p}+d^*_p d^*_{-p}+d_{-p}d_p\Big].
        \end{split}
    \end{equation}
    Using \eqref{eq:decay_d}, we bound the last two terms of $\mathcal{E}_{V,2}$ by 
    \begin{equation*}
    \begin{split}
        \pm \frac{1}{2}\sum_{p\in\Lambda_+^*}&\widehat{V} (p/N) \,\big(\langle \xi,d^*_p d^*_{-p}\xi  \rangle +\mathrm{h.c.}\big) \\\le\;& C \sum_{p\in\Lambda_+^*} |\widehat{V} (p/N) | \,\big\|(\mathcal{N}_++1)^{1/2}d^*_{-p}\xi\big\|\,\big\|(\mathcal{N}_++1)^{-1/2}d_p\xi\big\|\\
        \le\;&\frac{C}{N^2}\big\|(\mathcal{N}_++1)^{2}\xi\big\|\bigg(\sum_{p\in\Lambda_+^*}\frac{1}{|p|^2} |\widehat{V} (p/N)| \bigg)^{1/2}\\
        &\qquad\times \bigg(\sum_{p\in\Lambda_+^*} p^2|\mu_p|^2\big\|(\mathcal{N}_++1)\xi\big\|^2+\sum_{p\in\Lambda_+^*} p^2\big\|b_p(\mathcal{N}_++1)^{1/2}\xi\big\|^2 \bigg)^{1/2}\\
        \le\;&C N^{-1} \langle\xi, (\mathcal{K} + \cN_+^3 + 1) (\mathcal{N}_++1)\xi\rangle.
    \end{split}
    \end{equation*}
Using \eqref{eq:decay_d}, the rest of the terms of $\mathcal{E}_{V,2}$ can be shown to be negligible as well.
Arranging the main contributions to $F_1$ and $F_3$ in normal order, we obtain \eqref{eq:quad_on_L^2}, up to another negligible remainder.
\end{proof}

We now discuss the cubic term $\mathcal{G}^{(3)}_N=e^{-B_\mu} \mathcal{L}^{(3)}_Ne^{B_\mu}$.
\begin{prop} 
Under the same assumptions of Theorem \ref{thm:quadratic} we have 
\begin{equation}\label{eq:cGN3} 
\begin{split}
\mathcal{G}^{(3)}_N=\; &\mathcal{C}_{\mathcal{G}_N} + \mathcal{E}_{\mathcal{G}_{N}^{3}}
\end{split}
\end{equation}
where $\cC_{\cG_N}$ is defined as in (\ref{eq:cC_G_N}),  
and 
\begin{equation*}
    \pm \mathcal{E}_{\mathcal{G}_{N}^{3}} \le \eps \cN_++ \frac{C}{\eps N} (\mathcal{H}_N + \cN_+^2 + 1) (\cN_+ + 1) .
\end{equation*}
\end{prop}

\begin{proof}
From (\ref{eq:def_L}), we find 
      \begin{equation*} \label{eq:G_N^3_first_decomp}
        \mathcal{G}_N^{(3)}=\frac{1}{\sqrt{N}} \sum_{\substack{p,q\in \Lambda_+^*\\p+q\ne0}} \widehat{V} (p/N)  e^{-B_\mu} \big(b^*_{p+q} b^*_{-p} b_q+\mathrm{h.c.}\big)e^{B_\mu}+\mathcal{E}_{3,1},
    \end{equation*}
    with
    \begin{equation} \label{eq:E_G_N^3}
        \mathcal{E}_{3,1} =\frac{1}{\sqrt{N}} \sum_{\substack{p,q \in\Lambda_+^*\\p+q\ne0}}\widehat{V} (p/N)  e^{-B_\mu}b^*_{p+q}a^*_{-p}\frac{\mathcal{N}_+}{N}a_q e^{B_\mu}+\mathrm{h.c.}.
    \end{equation}
 Let us first focus on the main term, we will show later that $\cE_{3,1}$ can be absorbed in the error $\cE_{\cG_N^3}$. With \eqref{eq:def_d}, we decompose  
    \begin{equation*} \label{eq:M_decomp}
        \frac{1}{\sqrt{N}} \sum_{\substack{p,q\in \Lambda_+^*\\p+q\ne0}} \widehat{V} (p/N)  e^{-B_\mu} b^*_{p+q} b^*_{-p} b_qe^{B_\mu}=M_0+M_1+M_2+M_3,
    \end{equation*}
    where, for $j=0,\dots ,3$, $M_j$ collects contributions with $j$ factors $d, d^*$, ie.  
    \begin{equation*}
    \begin{split}
        M_0=\;&\frac{1}{\sqrt{N}} \sum_{\substack{p,q\in\Lambda_+^*\\p+q\ne0}} \widehat{V} (p/N) 
      \big( \g_{p+q} b_{p+q}^* + \s_{p+q} b_{-p-q} \big) \big( \g_p b_{-p}^* + \s_p b_p \big) \big(\g_q b_q + \s_q b_{-q}^* \big)   
            \end{split}
    \end{equation*}
    and
    \begin{equation*}
        \begin{split}
            M_1=\;& \frac{1}{\sqrt{N}}\sum_{\substack{p,q \in\Lambda_+^*\\p+q\ne0}} \widehat{V} (p/N) \Big\{ \big( \gamma_{p+q} b^*_{p+q} + \s_{p+q} b_{-p-q} \big)\big( \gamma_p b_{-p}^* + \s_p b_p \big) d_q \\ &\hspace{2cm} + 
       \big( \gamma_{p+q} b^*_{p+q} + \s_{p+q} b_{-p-q} \big) d^*_{-p}  \big(\gamma_q b_q+\sigma_q b^*_{-q}\big) \\ &\hspace{2cm} + d_{p+q}^* \big( \gamma_p b_{-p}^* + \s_p b_p \big)  \big(\gamma_q b_q+\sigma_q b^*_{-q}\big)   \Big\}  \\
        M_2=\;&\frac{1}{\sqrt{N}}\sum_{\substack{p,q \in\Lambda_+^*\\p+q\ne0}} \widehat{V} (p/N) \Big\{ \big( \gamma_{p+q} b_{p+q}^* + \s_{p+q} b^*_{-p-q} \big) d_{-p}^* d_q \\ &\hspace{2cm} + d^*_{p+q}  \big( \gamma_p b_{-p}^* + \s_p b_p \big)  d_q + d_{p+q}^* d_{-p}^* \big(\gamma_q b_q+\sigma_q b^*_{-q}\big)  \Big\} 
\\
        M_3=\;&\frac{1}{\sqrt{N}}\sum_{\substack{p,q \in\Lambda_+^*\\p+q\ne0}} \widehat{V} (p/N) d^*_{p+q} d^*_{-p} d_q.
        \end{split}
    \end{equation*}
    Let us first consider $M_0$. Rearranging terms in normal order and noticing that all contributions arising from commutators are negligible (because, due to translation invariance, labels of creation operators  cannot coincide with labels of annihilation operators without violating the condition $p, q, p+q \not = 0$), we find 
    \begin{equation}\label{eq:M0}
        \begin{split} 
        M_0 = \; &\mathcal{C}_{\mathcal{G}_N} +  \frac{1}{\sqrt{N}}\sum_{\substack{p,q\in \Lambda_+^*\\p+q\ne0}}\widehat{V} (p/N)(\gamma_{p+q} \gamma_p-1)\,b^*_{p+q}b^*_{-p}(\gamma_qb_q+\sigma_qb^*_{-q})+\mathrm{h.c.}\\
        &+\frac{1}{\sqrt{N}}\sum_{\substack{p,q\in \Lambda_+^*\\p+q\ne0}}\widehat{V}(p/N) \gamma_{p+q} \sigma_p\, b^*_{p+q}(\gamma_qb_q+\sigma_qb^*_{-q})b_{p}+\mathrm{h.c.}\\
        &+\frac{1}{\sqrt{N}}\sum_{\substack{p,q\in \Lambda_+^*\\p+q\ne0}}\widehat{V}(p/N) \sigma_{p+q} \gamma_p\, b^*_{-p}(\gamma_qb_q+\sigma_qb^*_{-q})b_{-p-q}+\mathrm{h.c.}\\
        &+\frac{1}{\sqrt{N}}\sum_{\substack{p,q\in \Lambda_+^*\\p+q\ne0}}\widehat{V}(p/N) \sigma_{p+q} \sigma_p(\gamma_qb_q+\sigma_qb^*_{-q}) b_{p}b_{-p-q}+\mathrm{h.c.} + \mathcal{E}_{3,2}
        \end{split} 
    \end{equation}
    with $\pm \mathcal{E}_{3,2} \le C N^{-3/2} (\mathcal{N}_++1)^2$. Except for $\mathcal{C}_{\mathcal{G}_N}$, all the terms in $M_0$ can also be treated as error as they are bounded by $\eps \cN_+ +C\eps^{-1} N^{-1}(\cN_++1)^2$. In fact, considering for example one of the contributions in  the last line of (\ref{eq:M0}), we have
\[ \begin{split} 
\Big| \frac{1}{\sqrt{N}}\sum_{\substack{p,q\in \Lambda_+^*\\p+q\ne0}} &\widehat{V}(p/N) \sigma_{p+q} \sigma_p \gamma_q \langle \xi, b_q  b_{p}b_{-p-q} \xi \rangle \Big| \\ \leq \; &\frac{C}{\sqrt{N}}\sum_{\substack{p,q\in \Lambda_+^*\\p+q\ne0}} |\eta_p| |\eta_{p+q}|  \| b_q b_p b_{-p-q} (\cN_+ + 1)^{-1} \xi \| \| (\cN_+ + 1) \xi \| \\ \leq \; &\frac{C}{\sqrt{N}} \Big[ \sum_{\substack{p,q\in \Lambda_+^*\\p+q\ne0}}  \| b_q b_p (\cN_+ + 1)^{-1/2} \xi \|^2 \Big]^{1/2} \Big[ \sum_{\substack{p,q\in \Lambda_+^*\\p+q\ne0}} |\eta_{p+q}|^2 |\eta_p|^2 \Big]^{1/2} \\ \leq \; &C N^{-1/2}  \| \cN_+^{1/2} \xi \| \| (\cN_+ + 1) \xi \|. \end{split} \] 
    
Next, we show that the terms $M_1$, $M_2$, and $M_3$ are negligible. First, let us consider $M_1$. Terms with at least one $\gamma$-coefficient can be estimated by Cauchy-Schwarz, using (\ref{eq:decay_d}). For example, 
         \begin{equation*}
    \begin{split}
        \pm &\frac{1}{\sqrt{N}} \sum_{\substack{p,q \in\Lambda_+^*\\p+q\ne0}} \widehat{V} (p/N)  \gamma_{p+q} \sigma_q \big( \langle \xi,b^*_{p+q} d^*_{-p} b^*_{-q}\xi\rangle +\mathrm{h.c.} \big)\\
        \le\;&\frac{C}{\sqrt{N}} \Big[ \sum_{\substack{p,q \in\Lambda_+^*\\p+q\ne0}} \frac{|\widehat{V} (p/N)|}{|p+q|^2} \sigma_q^2\|(\mathcal{N}_++1)^{3/2}\xi\|^2\Big]^{\frac{1}{2}} \Big[\sum_{\substack{p,q \in\Lambda_+^*\\p+q\ne0}}|p+q|^2\|(\mathcal{N}_++1)^{-1}d_{-p}b_{p+q}\xi\|^2\Big]^{\frac{1}{2}}\\
        \le\;&C N^{-1} \langle \xi, (\mathcal{K} + \cN^2_+ + 1) (\cN_+ + 1) \xi\rangle \, .
    \end{split}
    \end{equation*}
    The term proportional to $b_{-p-q} b_p d_q$ can be handled similarly, estimating $\| b_{-p-q}^* (\cN_+ + 1) \xi \| \leq \| (\cN_+ + 1)^{3/2} \xi \|$ and using the factor $\sigma_{p+q}$ to sum over $q$. The terms proportional to $b_{-p-q} d^*_{-p} b^*_{-q}$ or $d^*_{p+q} b_p b^*_{-q}$ are slightly more challenging, because we prefer to avoid commutators between $b$ and $d$ operators. Still, using (\ref{eq:decay_d}) (and the smallness of $[ b_{-p}, b^*_{-p-q}]$, for $q \not = 0$) we can estimate  
    \begin{equation*}
        \begin{split}
            \pm &\frac{1}{\sqrt{N}}\sum_{\substack{p,q \in\Lambda_+^*\\p+q\ne0}} \widehat{V} (p/N)  \sigma_{p+q} \sigma_{q} \big( \langle \xi, b_{-p-q} d^*_{-p} b^*_{-q} \xi\rangle+\mathrm{h.c.}\big)\\
            \le\;&\frac{C \|(\mathcal{N}_++1)^{3/2}\xi\|}{N^{3/2}} \sum_{\substack{p,q \in\Lambda_+^*\\p+q\ne0}} |\widehat{V} (p/N) | |\sigma_{p+q}| |\sigma_{q}| \Big[ |\eta_p|\|(\mathcal{N}_++1)^{1/2}b^*_{-p-q}\xi\|+\|b_{-p} b^*_{-p-q}\xi\|\Big] \\
            \le \;&
        C N^{-1} \langle \xi, (\mathcal{K} + \cN^2_+ + 1) (\cN_+ + 1) \xi\rangle    
                   \end{split}
    \end{equation*}
and similarly for the term proportional to $d^*_{p+q} b_p b^*_{-q}$. Thus 
    \begin{equation*}
        \pm (M_1+\mathrm{h.c.}) \le C N^{-1} (\mathcal{K} + \cN^2_+ + 1) (\cN_+ + 1) \,.
    \end{equation*}
    As for $M_2$, it follows from \cite[Eq. (7.32)]{BBCS} that
    \begin{equation*}
        \pm (M_2+\mathrm{h.c.}) \le C N^{-1} (\mathcal{V}_N + \cN_+ + 1) (\cN_+ + 1) \, .
    \end{equation*}
    To bound $M_3$, we switch to position space. With \eqref{eq:decay_d_position} and using (\ref{eq:norms_eta}) to show $\|\check{\eta}\|_\infty\le CN$, we obtain
    \begin{equation*}
    \begin{split}
        |\langle \xi, (M_3+\mathrm{h.c.})\xi\rangle| \le\;&  \int dx dy N^{5/2}V(N(x-y))\,\|(\mathcal{N}_++1)^{-1} \check{d}_x \check{d}_y\xi\|\, \|(\mathcal{N}_++1) \check{d}_x\xi\|\\
        \le\;& \frac{C}{N^3}\int dx dy N^{5/2}V(N(x-y))\Big(\|(\mathcal{N}_++1)^{5/2}\xi\|+\|\check{a}_x(\mathcal{N}_++1)^2\xi\|\Big)\\
        &\qquad\times\Big(N\|(\mathcal{N}_++1)\xi\|+\|\check{a}_x(\mathcal{N}_++1)^{3/2}\xi\|+\|\check{a}_x\check{a}_y(\mathcal{N}_++1)\xi\|\Big)\\
        \le\;& C N^{-3/2} \langle \xi, \big(\mathcal{V}_N + \cN_+^2 + 1) (\cN_+ + 1) \xi\rangle.
    \end{split}
    \end{equation*}
    Finally, let us get back to the term $\cE_{3,1}$, from (\ref{eq:E_G_N^3}). Using Lemma  \ref{lemma:a_priori_quadratic}, we can write 
    \[ \cE_{3,1} = \frac{1}{N^{3/2}} \sum_{\substack{p,q,r \in \L^*_+ \\ p+q \not = 0}} \widehat{V} (p/N) e^{-B_\mu} b^*_{p+q} b^*_{-p} e^{B_\mu} e^{-B_\mu} \cN_+ e^{B_\mu}  e^{-B_\mu} b_q e^{B_\mu}  + \cE_{3,3} \]
    where $\pm \cE_{3,3} \leq C N^{-1} (\cH_N + 1) (\cN_+ + 1)$. Now, we expand $ e^{-B_\mu} b^*_{p+q} b^*_{-p} e^{B_\mu} $ and, on the other side, $e^{-B_\mu} b_{q} e^{B_\mu}$ using (\ref{eq:def_d}). The resulting terms can be controlled as above; by (\ref{eq:control-eBmu}), the additional factor $e^{-B_\mu} \cN_+ e^{B_\mu}$ does not affect the estimates (when applying Cauchy-Schwarz, it is however important not to act with the kinetic energy operator $\cK$ on $e^{-B_\mu} \cN_+ e^{B_\mu}$). At the end, the main contribution has the form $\cC_{\cG_N} e^{-B_\mu} \cN_+ e^{B_\mu}/N$ (with $\cC_{\cG_N}$ as in (\ref{eq:cC_G_N})) and can be bounded, using repeatedly \eqref{eq:control-eBmu}, by 
    \begin{equation*}
    \begin{split}
        \pm \frac{1}{N}\Big\langle \xi, &\, \mathcal{C}_{\mathcal{G}_N} e^{-B_\mu}\mathcal{N}_+e^{B_\mu} \xi\Big\rangle \\ \le\;& \frac{1}{N^{3/2}} \bigg( \sum_{\substack{p,q\in\Lambda_+^*\\p+q \ne0}}\frac{|\widehat{V} (p/N)| }{p^2} \big[ \|b_q e^{-B_\mu}\mathcal{N}_+e^{B_\mu} \xi\|^2 + \s_q^2 \| \cN_+^{1/2} e^{-B_\mu}\mathcal{N}_+e^{B_\mu}\xi \|^2 \bigg)^{1/2}\\
        &\qquad\times \bigg(\sum_{\substack{p,q\in\Lambda_+^*\\p+q \ne0}} p^2 \|b_{-p} b_{p+q}\xi\|^2\bigg)^{1/2}\\    \le\;&C N^{-1} \langle\xi, (\mathcal{K}+ \cN_+^2 + 1) (\cN_+ +1) \xi\rangle.
    \end{split}
    \end{equation*}
We conclude that $\pm \cE_{3,1} \leq C N^{-1} (\cK + \cN_+^2 + 1) (\cN_+ +1)$. Together with (\ref{eq:M0}) and with the bounds for $M_1, M_2, M_3$, this concludes the proof of the proposition. 
\end{proof}

Finally, we consider the action of $B_\mu$ on the operator $\cV_N$, defined in (\ref{eq:K,V}). 
\begin{prop} \label{prop:quad_on_quartic} 
Under the same assumptions of Theorem \ref{thm:quadratic} we have
\begin{equation}
\label{eq:cG4}
\begin{split}
\mathcal{G}^{(4)}_N=\;&\mathcal{V}_N+\frac{1}{2N} \sum_{p,q \in \L^*_+} \widehat V ((p-q)/N)  \s_p\g_p\s_q\g_q \Big( 1+ \frac{1}{N} - 2\frac{\cN_+}{N}\Big) \\
& + \frac{1}{2N}\sum_{p,q  \in \L^*_+} \widehat V ((p-q)/N)  \eta_q\Big[ \g_p^2b^*_pb^*_{-p}\Big(1+\frac{1}{N}-\frac{\cN_+}{N}\Big) + 2 \g_p\s_p b^*_pb_p + \s_p^2 b_pb_{-p}\\
 & \hskip4.5cm+ b_{-p}d_p+ d_p(b_{-p} +\eta_pb^*_p) +\hc   \Big]\\
 &+\frac{1}{N^2}\sum_{p,q,u\in \Lambda_+^*} \widehat{V} ((p-q)/N) \eta_p\eta_q\Big[ (\gamma_u^2+\sigma_u^2)b^*_u b_u+\gamma_u \sigma_u(b_ub_{-u}+b^*_ub^*_{-u})+\sigma_u^2 \Big]\\
 &+ \frac{1}{2N}\sum_{p, q \in \L^*_+, r \in \L^*}\widehat{V} (r/N) \gamma_{p+r}\gamma_{q}\sigma_{p}\sigma_{q+r}\,b^*_{p+r}b^*_{q}b^*_{-p}b^*_{-q-r} + \hc\\
 & + \frac{1}{N^2} \sum_{p,q,u\in \Lambda^*_+} \widehat{V} ((p-q)/N)  \eta_q b^*_pb^*_{-p} \Big[(\gamma_u\sigma_u b^*_ub^*_{-u})+\sigma_u^2 \Big]+\mathrm{h.c.}\\
 & + \cE_{\cG_N^{4}}
\end{split}
\end{equation}
with
\begin{equation*}
    \pm \cE_{\cG_N^{4}} \leq \eps \cN_++\frac{C}{\eps N}(\cH_N +\cN^2_+ +1)(\cN_++1).
\end{equation*}
\end{prop}
{\it Remark.} Compared with the corresponding result in \cite{BBCS}, here we also keep track of the terms appearing in the fifth and sixth line of \eqref{eq:cG4}. These terms will be included in the operator $\mathcal{T}_{\mathcal{G}_N}$ defined in \eqref{eq:cT_G_N} and will be shown to be negligible after cubic conjugation.

\begin{proof}
Proceeding as in the proof of \cite[Lemma 7.4]{BBCS}, we find (with the notation $\sum^*_{p,q,r} = \sum_{p,q \in \Lambda^*_+, r \in \Lambda^* : r \not = -p , -q}$) 
      \begin{equation}\label{eq:cG41}
        \begin{split}
            \mathcal{G}_N^{(4)}=\;&\frac{N+1}{2N^2}\sum^* \widehat{V} (r/N) e^{-B_\mu}b^*_{p+r} b^*_qb_pb_{q+r} e^{B_\mu}\\
            &+\frac{1}{N^2}\sum^* \widehat{V} (r/N) e^{-B_\mu}b^*_{p+r} b^*_q b^*_u b_u b_{p} b_{q+r}e^{B_\mu}+\mathcal{E}_{4,1},
        \end{split}
    \end{equation}
    with $\pm \mathcal{E}_{4,1} \le CN^{-1}(\mathcal{V}_N+\mathcal{N}_++1)(\mathcal{N}_++1)$. Again, following \cite[Lemma 7.4]{BBCS}, we can write the first term on the r.h.s. as 
 \begin{equation}\label{eq:e42} 
 \frac{N+1}{2N^2}\sum^* \widehat{V} (r/N) e^{-B_\mu}b^*_{p+r} b^*_qb_pb_{q+r} e^{B_\mu}  = V_0 + V_1 + \cE_{4,2} \end{equation} 
where 
    \begin{equation*}
    \begin{split}
        V_0=\;&\frac{N+1}{2N^2}\sum^* \widehat{V} (r/N) \Big[\gamma_{p+r}\gamma_qb^*_{p+r} b^*_q+\gamma_{p+r}\sigma_qb^*_{p+r}b_{-q}+\sigma_{p+r}\sigma_qb_{-p-r}b_{-q}\\
        &\qquad\qquad\qquad\qquad\qquad\;+\sigma_{p+r}\gamma_q(b^*_q b_{-p-r}-N^{-1}a^*_qa_{-p-r})\Big]\\
        &\qquad\qquad\qquad\qquad\quad\times\Big[\sigma_p\sigma_{q+r}\,b^*_{-p} b^*_{-q-r}+\sigma_p\gamma_{q+r}\,b^*_{-p}b_{q+r}+\gamma_p\gamma_{q+r}\,b_pb_{q+r}\\
        &\qquad\qquad\qquad\qquad\qquad\;+\gamma_p\sigma_{q+r}(b^*_{-q-r}b_p-N^{-1}a^*_{-q-r}a_p)\Big]\\
        &+\frac{N+1}{2N^2}\sum_{p,q\in\Lambda_+^*}\widehat{V} ((p-q)/N) \gamma_q \sigma_q \Big[ \big(\gamma_p^2 b^*_pb^*_{-p}+2\gamma_p \sigma_pb^*_p b_p- N^{-1}\gamma_p \sigma_pa^*_p a_p\\
        &\qquad\qquad\qquad\qquad\qquad\qquad\qquad+\sigma_p^2b_pb_{-p}\big)\big(1-\tfrac{\mathcal{N}_+}{N}\big)+\mathrm{h.c.} \Big]\\
        &+\frac{N+1}{2N^2}\sum_{p,q\in\Lambda_+^*}\widehat{V} ((p-q)/N)  \gamma_p\sigma_p \gamma_q\sigma_q\big(1-\tfrac{\mathcal{N}_+}{N}\big)^2, \\
     V_1 =\; &\frac{1}{2N} \sum_{p,q \in\Lambda_+^*} \widehat{V}\Big( \frac{p-q}{N} \Big)\gamma_q\sigma_q \big[d_p(\gamma_pb_{-p}+\sigma_pb^*_p)+(\gamma_pb_p+\sigma_p b^*_{-p})d_{-p}\big]+\mathrm{h.c.} \end{split} 
    \end{equation*}
  and $\pm\mathcal{E}_{4,2} \le CN^{-1}(\mathcal{V}_N+\mathcal{N}_++1)(\mathcal{N}_++1)$ (this error includes also the terms $V_{12}, V_{13}$ in the proof of \cite[Lemma 7.4]{BBCS}). Considering separately quartic, quadratic and constant contributions to $V_0$, we find
\begin{equation} \label{eq:V0deco} \begin{split} V_0 = \; &\cV_N + \frac{1}{2N} \sum_{r \in \L^*, p,q \in \L^*_+} \widehat V (r/N) \g_{p+r}\g_{q}\s_p\s_{q+r} b^*_{p+r}b^*_{q}b^*_{-p}b^*_{-q-r} +\hc\\ 
        &\,+ \frac{1}{N} \sum_{r \in \L^*, p,q \in \L^*_+} \widehat V (r/N) \g_{p+r}\g_{q}\s_p\g_{q+r} b^*_{p+r}b^*_{q}b^*_{-p}b_{q+r} +\hc\\
        &\, + \frac{1}{N} \sum_{r \in \L^*, p,q \in \L^*_+} \widehat V (r/N) (\g_{p+r}\g_{q}-1)b^*_{p+r}b^*_{q}b_{p}b_{q+r} +\hc \\ 
  &\, + \frac{1}{2N}\sum_{\substack{p,q\in\Lambda_+^*}}\widehat{V} ((p-q)/N)  \g_q\s_q \\  &\hspace{2cm} \times \Big[\g_p^2b^*_pb^*_{-p}\left(1+\frac{1}{N} -\frac{\cN_+}{N}\right) +2\g_p\s_pb^*_pb_p + \s_p^2b_pb_{-p}+\hc \Big]  \\
  &\, +\frac{1}{2N}\sum_{\substack{p,q\in\Lambda_+^*}}\widehat{V} ((p-q)/N)  \s_q\g_q\s_p\g_p\left(1+ \frac{1}{N}-2\frac{\cN_+}{N}\right) + \cE_{4,3}\end{split} \end{equation} 
  where $\pm\cE_{4,3} \leq \, C N^{-1} (\cH_N+\cN_++1)(\cN_++1)$. Additionally, except for $\cV_N$, the quartic terms containing both creation and annihilation operators can be considered as errors and are bounded by $\eps\cN_+ + C\eps^{-1}N^{-1}\cK(\cN_++1)^2$ as in (\ref{eq:example_bound_E_2_2}). For the term $V_1$, proceeding as for (\ref{eq:E_V_2}), we get   
  \begin{equation}
      \label{eq:V_1}
       V_1 = \frac{1}{2N}\sum_{p,q  \in \L^*_+} \widehat V ((p-q)/N)  \sigma_q\gamma_q\big[(b_{-p}d_p+ d_p(b_{-p} +\eta_pb^*_p)\big] +\hc   
  \end{equation}
  Let us now consider the second term on the r.h.s. of (\ref{eq:cG41}). Proceeding as in the proof of \cite[Lemma 7.4]{BBCS}, we find 
  \begin{equation}\label{eq:W0W1} \begin{split} \frac{1}{N^2} \sum^* \; &\widehat{V} (r/N) e^{-B_\mu}b^*_{p+r} b^*_q b^*_u b_u b_{p} b_{q+r}e^{B_\mu} \\
 = \; & \frac{1}{N^2} \sum_{p,q,u \in \L^*_+} \widehat V ((p-q)/N) \s_q\g_q\s_p\g_p\Big[ (\gamma_u^2+\sigma_u^2)b^*_u b_u+\gamma_u \sigma_u(b_ub_{-u}+b^*_ub^*_{-u}) +\sigma_u^2 \Big] \\
 &+ W_1 + \cE_{4,4} \end{split} \end{equation} 
 where 
 \begin{equation*} \begin{split} W_1 =  \; &\frac{1}{N^2}  \sum_{p,q,u \in \L^*_+} \widehat V ((p-q)/N) \s_q\g_q \\ &\times \Big[ \g_p^2b^*_pb^*_{-p} + 2 \g_p\s_p b^*_pb_p - N^{-1}\g_p\s_pa^*_pa_p +\s_p^2b_pb_{-p} +\g_pb^*_{-p}d^*_p +\s_pb_pd^*_p \\ &\hspace{1cm} +\g_pd^*_{-p}b^*_p +\s_pd^*_pb_p + d^*_{p}d^*_{-p}\Big] \big(e^{-B}b^*_ub_ue^B\big) (1-\cN_+/N) +\hc + \cE_{4,4} \end{split} \end{equation*} 
and $\pm \cE_{4,4} \leq C N^{-1} (\cV_N +\cN_+ +1)(\cN_++1)$ (the first line on the r.h.s. of (\ref{eq:W0W1}) corresponds to the term $W_0$ in the proof of \cite[Lemma 7.4]{BBCS}; 
the term $W_2$ is absorbed here into the error $\cE_{4,4}$). We decompose
     \begin{equation}\label{eq:W1}
    \begin{split}
        W_1 = &\, \frac{1}{N^2}  \sum_{p,q,u \in \L^*_+} \widehat V ((p-q)/N) \s_q\g_q \g_p^2\, b^*_pb^*_{-p}e^{-B_\mu}b^*_ub_ue^{B_\mu} +\hc +W_{11}\\
        =&\, \frac{1}{N^2}  \sum_{p,q,u \in \L^*_+} \widehat V ((p-q)/N) \s_q\g_q \g_p^2\, b^*_pb^*_{-p}(\g_ub^*_u +\s_u b_{-u})(\g_ub_u + \s_u b^*_{-u}) +\hc\\
        &+W_{11} +W_{12}\,.
    \end{split}
    \end{equation}
Furthermore, we can write   \[
    \begin{split}
        W_{11} = &\, \frac{1}{N^2}  \sum_{p,q  \in \L^*_+} \widehat V ((p-q)/N) \s_q\g_q \big(  2 \g_p\s_p b^*_pb_p - N^{-1}\g_p\s_pa^*_pa_p +\s_p^2b_pb_{-p}\big) \\
       &\hspace{3cm} \times e^{-B_\mu}\cN_+\big(1-(\mathcal{N}_+-1)/N\big)e^{B_\mu} (1- \cN_+/N) +\hc\\
       &\, - \frac{1}{N^3}  \sum_{p,q \in \L^*_+} \widehat V ((p-q)/N) \s_q\g_q \g_p^2b^*_pb^*_{-p} e^{-B_\mu}\cN_+\big(1-(\mathcal{N}_+-1)/N\big) e^{B_\mu} \cN_++\hc\\
       &\, + \frac{1}{N^2}  \sum_{p,q \in \L^*_+} \widehat V ((p-q)/N) \s_q\g_q \big(  
       \g_pb^*_{-p}d^*_p +\s_pb_pd^*_p +\g_pd^*_{-p}b^*_p +\s_pd^*_pb_p + d^*_{p+q}d^*_q\big)\\
       &\hspace{3cm} \times e^{-B_\mu} \cN_+\big(1-(\mathcal{N}_+-1)/N\big) e^{B_\mu} (1- \cN_+/N)  +\hc\\
       =&\, W_{111}+W_{112} +W_{113}.
    \end{split}
    \]
    With \eqref{eq:control-eBmu}, we can bound 
    \[
    \begin{split}
       |\langle\xi, W_{111}\xi\rangle| \leq &\, \frac{C}{N^2} \sum_{p,q  \in \L^*_+} |\widehat V ((p-q)/N)| |\mu_q| \|b_p  (\cN_+ + 1)^{1/2} \xi \| \\ & \hspace{3cm} \times \| b_p (\cN_+ + 1)^{-1/2} e^{-B_\eta} \cN_+ e^{B_\eta} (1-\cN_+ /N)  \xi\| \\ 
       &+ \frac{C}{N^2} \sum_{p,q  \in \L^*_+} |\widehat V ((p-q)/N)| |\mu_q|  |\s_p|^2 \| (\cN_+ + 1) \xi \|^2 
       \\  \leq \; &\frac{C}{N} \langle \xi, (\cN_+ + 1)^2 \xi \rangle .
    \end{split}
    \]
    We control $W_{112}$ in position space, again with the help of \eqref{eq:control-eBmu}. We find 
    \[
    \pm W_{112} \leq CN^{-3/2} (\cV_N+\cN_++1)(\cN_++1).
    \]
    As for $W_{113}$, we partially switch to position space and use Eq. \eqref{eq:decay_d}, \eqref{eq:decay_d_position}, the bound $\|\check{\s}*\check{\g}\|_\infty \leq CN$ and the inequality 
    \[
    \frac{C}{N^2}\sum_{p,q \in \L^*_+} |\widehat V ((p-q)/N)| |\mu_p||\mu_q| \leq C < \infty 
    \]
    to estimate 
    \[
    \begin{split}
       |\langle\xi, W_{113} \xi\rangle| \leq &\, \frac{C}{N^3}\sum_{p,q \in \L^*_+} |\widehat V((p-q)/N)| |\mu_q|\Big[|\mu_p|\|(\cN_++1)^{3/2}b_p\xi\| +\|b_pb_{-p}(\cN_++1)\xi\| \\
       &\hskip3cm +|\mu_p|^2\|(\cN_++1)^{3/2}\xi\|\| + |\s_p|\big(\|b^*_pb_p(\cN_++1)\xi\| \\
       &\hskip3cm+\|(\cN_++1)\xi\|\big)\Big]\|e^{-B}\cN_+e^B\xi\|\\
       &\, + \frac{C}{N^2} \int dxdy\, N^3 V(N(x-y)) |(\check{\s}*\check{\g})(x-y)|\\
       &\hskip2cm\times\Big[ \| \check{b}_x\check d_y\xi\| + \|r_x\|\|(\cN_++1)^{1/2}\check d_y\xi\| + \|\check d_x\check d_y\xi\|\Big]\|\cN_+e^B\xi\|\\
       \leq &\, \frac{C}{N}\langle\xi, \big(\cK + \cN_+ +1\big)(\cN_++1)\xi\rangle \\
       &\, + \frac{C}{N^2}\int dxdy\, N^3 V(N(x-y))
       \Big[ \|\check a_x(\cN_++1)^{3/2}\xi\| +  \|\check a_y(\cN_++1)^{3/2}\xi\|\\
       &\hskip1.5cm+ (N + |\check \mu(x-y)|)\|(\cN_++1)\xi\| + \|\check a_y\check a_x(\cN_++1)\xi\|\Big]\|\cN_+e^B\xi\|\\
       \leq &\, \frac{C}{N}\langle\xi, \big(\cK + \cV_N + \cN_+ +1\big)(\cN_++1)\xi\rangle\,.\\
    \end{split}
    \]
    As for the term $W_{12}$ on the r.h.s. of (\ref{eq:W1}), we find, with \eqref{eq:decay_d}, 
    \[\begin{split}
    | &\langle\xi, W_{12} \xi\rangle|\\  \leq &\, \frac{C}{N^2} \sum_{p,q,u \in \L^*_+} | \widehat V ((p-q)/N)| |\mu_q| \\ &\times \Big[\| (\cN_+ + 1)^{-\frac{3}{2}} d_ub_pb_{-p}\|\big \{ \|b_u (\cN_+ + 1)^{\frac{3}{2}} \xi\| + |\s_u|\|(\cN_++1)^{2} \xi\| + \| (\cN_+ + 1)^{\frac{3}{2}} d_u\xi \| \big\} \\
    &\hskip3cm + \big\{ \| b_u b_p b_{-p} (\cN_+ + 1)^{-1/2} \xi\| + |\s_u| \| b_pb_{-p}  \xi\| \big\}  \|(\cN_++1)^{1/2} d_u \xi\| \Big]  \\    \leq &\, \frac{C}{N^3} \sum_{p,q,u \in \L^*_+} |\widehat V ((p-q)/N) | |\mu_q| \\ &\times \Big[|\mu_u|^2 \| b_p b_{-p}\xi\| \|(\cN_++1)^{2}\xi\| + |\mu_u| \| b_p b_{-p}\xi\|\|(\cN_++1)^{3/2} b_u\xi\|\\
    &\hskip.2cm +\|b_ub_pb_{-p} (\cN_+ + 1)^{-\frac{1}{2}} \xi\|\|b_u (\cN_+ + 1)^{\frac{3}{2}} \xi\| + |\mu_u| \|b_ub_pb_{-p} (\cN_+ + 1)^{-\frac{1}{2}} \xi\|\|(\cN_++1)^{2}\xi\|\Big]\\
    \leq & \, \frac{C}{N^{3/2}} \|\cK^{1/2}(\cN_++1)^{1/2} \xi\|\|(\cN_++1)^2\xi\|\, \leq \frac{C}{N} \langle \xi, (\cK+ \cN^2_+ + 1) (\cN_+ + 1) \xi \rangle .
    \end{split}\] 
    Combining (\ref{eq:cG41}), \eqref{eq:e42}, (\ref{eq:V0deco}),(\ref{eq:V_1}), (\ref{eq:W0W1}), (\ref{eq:W1}) with the bounds for $W_{111}, W_{112}, W_{113}, W_{12}$ and using $|\g_q \s_q - \mu_q| \leq C |q|^{-6}$, we obtain (\ref{eq:cG4}). 
\end{proof}

We are now ready to conclude the proof of Theorem \ref{thm:quadratic}. 

\begin{proof}[Proof of Theorem \ref{thm:quadratic}] Collecting all terms linear in the $d,d^*$ operators from (\ref{eq:quadratic_on_K}), (\ref{eq:quad_on_L^2}) and from (\ref{eq:cG4}), we define
\begin{equation*}\label{eq:cDN}
    \begin{split}
    \cD_N= \, & \sum_{p \in \L^*_+} \Big[p^2\eta_p +\frac12 \widehat V (p/N) + \frac1{2N} \big(\widehat V (\cdot / N) *\eta \big)_p\Big] b_pd_{-p} +\hc\\
    & + \frac12 \sum_{p \in \L^*_+}\big(\widehat V (\cdot / N)*\widehat{f}_{N,\ell} \big)_p  \big[d_p(b_{-p} + \eta_pb^*_p) \big]+\mathrm{h.c.}.
    \end{split}
\end{equation*}
With the scattering equation (\ref{scattering}) and (\ref{eq:decay_d}), we obtain
\[ \cD_N = \frac{1}{2} \sum_{p \in \L^*_+} (\widehat{V} (\cdot /N) *\widehat{f}_{N,\ell})_p  [d_p(b_{-p} + \eta_p b_p^*)] + \text{h.c.} + \cE_1 \]
where $\pm \cE_1 \leq N^{-1} (\cK+ \cN_+ + 1)(\cN_+ +1)$. Handling the contribution proportional to $\eta_p$ as in \cite[Section 7.5]{BBCS} (where the contribution is labeled $D_{33}$), and using \eqref{eq:d_detailed_expansion} to expand the rest of the term, we find 
\[ \begin{split} \cD_N  = \; &-\frac{1}{2N}\sum_{\substack{p\in\Lambda^*\\q\in\Lambda_+^*}}\big(\widehat{V}(\cdot / N)*\widehat{f}_{N,\ell} \big)_p\eta_p\big[\gamma_q\sigma_q(b_qb_{-q}+b^*_q b^*_{-q})+(\sigma_q^2+\gamma_q^2-1)b^*_q b_q+\sigma_q^2\big] \\&- \frac{1}{2N} \sum_{p,q\in\Lambda_+^*} \big( \widehat{V} (\cdot / N) *\widehat{f}_{N,\ell} \big)_p \eta_q \\ &\qquad\qquad\times\int_0^1ds  \,\gamma_p^{(s)}\,e^{(-1+s)B_\mu} b^*_{-q} a^*_q a_p e^{(1-s)B_\mu}b_{-p}+\mathrm{h.c.} + \cE_2 
    \end{split}
\]
where $\pm \cE_2 \leq CN^{-1} (\cN_+ + 1)^2$. Next, we compute the action of $(1-s) B_\mu$ on $b_q^* b_{-q}^* b_{p}$; with (\ref{eq:def_d}), we find 
\begin{equation*}
    \begin{split}
        \mathcal{D}_N=\;&-\frac{1}{2N}\sum_{\substack{p\in\Lambda^*\\q\in\Lambda_+^*}}\big(\widehat{V}( \cdot / N) *\widehat{f}_{N,\ell} \big)_p\eta_p \big[ \gamma_q\sigma_q(b_qb_{-q}+b^*_q b^*_{-q})+(\sigma_q^2+\gamma_q^2-1)b^*_q b_q+\sigma_q^2\big]\\
        &-\frac{1}{4N} \sum_{p,q\in\Lambda_+^*} \big(\widehat{V} (\cdot / N) *\widehat{f}_{N,\ell}\big)_p b^*_p b^*_{-p}\\
        &\qquad\qquad\times\Big[(-\mu_q+\gamma_q\sigma_q) b^*_q b^*_{-q}+(\mu_q+\gamma_q\sigma_q)  b_q b_{-q} +2\sigma_q^2 b^*_q b_{q} + \sigma_q^2\Big]+\mathrm{h.c.}\\
        &+\mathcal{E}_{3},
    \end{split}
\end{equation*}
with $\pm\mathcal{E}_{3} \le C N^{-1} (\mathcal{K} + \cN_+^2 + 1) (\cN_+ + 1)$. In a similar way to \eqref{eq:example_bound_E_2_2}, one can show that the contribution of the quartic terms which contain both creation and annihilation operators can be treated as errors bounded by $\eps \cN_++C\eps^{-1}N^{-1}\cK(\cN_++1)^2$. Combining this with all other contributions in (\ref{eq:quadratic_on_K}), (\ref{eq:quad_on_L^2}), (\ref{eq:cGN3})  and (\ref{eq:cG4}), we obtain 
\begin{equation}\label{eq:cGN-fin1} \begin{split} 
\cG_N = C_{\cG_N} + \sum_{p \in \L^*_+} \Big[ \Phi_p b_p^* b_p + \frac{1}{2}  \Gamma_p \big( b_p^* b_{-p}^* + b_p b_{-p} \big) \Big] + \cC_{\cG_N} + \cV_N + \cT_{\cG_N}  + \cE_{\cG_N} 
\end{split} \end{equation}
where $C_{\cG_N}, \cC_{\cG_N},\cT_{\cG_N}$ are defined as in (\ref{eq:C_G_N}), (\ref{eq:cC_G_N}), (\ref{eq:cT_G_N}), and 
\begin{equation} \label{eq:PhiGamma_p}
\begin{split}
    \Phi_p =  \;& p^2 (\g_p^2 + \s_p^2) + \widehat V (p/N) (\g_p +\s_p)^2 + \frac{2\g_p \s_p}{N} \big(\widehat V (\cdot /N)*\eta\big)_p  - \frac{(\g_p^2 + \s_p^2)}{N} \sum_{q \in \L^*}\widehat V (q/N)\eta_q \\  
      \Gamma_p = \;& 2 p^2\g_p\s_p + \widehat V (p/N) (\g_p +\s_p)^2 + \frac{(\g_p^2 +\s_p^2)}{N}   \big( \widehat V (\cdot /N) * \eta)_p  - \frac{2\g_p\s_p}{N} \sum_{q \in \L^*} \widehat V (q/N) \eta_q.
    \end{split}
\end{equation}
Moreover, 
\begin{equation}\label{eq:wtcEGN} \pm \cE_{\cG_N} \leq   \eps \cN_+ + \frac{C}{\eps N} (\cH_N + \cN_+^3 + 1) (\cN_+ + 1) . \end{equation}

To conclude the proof of Theorem \ref{thm:quadratic}, we consider the quadratic term on the r.h.s. of (\ref{eq:cGN-fin1}). Adding and subtracting the contributions that will arise from the cubic conjugation in Theorem \ref{thm:cubic}, we rewrite the coefficients in (\ref{eq:PhiGamma_p}) as 
 \[ \begin{split} \Phi_p &= F_p - \frac{(\g_p^2 + \s_p^2)}{N} \sum_{q \in \Lambda_+^*} \big( \widehat{V} (q/N) + \widehat{V} ((q+p)/N) \big) \eta_q \\ \Gamma_p &= G_p - \frac{2\gamma_p \sigma_p}{N} \sum_{q \in \L^*_+} \big( \widehat{V} (q/N) + \widehat{V} ((q+p)/N) \big) \eta_q \end{split} \] 
 where   
 \[ \begin{split} F_p &= p^2 (\gamma_p^2 + \s_p^2) + \big( \widehat{V} (\cdot /N) * \widehat{f}_{N,\ell} \big)_p (\g_p + \s_p)^2 \\
 G_p &= 2p^2 \g_p \s_p +  \big( \widehat{V} (\cdot /N) * \widehat{f}_{N,\ell} \big)_p (\g_p + \s_p)^2 .
 \end{split} \]
 Recalling $\g_p = \cosh \mu_p = \cosh (\eta_p + \tau_p)$, $\s_p = \sinh \mu_p = \sinh (\eta_p + \tau_p)$ and the definition (\ref{eq:deftau}) of the coefficients $\tau_p$, we obtain, 
\[ \begin{split} 
G_p &= \big[ p^2 + \big( \widehat{V} (\cdot /N) * \widehat{f}_{N,\ell} \big)_p \big] \sinh (2\tau_p + 2\eta_p) + \big( \widehat{V} (\cdot /N) * \widehat{f}_{N,\ell} \big)_p \cosh (2\tau_p + 2 \eta_p) \\ &= 
 \big[ p^2 + \big( \widehat{V} (\cdot /N) * \widehat{f}_{N,\ell} \big)_p \big]  \cosh (2\tau_p) (\sinh (2\eta_p) + \tanh (2\tau_p) \cosh (2\eta_p)) \\ &\hspace{3cm} + \big( \widehat{V} (\cdot /N) * \widehat{f}_{N,\ell} \big)_p \cosh (2\tau_p) (\cosh (2\eta_p) + \tanh (2\tau_p) \sinh (2\eta_p)) \end{split} \]
which implies that $G_p = 0$, and 
\begin{equation}\label{eq:F-compu} \begin{split} F_p &= \big[ p^2 + \big( \widehat{V} (\cdot /N) * \widehat{f}_{N,\ell} \big)_p \big] \cosh (2\tau_p + 2 \eta_p) + \big( \widehat{V} (\cdot /N) * \widehat{f}_{N,\ell} \big)_p \sinh (2\tau_p + 2 \eta_p) \\ &= \big[ p^2 + \big( \widehat{V} (\cdot /N) * \widehat{f}_{N,\ell} \big)_p \big]  \cosh (2\tau_p) (\cosh (2\eta_p) + \tanh (2\tau_p) \sinh (2\eta_p)) \\ &\hspace{1cm} + \big( \widehat{V} (\cdot /N) * \widehat{f}_{N,\ell} \big)_p \cosh (2\tau_p)  (\sinh (2\eta_p) + \tanh (2\tau_p) \cosh (2\eta_p))  
\end{split} \end{equation} 
 leading to $F_p = \sqrt{|p|^4 + 2 (\widehat{V} (\cdot / N) * \widehat{f}_{N,\ell})_p p^2}$. This concludes the proof of Theorem \ref{thm:quadratic}. 
\end{proof}

\section{Proof of Proposition \ref{prop:a_priori_A}} 
\label{sec:gron-A} 

The goal of this section is to show Prop. \ref{prop:a_priori_A}, controlling the growth of the number of excitations and of their energy with respect to cubic conjugation. To reach this goal, we are going to estimate the commutators of $\cN_+$ and of the Hamiltonian $\cH_N = \cK + \cV_N$ with the cubic operator $A$, as defined in (\ref{eq:def_A}). Since  in the proof of Theorem \ref{thm:cubic} we will need to compute $e^{-A} \cH_N e^A$, the next proposition contains precise estimates, which are not really needed in the proof of Prop. \ref{prop:a_priori_A}.
\begin{lemma} \label{lemma:K,V,A}
We have
\begin{equation}\label{eq:[K,A]1} 
 \big[ \cK, A \big] = [ \cK , A ]_1 +[ \cK , A ]_2+  \mathrm{h.c.}  \end{equation}
where
    \begin{align}\label{eq:[K,A]3} 
            [\mathcal{K},A]_1=\;&\frac{2}{\sqrt{N}}\sum_{\substack{r,v\in\Lambda_+^*\\r+v\ne0}  }r^2\eta_r \,b^*_{r+v}b^*_{-r}(\gamma_vb_{v}+\sigma_vb^*_{-v})\\ \label{eq:[K,A]2}[\mathcal{K},A]_2=\;&\frac{2}{\sqrt{N}}\sum_{\substack{r,v\in \Lambda_+^*\\r+v\ne0} }(r\cdot v)\eta_r\gamma_v\,b^*_{r+v}b^*_{-r}b_{v}\,.
            \end{align} 
Moreover, 
\begin{equation} \label{eq:[VN,A]1} \big[ \cV_N , A \big] = [ \cV_N , A]_1 + [\cV_N , A ]_2 + \hc +  \cE_{[\cV_N , A]}  \end{equation} 
where 
\begin{equation}  \label{eq:[VN,A]2}     \begin{split}   
             [\mathcal{V}_N,A]_1=\; &\frac{1}{N^{3/2}}\sum_{\substack{r,v,s \in \Lambda_+^*\\r+v,\,r+v+s\ne0}} \widehat{V} ((r-s)/N) \eta_s \,b^*_{r+v}b^*_{-r}(\gamma_v b_{v}+\sigma_vb^*_{-v}) \\
             [\mathcal{V}_N,A]_2=\;& 
             \frac{1}{N^{3/2}} \sum_{\substack{\;r,v,s\in \Lambda_+^*\\\,s+v,\,r+v\ne0}}\widehat{V}((r-s)/N) \eta_{s}\frac{2\eta_{s+v}s^2-2\sigma_v\,(v\cdot s)}{s^2+v^2+|s+v|^2}\,b^*_{r+v}b^*_{-r} b^*_{-v}
                   \end{split}
    \end{equation}
    and 
    \begin{equation}\label{eq:[VN,A]3} \begin{split} | \langle \xi, \cE_{[\cV_N , A]} \xi \rangle | \leq CN^{-1} \langle \xi , (\cH_N + \cN^2_+ + 1) (\cN_+^2 + 1) \xi \rangle .\end{split} \end{equation}  
    Furthermore, we have 
    \begin{equation}\label{eq:comm-gron} | \langle \xi_1, \big[ \cK + \cV_N , A \big] \xi_2 \rangle | \leq C \langle \xi_1, (\cH_N + \cN_+^2 +1) \xi_1 \rangle + \langle \xi_2, (\cH_N + \cN_+^2 + 1) \xi_2 \rangle \end{equation} 
  for  any $\xi_1,\xi_2 \in \cF^{\leq N}_+$. 
\end{lemma} 

\textit{Remark:} The coefficients $\nu_{r,v}$ entering the cubic operator $A$ are defined in (\ref{eq:def_wtsigma}) exactly so that the contribution to the commutator $[\cK,A]$  proportional to $b^*_{r+v} b^*_{-r} b^*_{-v}$ only enters in the term containing $r^2$ in \eqref{eq:[K,A]3} (and it is absent from the term containing $r\cdot v$  in (\ref{eq:[K,A]2})).
\medskip

To prove Lemma \ref{lemma:K,V,A}, we will use the following auxiliary lemma.  
\begin{lemma}
 For $r,v \in \Lambda^*_+$, we define the coefficients 
    \begin{equation}
\label{eq:defalpha}\begin{split}
\alpha_{r,v}:=&\, \frac{1}{N^{1/2}} \sum_{\substack{\;s\in \Lambda_+^*\\\,s+v\ne0}}\widehat{V}((r-s)/N) \eta_{s}\Big[\frac{2\eta_{s+v}s^2}{s^2+v^2+|s+v|^2} - \frac{2\sigma_v \,(v\cdot s)}{s^2+v^2+|s+v|^2}\Big]\, .
\end{split}\end{equation}
Then, we have 
\begin{equation}
\label{eq:boundalpha}
\begin{split}
 \sum_{v\in\Lambda_+^*} \sup_{r\in\Lambda_+^*} |\alpha_{r,v}|^2 + \frac{1}{N} &\sum_{v\in\Lambda_+^*} \left[ \sup_{z \in \Lambda^*} \sum_{r\in\Lambda_+^*} \frac{|\alpha_{r,v}|^2}{(r+z)^2} \right] + N^{1/2}\sup_{r\in\Lambda_+^*} |\alpha_{r,v}| |v|\leq C\,.
\end{split}
\end{equation}

\end{lemma}
\begin{proof}
We split $\alpha_{r,v} = \alpha^{(1)}_{r,v}+\alpha^{(2)}_{r,v}$, with $\alpha^{(1)}_{r,v}, \alpha^{(2)}_{r,v}$ indicating the contribution of the first, respectively, the second term in the square brackets. We will show \eqref{eq:boundalpha} separately, for $\alpha^{(1)}_{r,v}$ and $\alpha^{(2)}_{r,v}$. We have 
\begin{equation*}
\label{bound:f1_1}
\begin{split}
 \sum_{v} \sup_r |\alpha^{(1)}_{r,v}|^2&\leq \frac{C}{N}\sum_{s,s',v} |\eta_s| |\eta_{s+v}||\eta_{s'+v}||\eta_{s'}| \\ &\leq\frac{C}{N}\sum_{v,s} \frac{|\eta_s| |\eta_{s+v}|}{|v|} \leq \frac{C}{N} \sum_{|v|\leq N}\frac{1}{|v|^2}+\frac{C}{N^2}\sum_{s,v;\,|v|>N} |\eta_s| |\eta_{s+v}|\leq C,
\end{split}
\end{equation*}
Similarly, we obtain 
\begin{equation*}
\label{bound:f1_2}
\begin{split}
\frac{1}{N} \sum_{v} \left[ \sup_{z\in \Lambda^*}\sum_{r} \frac{|\alpha^{(1)}_{r,v}|^2}{(r+z)^2}\right]&\leq \frac{C}{N^2}\sum_{v} \sup_{z \in \Lambda^*}\sum_{s,s',r}\frac{|\widehat{V}((r-s)/N)|}{(r+z)^2}|\eta_s| |\eta_{s+v}||\eta_{s'+v}||\eta_{s'}|\\& \leq \frac{C}{N} \sum_{s,s',v} |\eta_s| |\eta_{s+v}||\eta_{s'+v}||\eta_{s'}|\leq C \, , 
\end{split}
\end{equation*}
using the bound $\| \widehat{V} (\cdot / N)\|_2\leq CN^{3/2}$, in the region $|r+z|>N$. The bound for the third summand on the left hand side of \eqref{eq:boundalpha} follows immediately from $\sum_s|\eta_s||\eta_{s+v}|\leq |v|^{-1}$.

To handle $\alpha^{(2)}_{r,v}$, we decompose $\alpha^{(2)}_{r,v} = \alpha^{(2, >)}_{r,v} + \alpha^{(2,<)}_{r,v}$ with 
\[ \begin{split} \alpha^{(2,>)}_{r,v} &= - \frac{1}{N^{1/2}} \sum_{\substack{\;s\in \Lambda_+^*, |s| > N \\\,s+v\ne0}} \widehat{V}((r-s)/N) \eta_{s} \, \frac{2\sigma_v \,(v\cdot s)}{s^2+v^2+|s+v|^2} \, , \\
\alpha^{(2,<)}_{r,v} &= - \frac{1}{N^{1/2}} \sum_{\substack{\;s\in \Lambda_+^*, |s| \leq N \\\,s+v\ne0}} \widehat{V}((r-s)/N) \eta_{s} \, \frac{2\sigma_v \,(v\cdot s)}{s^2+v^2+|s+v|^2}  \end{split} \]
Estimating $|\alpha^{(2,>)}_{r,v}| \leq |v| |\sigma_v| / \sqrt{N}$, it is easy to check that \eqref{eq:boundalpha} holds true, if we replace $\alpha_{r,v}$ with $\alpha^{(2,>)}_{r,v}$. 
Let us now consider the contribution of $\alpha^{(2,<)}_{r,v}$.  Through a change of variable $s\rightarrow-s$ we find
\begin{equation} \label{eq:alpha2}\begin{split} 
\alpha^{(2,<)}_{r,v} :=\; &\frac{1}{N^{1/2}} \sum_{|s|\leq N}\eta_s\sigma_v(s\cdot v) \left[\frac{\widehat{V}((r+s)/N)}{s^2+v^2+(s-v)^2}-\frac{\widehat{V}((r-s)/N)}{s^2+v^2+(s+v)^2}\right] \\ =\;&\frac{1}{N^{1/2}} \sum_{|s|\leq N}[\widehat{V}((r+s)/N)-\widehat{V}((r-s)/N)]\frac{\eta_s\sigma_v(s\cdot v)}{s^2+v^2+(s+v)^2}\\ &\qquad+\frac{1}{N^{1/2}}\sum_{|s|\leq N}\widehat{V}((r+s)/N)\frac{4\eta_s\sigma_v(s\cdot v)^2}{(s^2+v^2+(s+v)^2)(s^2+v^2+(s-v)^2)} \\ =\;&\alpha^{(2,<,1)}_{r,v} + \alpha^{(2,<,2)}_{r,v} 
.\end{split} 
\end{equation}
The contribution of $\alpha_{r,v}^{(2,<,2)}$ to (\ref{eq:boundalpha}) can be bounded easily, as we did with $\alpha^{(1)}_{r,v}$. Let us focus on the contribution of $\alpha_{r,v}^{(2,<,1)}$. From the Lipschitz continuity of $\widehat{V}$, we obtain 
\begin{equation}\label{eq:alpha2<1} \begin{split} 
\sum_{v} \sup_r |\alpha^{(2,<,1)}_{r,v}|^2&\leq \frac{C}{N^3}\sum_{v} |\sigma_v|^2 v^2\Big(\sum_{|s|\leq N} |\eta_s| \Big)^2 \leq C.
\end{split} \end{equation}
Similarly, we can also estimate the last term on the l.h.s. of \eqref{eq:boundalpha}. As for the second term in \eqref{eq:boundalpha}, it can be bounded by (\ref{eq:alpha2<1}), if we restrict the sum over $r \in \Lambda^*_+$ to momenta with $|r+z|\leq N$. For $|r+z|>N$, on the other hand, we switch to position space and get
\begin{equation}
\label{bound:f2_1}
\begin{split}
\frac{1}{N}& \sum_{v} \left[ \sup_z \sum_{r} \frac{|\alpha^{(2,<,1)}_{r,v}|^2}{(r+z)^2}\right]\leq \, \frac{C}{N^2}\sum_{v}|\sigma_v|^{2} \sup_{z} \left( \sum_{\substack{|r+z|>N}}\frac{1}{(r+z)^4}\right)^{1/2}  \\ & \times \Big(\sum_{\substack{r\\ |s_1|,|s_2|,\\ |s_3|,|s_4|\leq N}}\prod_{i=1,2,3,4}[\widehat{V}((r+s_i)/N)-\widehat{V}((r-s_i)/N)]\frac{\eta_{s_i}(s_i\cdot v)}{s_i^2+v^2+(s_i+v)^2} \Big)^{1/2}\\ &\leq \frac{C}{N}\sum_{v}|\sigma_v|^2 v^{2}\, \Big(\int_{x_1,x_2,x_2\in \Lambda} dx_1dx_2dx_3|V(x_1)||V(x_2)||V(x_3)||V(x_1+x_2+x_3)|\\
&\qquad \qquad  \times \sum_{\substack{ |s_1|,|s_2|,\\ |s_3|,|s_4|\leq N}} \, \prod_{i=1,2,3,4}\,\frac{ |\eta_{s_i}|}{N}\Big)^{1/2}\\
&\leq C \, \|V\cdot (V\ast V\ast V) \|_1 \leq  C
\end{split}
\end{equation}
where we used the bound $|e^{i(s\cdot x)/N}- e^{-i(s\cdot x)/N}|\leq C|s|/N$ for $x\in \Lambda$. 
\end{proof}

Now we are ready to prove Lemma \ref{lemma:K,V,A}. 
\begin{proof}[Proof of Lemma \ref{lemma:K,V,A}]
With the commutation relations (\ref{eq:CCR_b}) and the definition (\ref{eq:def_A}), we find
\[ \begin{split} \big[ \cK, A \big] = \; &\frac{1}{\sqrt{N}} \sum_{\substack{r,v \in \Lambda_+^* \\ r+v \not = 0}} \eta_r \gamma_v ((r+v)^2 +r^2 -v^2) b^*_{r+v} b^*_{-r} b_v \\ &+ \frac{1}{\sqrt{N}} \sum_{\substack{r,v \in \Lambda_+^* \\ r+v \not = 0}} \nu_{r,v} ((r+v)^2 +r^2 + v^2) b^*_{r+v} b^*_{-r} b^*_{-v} + \text{h.c.} \end{split} \]
Recalling (\ref{eq:def_wtsigma}), we immediately obtain (\ref{eq:[K,A]1}). 

A slightly longer (but still straightforward) computation shows that (\ref{eq:[VN,A]1}) holds true, with 
\begin{equation}\label{eq:cEVA} \begin{split} &\cE_{[\cV_N,A]} \\  =\; &\frac{1}{N^{3/2}} \sum_{\substack{r,v,s\in \Lambda_+^*\\r+v+s \ne0}}\! \widehat{V} ((r-s)/N) \Big[\frac{ 2s^2 \eta_s(\sigma_{s+v} -\eta_{s+v}) }{s^2+v^2+(s+v)^2} + \frac{2 v^2 [\eta_v(\sigma_s-\eta_s) - \eta_s(\sigma_v-\eta_v)]}{s^2+v^2+(s+v)^2}\Big]\\ 
& \hspace{8cm}\times b^*_{r+v}b^*_{-r}b^*_{-v}+\mathrm{h.c.} \\ 
& + \frac{1}{N^{3/2}} \sum_{\substack{u\in \Lambda^*,\;p,r,v\in\Lambda_+^*\\v-u,r+v-u,p+u\ne0}}\widehat{V} (u/N)\eta_r \,(\gamma_{v}-\g_{v-u})\, b^*_{r+v-u}b^*_{-r}a^*_{p+u}a_pb_v+\mathrm{h.c.}\\ &+ \frac{1}{N^{3/2}} \sum_{\substack{u\in \Lambda^*,\;p,r,v\in\Lambda_+^*\\r+v,r+u, p+u\ne0}}\widehat{V} (u/N) \eta_r \,\gamma_{v}\, b^*_{r+v}b^*_{-r-u}a^*_{p+u}a_pb_v+\mathrm{h.c.} \\ &+ \frac{1}{N^{3/2}} \sum_{\substack{u\in \Lambda^*,\;p,r,v\in\Lambda_+^*\\v-u\ne0\\r+v-u,p+u\ne0}}\widehat{V} (u/N)\big[\nu_{r,v} +\nu_{r,v-u}+ \nu_{r-u,v}\big] b^*_{r+v-u}b^*_{-r}b^*_{-v}a^*_{p+u}a_p+\mathrm{h.c.} 
 \end{split} \end{equation} 
 
 Using that $|\s_s - \eta_s| \leq C |\tau_s|$, we can bound the first term by
 \[\begin{split} \Big| \frac{1}{N^{3/2}} \sum_{\substack{r,v,s\in \Lambda_+^*\\r+v+s \ne0}} \widehat{V} ((r+s)/N) \Big[\frac{ 2s^2 \eta_s(\sigma_{s+v} -\eta_{s+v}) }{s^2+v^2+(s+v)^2} + \frac{2 v^2 [\eta_v(\sigma_s-\eta_s) - \eta_s(\sigma_v-\eta_v)]}{s^2+v^2+(s+v)^2}\Big] \\
 \times \langle \xi, b^*_{r+v}b^*_{-r}b^*_{-v} \xi \rangle \Big| \leq \frac{C}{N} \| \cK^{1/2} \xi \| \| (\cN_+ + 1) \xi \|  \, . \end{split}\]  
As for the second term on the r.h.s. of (\ref{eq:cEVA}), we observe that $|\gamma_v - \g_{v-u}| \leq C (\eta_v^2 + \eta_{v-u}^2)$. We obtain  
\[ \begin{split} \Big| \frac{1}{N^{3/2}} \sum_{\substack{u\in \Lambda^*,\;p,r,v\in\Lambda_+^*\\v-u,r+v-u,p+u\ne0}}\widehat{V} (u/N)\eta_r \,(\gamma_{v}-\g_{v-u})\, &\langle \xi, b^*_{r+v-u}b^*_{-r}a^*_{p+u}a_pb_v \xi \rangle \Big| \\  &\leq \frac{C}{N} \| \cK^{1/2} (\cN_+ + 1)^{1/2} \xi \| \| (\cN_+ + 1)^{3/2} \xi \| \, . \end{split} \] 
To bound the third term, we switch to position space. We find 
\[
\begin{split}
\Big| \frac{1}{N^{3/2}} &\sum_{\substack{u\in \Lambda^*,\;p,r,v\in\Lambda_+^*\\r+v,r+u, p+u\ne0}}\widehat{V} (u/N) \eta_r \,\gamma_{v}\, \langle \xi, b^*_{r+v}b^*_{-r-u}a^*_{p+u}a_pb_v \xi \rangle \Big|  \\
\leq  &\, \frac{1}{N^{1/2}} \int dx\, dy\, dz\,  N^{2}V(N(x-z))|\check\eta (y-z)||\langle \xi, \check b^*_y\check b^*_z \check a^*_x\check a_x b(\check \g_y)\xi\rangle| \\
    \leq &\, \frac{C}{N} \Big[\int dx\, dy\, dz\,  N^{2}V(N(x-z))\|\check b_y\check b_z \check a_x \xi\|^2\Big]^{1/2} \\
    &\hspace{2cm}\times\Big[\int dx\, dy\, dz\,  N^{3}V(N(x-z))|\check\eta (y-z)|^2 \|\check a_x  b (\check{\gamma}_y) \xi\|^2\Big]^{1/2}\\
    \leq &\, \frac{C}{N} \| \cV_N^{1/2} (\cN_+ + 1)^{1/2} \xi \| \| \cK^{1/2} (\cN_+ + 1)^{1/2} \xi \| \end{split} \] 
  where, in the last step, we used $|\check{\eta} (y-z)| \leq C / |y-z|$ (see \eqref{eq:defeta} and \eqref{eq:decay_w}) and Hardy's inequality (and the estimate $\| \check a_x  b (\check{\gamma}_y) \xi\| \leq \| \check a_x  \check a_y \xi\| + \| \check a_x \cN_+^{1/2} \xi \|$). Finally, let us consider the last term on the r.h.s. of (\ref{eq:cEVA}). We only show how to bound the contribution proportional to $\nu_{r,v}$, the others can be estimated similarly. With $|\nu_{r,v}|\leq C|r|^{-2}|v|^{-2}$, we obtain
  \[\begin{split}
    &\Big|\frac{1}{N^{3/2}} \sum_{\substack{u\in \Lambda^*,\;p,r,v\in\Lambda_+^*\\v-u\ne0\\r+v-u,p+u\ne0}}\widehat{V} (u/N) \nu_{r,v} \, \langle \xi,  b^*_{r+v-u}b^*_{-r}b^*_{-v}a^*_{p+u}a_p\xi\rangle\Big|\\ 
    &\quad \leq \frac{C}{N^{3/2}} \Big[\sum_{\substack{u\in \Lambda^*,\;p,r,v\in\Lambda_+^*\\v-u\ne0\\r+v-u,p+u\ne0}}\frac{|\widehat{V} (u/N)|}{|p+u|^2} \frac{1}{|r|^4|v|^4} \|(\cN_++1)a_p\xi\|^2\Big]^{1/2}\\
    & \quad\quad \quad\quad \times \Big[\sum_{\substack{u\in \Lambda^*,\;p,r,v\in\Lambda_+^*\\v-u\ne0\\r+v-u,p+u\ne0}} |p+u|^2 \|(\cN_++1)^{-1}a_{p+u}a_{-v}a_{-r}a_{r+v-u}\xi\|^2\Big]^{1/2}\\
    &\quad \leq \frac{C}{N} \|\cK^{1/2}(\cN_++1)^{1/2}\xi\|\|(\cN_++1)^{3/2}\xi\|\,.
  \end{split}
\]
This concludes the proof of (\ref{eq:[VN,A]3}). 
  
Proceeding as we did above, it is also simple to verify that the error term $\cE_{[\cV_N , A]}$ satisfy (\ref{eq:comm-gron}). To conclude the proof of the bound (\ref{eq:comm-gron}), we observe, first of all, that 
 \[
|\langle \xi_1 , [\cK,A]_1 \xi_2 \rangle| \leq C \|\cK^{1/2} \xi_1 \| \|(\cN_++1) \xi_2 \| 
\]
and 
\begin{equation}\label{eq:bound[K,A]2} 
\begin{split}
  |\langle \xi_1 , &[\cK,A]_2 \xi_2 \rangle | \\  \leq & \, \frac{C}{N^{1/2}} \sum_{\substack{r,v\in \Lambda_+^*\\r+v\ne0} }|r| |v||\eta_r| \|b_{r+v}b_{-r}\xi_1 \|\|b_{v}\xi_2 \|\\
  \leq &\, \frac{C}{N^{1/2}} \Big[\sum_{\substack{r,v\in \Lambda_+^*\\r+v\ne0} }|r|^2\,\|b_{r+v}b_{-r}\xi_1\|^2\Big]^{1/2}\Big[\sum_{\substack{r,v\in \Lambda_+^*\\r+v\ne0}}|\eta_r|^2|v|^2\|b_{v}\xi_2 \|^2\Big]^{1/2}\\
   \leq &\, \frac{C}{N^{1/2}} \|\cK^{1/2}(\cN_++1)^{1/2}\xi_1\|\|\cK^{1/2}\xi_2\|\,. 
\end{split}
\end{equation}

We rewrite the term $[\cV_N,A]_1$ as 
\begin{equation*}\label{eq:[cVN,A]1}\begin{split} 
[\cV_N,A]_1 = &\,  \frac{1}{N^{3/2}}\sum_{\substack{ r,v \in \Lambda_+^*\\r+v\ne0}} (\widehat{V}\left(\cdot/{N}\right)*\eta)(r)\,b^*_{r+v}b^*_{-r}(\gamma_v b_{v}+\sigma_vb^*_{-v}) \\
&\, - \frac{1}{N^{3/2}}\sum_{\substack{r,v \in \Lambda_+^*\\r+v\ne0}} 
\widehat{V} (r/N) \eta_0\,b^*_{r+v}b^*_{-r}(\gamma_v b_{v}+\sigma_vb^*_{-v}) 
\end{split}
\end{equation*}
Using \eqref{eq:defeta}, the decomposition $\eta_p = - N\d_{p,0} + N \widehat{f}_{N,\ell} (p)$, and the fact that  $|\eta_0|\leq C$ (see \eqref{eq:norms_eta}), we can bound 
\[
|\langle \xi_1 , [\cV_N,A]_1 \xi_2 \rangle| \leq  C \|\cK^{1/2}\xi_1 \|\|(\cN_++1)\xi_2\|\,.
\]

 With \eqref{eq:boundalpha}, we bound  $[\cV_N, A]_2$ as
 \begin{equation}\label{eq:bound[VN,A]2}
 \begin{split}
     |\langle\xi_1, [\cV_N,A]_2\,\xi_2\rangle| \leq &\, \frac{C}{N}\sum_{\substack{\;r,v\in \Lambda_+^*\\\,r+v\ne0}} |\alpha_{r,v}|\|(\cN_++1)^{-1}b_{r+v}b_{-r}b_{-v}\xi_1 \|\|(\cN_++1)\xi_2 \|\\
     \leq &\, \frac{C}{N^{1/2}} \|(\cN_++1)\xi_2 \|\Big[ \frac{1}{N} \sum_{\substack{\;r,v\in \Lambda_+^*\\\,r+v\ne0}} \frac{|\alpha_{r,v}|^2}{|r|^2} \Big]^{1/2}\\
     & \hspace{1.5cm}\times\Big[ \sum_{\substack{\;r,v\in \Lambda_+^*\\\,r+v\ne0}} |r|^2 \| b_{r+v} b_{-r}b_{-v}(\cN_++1)^{-1}\xi_1 \|^2\Big]^{1/2} \\
     \leq &\, \frac{C}{N^{1/2}}\|\cK^{1/2}\xi_1\|\|(\cN_++1)\xi_2 \|\,.
 \end{split}
 \end{equation}
 


\end{proof}

We can now prove Proposition \ref{prop:a_priori_A}
\begin{proof}[Proof of Prop. \ref{prop:a_priori_A}] 
The proof of the bound \eqref{eq:growthN} follows similarly as in \cite[Prop. 4.2]{BBCS}. To show \eqref{eq:growthNimproved} we consider the function
\[
g_\xi(s) =\langle \xi, e^{-sA}\cN_+ e^{sA} \xi \rangle
\]
for $s \in [0,1]$, and its derivative
\[
g'_\xi(s) = \langle \xi, e^{-sA}[\cN_+,A]e^{sA}\xi\rangle
\]
Since 
\[
\begin{split}
    [\cN_+,A] = &\, \frac{1}{\sqrt N} \sum_{r,v \in \L^*_+} b^*_{r+v}b^*_{-r} \big(\eta_r\g_vb_v+3\nu_{r,v} b^*_{-v}\big) +\hc 
\end{split}
\]
and using that $|\nu_{r,v}| \leq C |r|^{-2}|v|^{-2}$, we immediately find 
\[\begin{split}
|g'_\xi(s)| 
\leq & \frac{C}{\sqrt{N}} \sum_{r,v \in \L^*_+}  \Big(|\eta_r|\|b_{r+v}b_{-r}e^{sA}\xi\|\| b_ve^{sA}\xi\|\\
& \qquad\qquad+ |\nu_{r,v}|\|(\cN_++1)^{-1}b_{r+v}b_{-r}b_{-v}e^{sA}\xi\|\|(\cN_++1)e^{sA}\xi\|\Big) \\ 
\leq & \frac{C}{\sqrt N} \|(\cN_++1)e^{sA}\xi\|\|\cN_+^{1/2}e^{sA}\xi\|\\
\leq &\, \langle \xi, e^{-sA}\cN_+e^{sA}\xi\rangle  + \frac{C}{N} \langle \xi, e^{-sA}(\cN_++1)^2e^{sA}\xi\rangle\,.
\end{split}\]
Using Gronwall's lemma, and Eq. \eqref{eq:growthN}, we have 
\[
g_\xi(s) \leq\langle \xi, \cN_+ \xi\rangle+ \frac{C}{N} \langle\xi, (\cN_++1)^2\xi\rangle\,,
\]
which concludes the proof of \eqref{eq:growthNimproved}. To show (\ref{eq:growthHN}), we define  
\begin{equation*}
    \label{eq:phixi}
    \ph_\xi(s) :=  \langle \xi, e^{-sA} (\cN_++1)^k(\cH_N+1) e^{sA} \xi \rangle 
\end{equation*}
for $s \in [0,1]$. Differentiating with respect to $s$, we have 
\[\begin{split}
\ph'_\xi(s) = &\, \langle \xi, e^{-sA} [(\cN_++1)^k(\cH_N+1),A] e^{sA} \xi \rangle \\
= & \, \langle \xi, e^{-sA} (\cN_++1)^k[\cK +\cV_N,A] e^{sA} \xi \rangle  + \langle \xi, e^{-sA} [(\cN_++1)^k,A](\cH_N+1) e^{sA} \xi \rangle \\= &\, P_1 +P_2\,.
\end{split}\]
We consider first $P_1$. From (\ref{eq:comm-gron}), we find  
\begin{equation*}\label{eq:gronwallP1}
\begin{split}
|\langle \xi, e^{-sA} (\cN_++1)^k[\cK +\cV_N,A] e^{sA} \xi \rangle| \leq \,&C \langle \xi, e^{-sA} (\cN_++1)^k(\cH_N+1) e^{sA} \xi \rangle \\
&+ C\langle \xi, (\cN_++1)^{k+2}  \xi \rangle\,.
\end{split}
\end{equation*}
As for the term $P_2$, we observe that the proof of \cite[Prop. 4.4]{BBCS} (restricted to $k=1$) can be easily extended to general $k \in \bN$. We conclude that  
\[
|P_2| \leq C\langle \xi, e^{-sA}(\cH_N+1)(\cN_++1)e^{sA}\xi\rangle + C \langle \xi,(\cN_++1)^{k+2}\xi\rangle\,.
\]
Putting together the estimates for $P_1$ and $P_2$, we obtain 
\[
\ph'_\xi(s) \leq C\ph_\xi(s) + C \langle \xi, (\cN_++1)^{k+2}\xi\rangle 
\]
for any $\xi \in \cF^{\leq N}_+$ and for some constant $C>0$. With Gronwall's lemma we get the desired result. 
\end{proof}

\section{Cubic renormalization: proof of Theorem \ref{thm:cubic}} 
\label{sec:cubic} 

The starting point for proving Theorem \ref{thm:cubic} is the representation (\ref{eq:expansion_G_N}) for the excitation Hamiltonian $\cG_N$. To determine the structure of $\cJ_N = e^{-A} \cG_N e^A$, we will separately apply the cubic conjugation to the different summands in $\cG_N$.

\subsection{Control of quadratic terms} 

To conjugate quadratic terms in $\cG_N$ (excluding the kinetic energy), we will make use of the following lemma.   
\begin{lemma}\label{lm:quad-A}
Let $A$ be defined as in \eqref{eq:def_A}. For $p\in \Lambda_+^*$, let $w_p \in \bR$. Set 
\[ \cW = \sum_{p\in \L^*_+} w_p a_p^* a_p \]
Then
\begin{equation}\label{eq:quad-dia} | \langle \xi,  \big[ \cW , A \big] \xi \rangle | \leq \frac{C \| w \|_\infty}{\sqrt{N}} \| \cN_+^{1/2} \xi \| \| (\cN_+ + 1) \xi \| \end{equation} 
Moreover, 
\begin{equation}\label{eq:quad-dia2} | \langle \xi,  \big[ \cW , A \big] \xi \rangle | \leq \frac{C \sup_{p \in \Lambda^*_+} |w_p|/|p|}{\sqrt{N}} \| \cN_+^{1/2} \xi \| \| \cK^{1/2} (\cN_+ + 1)^{1/2}  \xi \| \end{equation} 
Suppose now $O_p \in \bC$ for all $p \in \Lambda^*_+$ and set 
\[ \cO = \sum_{p \in \L^*_+} O_p b_p^* b_{-p}^*. \]
Then  
\begin{equation}\label{eq:quad-off1} | \langle \xi,  \big[ \cO , A \big] \xi \rangle | \leq \frac{C \| O \|_2}{\sqrt{N}} \| \cN_+^{1/2} \xi \| \| (\cN_+ + 1) \xi \| \end{equation} 
and 
\begin{equation}\label{eq:quad-off} | \langle \xi,  \big[ \cO , A \big] \xi \rangle | \leq  \frac{C}{\sqrt{N}}  \left[ \| O \|_\infty + \sum_{p \in \Lambda_+^*} \frac{|O_p|^2}{p^2} \right]^{1/2} \| \cK^{1/2} \xi \| \| (\cN_+ + 1) \xi \| .\end{equation} 
\end{lemma} 

\begin{proof}
From the identity
\begin{equation*} \label{eq:commut_quad_diag}
\begin{split}
    [\mathcal{W},A]=\;&\frac{1}{\sqrt{N}} \sum_{\substack{r,v\in \Lambda_+^*\\r+v\ne0}}\big(w_{r+v}+w_r-w_v\big) \eta_r \gamma_v b^*_{r+v} b^*_{-r} b_v+\mathrm{h.c.}\\
    &+\frac{1}{\sqrt{N}} \sum_{\substack{r,v\in \Lambda_+^*\\r+v\ne0}}\big(w_{r+v}+w_r+w_v\big) \nu_{r,v} b^*_{r+v} b^*_{-r} b_{-v}^*+\mathrm{h.c.}.
\end{split}
\end{equation*}
we find (\ref{eq:quad-dia}) and (\ref{eq:quad-dia2}), using that $\gamma \in \ell^\infty$ and that $\eta , \nu \in \ell^2$ ($\nu$ is square integrable in both its variables), uniformly in $N$. To prove (\ref{eq:quad-off1}), (\ref{eq:quad-off}), we proceed as in the proof of \cite[Prop. 8.2]{BBCS}. From the commutation relations (\ref{eq:CCR_b}), we obtain 
 \begin{equation} \label{eq:commut_quad_offdiag1}
\begin{split}
    \big[ \mathcal{O}+ \mathcal{O}^*,A\big]=\;&\frac{2}{\sqrt{N}} \sum_{\substack{r,v\in \Lambda_+^*\\r+v\ne0}} O_{r+v} \, b^*_{-r-v} \big( \eta_r\gamma_v b^*_v +\nu_{r,v}b_{-v} \big) b_r + \mathrm{h.c.}\\
    &+\frac{2}{\sqrt{N}} \sum_{\substack{r,v\in \Lambda_+^*\\r+v\ne0}} O_r \, b^*_r \big( \eta_r\gamma_v b^*_v +\nu_{r,v} b_{-v} \big) b_{r+v}+\mathrm{h.c.}\\
    &+\frac{2}{\sqrt{N}} \sum_{\substack{r,v\in \Lambda_+^*\\r+v\ne0}} O_v \nu_{r,v} \, b^*_v b_{r+v} b_{-r}+\mathrm{h.c.}\\
    &-\frac{2}{\sqrt{N}} \sum_{\substack{r,v\in \Lambda_+^*\\r+v\ne0}}  O_v \eta_r \gamma_v  \,b^*_{r+v} b^*_{-r} b^*_{-v}+\mathrm{h.c.}+\mathcal{E},
\end{split}
\end{equation}
where the error term $\cE$ collects contributions due to the fact that the modified creation and annihilation operators $b, b^*$ do not exactly satisfy canonical commutation relations. We have 
\begin{equation}  
    | \langle \xi, \mathcal{E} \xi \rangle| \le \frac{C\|O\|_2}{N} \| \cN_+^{1/2} \xi \| \| (\mathcal{N}_++1) \xi \|  
\end{equation}
and also 
\begin{equation} 
    | \langle \xi , \mathcal{E}  \xi \rangle | \le \frac{C}{N} \left[ \| O \|_\infty + \sum_p \frac{|O_p|^2}{p^2} \right]  \| \cK^{1/2} \xi \| \| (\mathcal{N}_++1)  \xi \| \, .
\end{equation}
Bounding the explicit terms on the r.h.s. of (\ref{eq:commut_quad_offdiag1}) one by one, we conclude the proof of the lemma. As an example, consider the first term on the first line. We can estimate
\[ \begin{split} \Big| \frac{2}{\sqrt{N}} \sum_{r,v} O_{r+v} \eta_r \gamma_v \langle \xi, b^*_{-r-v} b_v^* b_r \xi \rangle \Big| &\leq \frac{C}{\sqrt{N}} \sum_{r,v} |O_{r+v}| \| b_{-r-v} b_v \xi \| \| b_r \xi \|  \\ &\leq \frac{C}{\sqrt{N}} \| O \|_2 \| \cN^{1/2} _+ \xi \| \| \cN_+  \xi \| \end{split} \]
or, if $O \not \in \ell^2$,  
\[ \begin{split}  \Big| \frac{2}{\sqrt{N}} \sum_{r,v} O_{r+v} \eta_r \gamma_v \langle & \xi, b^*_{-r-v} b_v^* b_r \xi \rangle \Big| \\ \leq \; &\frac{C}{\sqrt{N}} \sum_{r,v} \frac{|O_{r+v}|}{|r+v|} \, |r+v|  \| b_{-r-v} b_v \xi \| \| b_r \xi \|  \\ \leq \; &\frac{C}{\sqrt{N}} \left[ \sum_r \frac{|O_r|^2}{r^2} \right]^{1/2}  \| \cK^{1/2}  \xi \| \| \cN_+ \xi \| .\end{split} \]
 
In contrast to the proof of \eqref{eq:quad-dia} and \eqref{eq:quad-dia2} (where we only used the $\ell^\infty$ norm of $w$ and of $w_\cdot / | \cdot |$), here we need the decay of the observable $O$ (because $\eta_r$ and $\| b_r \xi \|$ both decay in the same variable $r$ and in one of the factors arising from the Cauchy-Schwarz inequality only the observable $O_{r+v}$ provides decays in $v$). The last term on the r.h.s. of (\ref{eq:commut_quad_offdiag1}) can be bounded similarly. The other terms can be estimated using the $\ell^\infty$ norm of $O$ and the fact that $\gamma \in \ell^\infty$, $\eta , \nu \in \ell^2$, uniformly in $N$.
\end{proof}  

Applying Lemma \ref{lm:quad-A}, we obtain the following proposition.
\begin{prop}\label{prop:quad-A}
Recall the operator $\cT^{(2)}_{\cG_N}$, defined in (\ref{eq:cT_G_N}). Furthermore, from (\ref{eq:expansion_G_N}), we consider the operators 
\[ \begin{split} D &= \sum_{p\in \Lambda^*_+} \big[ \sqrt{|p|^4 + 2 (\widehat{V} (\cdot/N) * \widehat{f}_{N,\ell})_p \, p^2} -p^2 \big] \, a_p^* a_p \\ 
E &= - \frac{1}{N} \sum_{p,q \in \L^*_+} \big( \widehat{V} (q/N) + \widehat{V} ((q+p)/N) \big) \eta_q \big[ (\g_p^2 + \s_p^2) b_p^* b_p + \g_p \s_p (b_p^* b_{-p}^* + b_p b_{-p} ) \big]
\end{split} \]
Then, for every $0< \eps <1$, we have 
\begin{equation}\label{eq:eT2e} \pm e^{-A} \cT^{(2)}_{\cG_N} e^A  \leq \eps \cK + \frac{C}{\eps N} (\cH_N + \cN_+^2 + 1) \end{equation} 
and also 
\[ e^{-A} D e^A = D + \cE_D, \qquad e^{-A} E e^A = E + \cE_E \]
where
\[ \pm \cE_D, \cE_E \leq \eps \cN_+ + \frac{C}{\eps N} (\cN_+ + 1)^2 \]
\end{prop} 

\begin{proof}
First of all, we observe that 
\[ \Big| \sqrt{|p|^4 + 2 (\widehat{V} (\cdot/N) * \widehat{f}_{N,\ell})_p \, p^2} -p^2 \Big| \leq C \]
and also that 
\[ \begin{split} \Big| \frac{1}{N} \sum_q \big( \widehat{V} (q/N) + \widehat{V} ((q+p)/N) \big) \eta_q (\g_p^2 + \s_p^2) \Big| &\leq C \\
\Big| \frac{1}{N} \sum_q \big( \widehat{V} (q/N) + \widehat{V} ((q+p)/N) \big) \eta_q  \g_p \s_p \Big| &\leq C |p|^{-2} \end{split}  \]
for all $p \in \Lambda_+^*$. Thus, we can apply Lemma \ref{lm:quad-A} (and in particular (\ref{eq:quad-dia}) and (\ref{eq:quad-off1})) to estimate the commutators $[ D,A] , [E,A]$. Writing 
\[ \cE_D  = \int_0^1 ds \, e^{-sA} [D,A] e^{sA}  \]
we can therefore bound
\[ |\langle \xi, \cE_D \xi \rangle | \leq \frac{C}{\sqrt{N}} \int_0^1 ds \, \| \cN_+^{1/2} e^{-sA} \xi \| \| (\cN_+ + 1) e^{-sA} \xi \|.  \]
From Prop. \ref{prop:a_priori_A}, we conclude that 
\[ |\langle \xi, \cE_D \xi \rangle | \leq \frac{C}{\sqrt{N}} \| \cN_+^{1/2} \xi \| \| (\cN_+ + 1) \xi \| + \frac{C}{N} \| (\cN_+ + 1) \xi \|^2. \]
Similarly, we can also bound $\delta_E$. 

To show (\ref{eq:eT2e}), we rewrite 
\[ \cT^{(2)}_{\cG_N} = \sum_{p \in \L^*_+} O_p (b_p b_{-p} + b_{-p}^* b_p^*) \]
with
\[ O_p = \frac{(1+2\| \sigma \|_2^2)}{N^2} \, (\widehat{V} (\cdot /N) * \eta)_p + \frac{1}{2N} \sum_{q \in \Lambda^*_+} (2\sigma_q^2 + \frac{\gamma_q \sigma_q}{\mu_q} - 1) \,  p^2 \eta_p \]
and we observe that $|O_p| \leq C/N$ for all $p\in \Lambda^*_+$ and that 
\begin{equation}\label{eq:O-ass0} \sum_{p \in \Lambda^*_+} \frac{|O_p|^2}{p^2} \leq \frac{C}{N}. \end{equation} 
This implies that 
\[  \big| \langle \xi, \cT^{(2)}_{\cG_N} \xi \rangle \big| \leq \left[ \sum_p \frac{|O_p|^2}{p^2} \right]^{1/2} \| \cK^{1/2} \xi \| \| (\cN_+ + 1)^{1/2} \xi \| \leq \frac{C}{\sqrt{N}} \| \cK^{1/2} \xi \| \| (\cN_+ + 1) \xi \| \]
and, by Lemma \ref{lm:quad-A}, also that 
\[  \big| \langle \xi, e^{-A} \big[ \cT^{(2)}_{\cG_N} , A \big] e^A \xi \rangle \big|  \leq \frac{C}{N} \| \cK^{1/2} e^A \xi \| \| (\cN_+ + 1) e^A \xi \|. \] 
Writing 
\[ e^{-A} \cT^{(2)}_{\cG_N} e^A = \cT_{\cG_N}^{(2)} + \int_0^1 ds \, e^{-sA}  \big[ \cT^{(2)}_{\cG_N} , A \big] e^{sA} \]
and applying also Prop. \ref{prop:a_priori_A}, we arrive at (\ref{eq:eT2e}). 
\end{proof}

\subsection{Control of quartic terms} 

Next, we control the quartic error term $\cT_{\cG_N}^{(4)}$, defined in (\ref{eq:cT_G_N}). To this end, we will make use of the following lemma. 
\begin{lemma} \label{lemma:commut_quartic_A_1}
Let $A$ be defined as in \eqref{eq:def_A}. We consider a quartic operator of the form 
\begin{equation*} 
\begin{split}
    \mathcal{D}= \;&\frac{1}{N} \sum_{\substack{r,p,q\in \Lambda_+^*\\r+p+q\ne0}} D_{r,p,q}\, b^*_{r+p} b^*_{q} b^*_{-p} b^*_{-r-q}+\mathrm{h.c.} 
\end{split}
\end{equation*}
where we assume that, for $r,p, q \in \L^*_+$, the coefficients $D_{r,p,q}$ are so that
\begin{equation}\label{eq:D-ass} \begin{split} 
\min \Big\{ \sum_{p,q} \sup_r |D_{r,p,q}|^2 + \frac{1}{N} &\sum_{p,q} \left[ \sup_z \sum_{r} \frac{|D_{r,p,q}|^2}{(r+z)^2} \right] + \frac{1}{N^2} \sum_{p,q} \left[ \sup_z \sum_r \frac{|D_{r,p,q}|}{(r+z)^2} \right]^2 , \\ &\hspace{5.2cm} (r \leftrightarrow p) , (r \leftrightarrow q) \Big\} \leq C \end{split}  \end{equation} 
Here, $(r \leftrightarrow p)$ means the same quantity, but inverting the role of the momenta $r$ and $p$ (ie. in the first term, we sum over $r,q$ and we take the supremum over $p$, and similarly for the other terms). This assumption reflects the fact that (in applications) one of the momenta $r,p,q$ is the argument of the potential $\widehat{V} (./N)$ while the dependence on the other two momenta is square-summable. Then there exists $C > 0$ such that 
\[ | \langle \xi, \big[\cD , A \big] \xi \rangle | \leq \frac{C}{\sqrt{N}} \| \cN^{1/2}_+ \xi \| \| (\cN_+ + 1) \xi \| +\frac{C}{N} \langle \xi, (\cK+ 1) (\cN_+ + 1)^3 \xi \rangle  \]
for all $\xi \in \cF_{\perp}^{\leq N}$. 
\end{lemma}

\begin{proof}
We decompose 
\[ \begin{split} 
\big[ D, A \big] = \; &\frac{1}{N^{3/2}} \sum_{r,p,q,s,w} D_{r,p,q}  \, \eta_s \gamma_w \, b_{s+w}^* b_{-s}^* \big[ b_{r+p}^* b^*_q b_{-p}^* b^*_{-r-q} , b_w \big]+\mathrm{h.c.} \\ &-\frac{1}{N^{3/2}} \sum_{r,p,q,s,w} D_{r,p,q} \, \eta_s \gamma_w b_w^* \big[ b_{r+p}^* b_q^* b_{-p}^* b_{-r-q}^* , b_{-s} b_{s+w} \big]+\mathrm{h.c.} \\ &- \frac{1}{N^{3/2}} \sum_{r,p,q,s,w} D_{r,p,q} \, \nu_{s,w} \big[ b^*_{r+p} b^*_q b_{-p}^* b_{-r-q}^*, b_{-w} b_{-s} b_{s+w} \big]+\mathrm{h.c.} \\ =: \; & \big(\text{I} + \text{II} + \text{III}\big)+\mathrm{h.c.} \end{split} \]
After a long but straightforward computation based on the commutation relations (\ref{eq:CCR_b}), we find 
\begin{equation}\label{eq:quarI+II} \begin{split} \text{I} = \; &- \frac{1}{N^{3/2}} \sum_{r,p,q,s} D^\text{I}_{r,p,q} \gamma_{r+p} \eta_s \,   b_{s+p+r}^* b_q^* b_{-s}^* b_{-p}^* b_{-r-q}^* + \cE_1 \\ 
\text{II} = \; &\frac{1}{N^{3/2}} \sum_{r,p,q,w}  D^{\text{II},1}_{r,p,q} \, (\eta_q + \eta_{q-w}) \gamma_w \,  b_w^* b_{p+r}^* b_{-r-q}^* b_{-p}^* b_{w-q} \\ &+\frac{1}{N^{3/2}} \sum_{r,p,q} D^{\text{II},2}_{r,p,q} \, (\eta_p + \eta_{p+r}) \gamma_r b_r^* b_q^* b_{-r-q}^* + \cE_2 \end{split} \end{equation} 
with the coefficients 
\begin{equation}\label{eq:Drpq}  \begin{split} 
D^{\text{I}}_{r,p,q} &= D_{r,p,q} + D_{-r, q+r, p+r} + D_{r, -p-r, q} + D_{-r, -q, -p} \, ,  \\ 
D^{\text{II},1}_{r,p,q} &= D_{r,p,q} + D_{-r, q+r, p+r} + D_{r,-q,-p} + D_{-r, -q, -p} + D_{r,p, -q-r} \, ,  \\ D^{\text{II},2}_{r,p,q} &= D_{r,p,q} + D_{r+p+q, -q, -p} + D_{p-q, r+q, q} + D_{q-p, p, r+p} + D_{-r, -q, -p} + D_{-r-p-q, p,q} \end{split} \end{equation} and
\begin{equation} \label{eq:quarIII} \begin{split} 
\text{III} = \; &\frac{1}{N^{3/2}} \sum_{p,q,r,s} D^{\text{III},1}_{r,p,q}  \, \big(\nu_{s,-q} + \nu_{s,q-s} + \nu_{q,-s} \big) \, b_{r+p}^* b_{-p}^* b_{-r-q}^* b_{-s} b_{s-q} \\ &+ \frac{1}{N^{3/2}} \sum_{r,p,q}  D^{\text{III},2}_{r,p,q}  \, \nu_{p+r,-q-r} b_{-p}^* b_q^* b_{q-p} \\ &+\frac{1}{N^{3/2}} \sum_{p,q,r} D^{\text{III},3}_{r,p,q} \, \big( \nu_{p-q,q+r} + \nu_{q+r,p-q} \big) b_q^* b_{-p}^* b_{q-p} \\
&+\frac{1}{N^{3/2}} \sum_{r,p,q} D^{\text{III},4}_{r,p,q} \, \big( \nu_{p-q,q+r}+ \nu_{q+r,p-q}\big) b_q^* b_{-p}^* b_{q-p} + \cE_3 \end{split} \end{equation} 
with 
\[ \begin{split} D^{\text{III},1}_{r,p,q} = \; & D_{r,p,q} + D_{-r, q+r, p+r} + D_{-r, -q,-p} + D_{r,p,-q-r}  \, \\
D^{\text{III},2}_{r,p,q} = \; &D_{r,p,q} + D_{r,p,-q-r} + D_{r, -p-r, -q-r} + D_{r, -p-r, q} + D_{-r, -q, p+r}  \\ &+ D_{-r, q+r, p+r} + D_{-r,-q, -p} + D_{-r, q+r, -p} + D_{p-q, r+q, q}  + D_{p-q, -r-p, -p} \\ &+ D_{q-p, p, p+r} + D_{q-p, -q, -q-r}   \, \\
D^{\text{III},3}_{r,p,q} = \; &  D_{r,p,q}+ D_{r,p, -q-r} + D_{p-q, q+r, q} + D_{p-q, -p-r, q} + D_{r,-p-r, -q-r} + D_{r, -p-r, q}  \, \\ 
D^{\text{III},4}_{r,p,q} = \; & D_{r,p,q}+ D_{r,p, -q-r} + D_{q-p, p, p+r} + D_{q-p, p, -q-r} + D_{r,-p-r, -q-r}  + D_{r, -p-r, q}   \, . \end{split} \]
In (\ref{eq:quarI+II}), (\ref{eq:quarIII}), the error terms $\cE_1, \cE_2, \cE_3$ are produced by the fact that the commutation relations (\ref{eq:CCR_b}) are not precisely canonical and can be controlled by 
\[ \pm \cE_1, \pm \cE_2, \pm \cE_3 \leq \frac{C}{N} (\cK+1) (\cN_+ + 1)^3 .\]

To bound the first term in (\ref{eq:quarI+II}), we can use Cauchy-Schwarz. Let us consider for example the contribution $\text{I}_1$, arising from the first term in the identity (\ref{eq:Drpq}) for $D^\text{I}_{r,p,q}$. We find 
\[ \begin{split} \big| \langle \xi, \text{I}_1 \,  \xi \rangle \big| &\leq \frac{C}{N^{3/2}} \| (\cN_+ + 1) \xi \|  \left[ \sum_{r,p,q,s} \frac{|D_{r,p,q}|^2}{(r+q)^2} \eta_s^2  \right]^{1/2} \\ &\hspace{2cm} \times \left[ \sum_{r,p,q,s} (q+r)^2 \, \| b_q b_{-s} b_{-p} b_{-r-q} (\cN_+ + 1)^{-1/2} \xi \|^2 \right]^{1/2} \\ &\leq \frac{C}{N^{3/2}}  \left[ \sum_{p,q} \sup_z \sum_r \frac{|D_{r,p,q}|^2}{(r+z)^2} \right]^{1/2} \| (\cN_+ + 1) \xi \| \| \cK^{1/2} (\cN_+ + 1) \xi \|. \end{split} \]
With different choices of the weight in the Cauchy-Schwarz inequality, we can exchange the role of the labels $r,p,q$. Proceeding analogously for the terms arising from the other contributions to $D^\text{I}_{r,p,q}$, we arrive at 
\[ \begin{split}  \big| \langle \xi, \text{I} \, \xi \rangle \big| &\leq \frac{C}{N^{3/2}} \min \Big\{ \sum_{p,q} \sup_z \sum_{r} \frac{|D_{r,p,q}|^2}{(r+z)^2} , \sum_{r,q} \sup_z \sum_p \frac{|D_{r,p,q}|^2}{(p+z)^2} , \sum_{r,p} \sup_z \sum_q \frac{|D_{r,p,q}|^2}{(q+z)^2}
 \Big\}^{1/2}  \\ & \hspace{9cm} \times \langle \xi, (\cK+1) (\cN_+ + 1)^2 \xi \rangle \\ &\leq \frac{C}{N} \langle \xi, (\cK +1) (\cN_+ + 1)^2 \xi \rangle\end{split} \]
where we applied the assumption (\ref{eq:D-ass}). Also the other quintic terms, in (\ref{eq:quarI+II}) and in (\ref{eq:quarIII}), can be handled similarly. Let us now consider the cubic terms. The contribution to the cubic term on the last line of (\ref{eq:quarI+II}) arising from the first term in the expression for $D^{\text{II},2}_{r,p,}$ in (\ref{eq:Drpq}) can be bounded by 
\[ \begin{split}  |\langle \xi , \text{II}_2 \xi \rangle | &\leq \frac{C\| (\cN_+ + 1)^{1/2} \xi \|}{N^{3/2}} \left[ \sum_{r,p,q} \frac{|D_{r,p,q}|^2}{r^2} \right]^{1/2} \left[ \sum_{r,,p,q} (\eta_p^2 + \eta_{p+r}^2)  r^2 \| b_r b_q \xi \|^2 \right]^{1/2} \\ &\leq \frac{C}{N^{3/2}}  \left[ \sum_{r,p,q} \frac{|D_{r,p,q}|^2}{r^2} \right]^{1/2}  \| (\cN_+ + 1)^{1/2} \xi \| \| \cK^{1/2} (\cN_+ + 1)^{1/2} \xi \|. \end{split} \]
Analogously, the estimate also holds if we replace $r^{-2}$ with $q^{-2}$. But we cannot replace $r^{-2}$ with $p^{-2}$. So, we proceed slightly differently, using $|\eta_p| \leq C / p^2$ to bound 
\[  \begin{split}  |\langle \xi , \text{II}_2 \xi \rangle | &\leq \frac{C\| (\cN_+ + 1) \xi \|}{N^{3/2}} \left[ \sum_{r,q}  \left( \sup_z \sum_p \frac{|D_{r,p,q}|}{(p+z)^2} \right)^{2} \right]^{1/2} \left[ \sum_{r,q}  \| b_r b_q (\cN_+ + 1)^{-1/2} \xi \|^2 \right]^{1/2} \\ &\leq \frac{C}{N^{3/2}} \left[ \sum_{r,q}  \left(  \sup_z \sum_p \frac{|D_{r,p,q}|}{(p+z)^2} \right)^{2} \right]^{1/2}  \| (\cN_+ + 1) \xi \| \| \cN^{1/2} _+ \xi \|. \end{split} \]
From the assumption (\ref{eq:D-ass}), we conclude that 
\[ |\langle \xi, \text{II}_2 \xi \rangle | \leq \frac{C}{N} \langle \xi, (\cK +1) (\cN_+ + 1) \xi \rangle + \frac{C}{\sqrt{N}} \| \cN^{1/2}_+ \xi \| \| (\cN_+ + 1) \xi \|. \]
Also the other contributions to the cubic term in (\ref{eq:quarI+II}) can be handled similarly. The second term $\text{III}_2$ on the r.h.s. of (\ref{eq:quarIII}) can be estimated by 
\[ \begin{split} |\langle \xi, \text{III}_2 \xi \rangle | &\leq \frac{C}{N^{3/2}}  \sum_{r,p,q} |D^{\text{III},2}_{r,p,q}| \frac{|\eta_{p+r}|}{|q+r|^{2}} \| b_{-p} b_q \xi \| \| b_{q-p} \xi \| \\ &\leq \frac{C}{N^{3/2}} \min \Big\{ \sum_{q}  \sup_{r,p}  |D^{\text{III},2}_{r,p,q}|^2 , \sum_p \sup_{r,q}  |D^{\text{III},2}_{r,p,q}|^2 \Big\}^{1/2} \, \| \cN_+ \xi \| \| \cN_+^{1/2} \xi \| \\ &\leq \frac{C}{N} \|  \cN_+ \xi \| \| \cN_+^{1/2} \xi \|. \end{split} \]
As for the last two terms on the r.h.s. of (\ref{eq:quarIII}), they can be controlled similarly as $\text{II}_2$. We skip the details.
\end{proof}

We can now control the conjugation of the quartic component of $\cT_{\cG_N}$. 
\begin{prop} \label{prop:T4-bd} 
Let $\cT_{\cG_N}^{(4)}$ be defined as in (\ref{eq:cT_G_N}) and $A$ as in (\ref{eq:def_A}). Then, for every $\eps > 0$ small enough, we have   
\[ \pm e^{-A} \cT^{(4)}_{\cG_N} e^A  \leq \eps \cK + \frac{C}{\eps N} (\cH_N + 1) (\cN_+ + 1)^4 \]
\end{prop} 

\begin{proof}
We write 
\begin{equation}\label{eq:eT4e} \langle \xi, e^{-A} \cT^{(4)}_{\cG_N} e^A \xi \rangle  = \langle \xi,  \cT^{(4)}_{\cG_N} \xi \rangle + \int_0^1 ds \, \langle \xi, e^{-sA} \big[ \cT^{(4)}_{\cG_N} , A \big] e^{sA} \xi \rangle. \end{equation} 
Next, we observe that 
\[ | \langle \xi,  \cT^{(4)}_{\cG_N} \xi \rangle | \leq \frac{C}{\sqrt{N}} \| \cK^{1/2} \xi \| \| (\cN_+ + 1)^{3/2} \xi \|. \] 
In fact, the contribution of the first term in the definition (\ref{eq:cT_G_N}) of $\cT^{(4)}_{\cG_N}$ is bounded by 
\[ \begin{split} \Big| \frac{1}{2N} &\sum_{p,q,r} \widehat{V} (r/N) \sigma_p \sigma_{q+r} \langle \xi, b^*_{p+r} b^*_q b^*_{-p} b^*_{-q-r} \xi \rangle \Big| \\  \leq \; &\frac{C}{N} \| (\cN_+ + 1)^{3/2} \xi \| \left[ \sum_{p,q,r} (p+r)^2 \| b_{p+r} b_q b_{-p} b_{-q-r} (\cN_+ + 1)^{-3/2} \xi \|^2 \right]^{1/2}\\ &\hspace{4cm} \times  \left[ \sum_{p,q,r} \frac{|\widehat{V} (r/N)|^2}{(r+p)^2} \sigma_p^2 \sigma_{q+r}^2 \right]^{1/2} \\ \leq \; &\frac{C}{\sqrt{N}}   \| \cK^{1/2} \xi \| \| (\cN_+ + 1)^{3/2} \xi \|. \end{split} \]
Also the other contributions to $\cT_{\cG_N}^{(4)}$ can be bounded similarly. As for the second term on the r.h.s. of (\ref{eq:eT4e}), we apply Lemma \ref{lemma:commut_quartic_A_1}. To this end, we observe that all terms in $\cT_{\cG_N}^{(4)}$ satisfy the assumption (\ref{eq:D-ass}) (with $D_{r,p,q} = \widehat{V} (r/N) \sigma_p \sigma_{q+r}$ for the first term, $D_{r,p,q} = N^{-1} (\widehat{V} (./N) *\eta)_p 
\gamma_q \sigma_q \delta_{r,0}$ for the second term, $D_{r,p,q} = (\widehat{V} (./N) * \hat{f}_{N,\ell})_p (\gamma_q \sigma_q - \eta_q/2 - \sigma_q^2/(2\mu_q))$ for the third term). We obtain 
\[ \begin{split} \Big|  \langle \xi, e^{-sA} &\big[ \cT^{(4)}_{\cG_N} , A \big] e^{sA} \xi \rangle \Big|\\  \leq \; &\frac{C}{N} \langle \xi, e^{-sA} (\cK+1) (\cN_+ + 1)^3 e^{sA} \xi \rangle + \frac{C}{\sqrt{N}} \| \cN_+^{1/2} e^{sA} \xi \| \| (\cN_+ + 1) e^{sA} \xi \|\\  \leq \; &\frac{C}{N} \langle \xi, (\cH_N +1) (\cN_+ + 1)^4 \xi \rangle + \frac{C}{\sqrt{N}} \| \cN_+^{1/2} \xi \|  \| (\cN_+ + 1) \xi \| \end{split} \]
where in the last step we also used Prop. \ref{prop:a_priori_A}. 
\end{proof}

\subsection{Control of cubic term} 

In this section, we study the conjugation of the cubic term $\cC_{\cG_N}$. We will use the next lemma.

\begin{lemma} \label{lemma:commut_C,A} 
    Let $\cC_{\cG_N}$ be defined as in \eqref{eq:cC_G_N} and $A$ as in \eqref{eq:def_A}. Then, we have 
    \begin{equation} \label{eq:commut_C,A}
    \begin{split}
    [ \mathcal{C}_{\mathcal{G}_N},A]=\;&\frac{2}{N}\sum_{p,q \in\Lambda_+^* :p+q \ne 0} \big(\widehat{V} (p/N)+ \widehat{V} ((p+q)/N) \big)\eta_p \\ &\hspace{2.5cm} \times \Big[\sigma_q^2+\big(\gamma_q^2+\sigma_{q}^2\big)b^*_qb_q+\gamma_q\sigma_q\big(b_qb_{-q}+b^*_q b^*_{-q}\big)\Big]\\
    &+\frac{2}{N}\sum_{\substack{p,q\in\Lambda_+^*\\p+q\ne0}}\bigg[\widehat{V} (p/N)+\widehat{V} ((p+q)/N)\bigg]\sigma_q\eta_p\frac{2\eta_{q+p}(q+p)^2-2\sigma_q(p\cdot q)}{p^2+q^2+(p+q)^2} \\&+  \cB+ \cD  + \cE_{[\mathcal{C},A]},
    \end{split}
    \end{equation}
    where
    \begin{equation}\label{eq:B-def} 
    \cB = \sum_{p \in \Lambda^*_+} O_{p} \big( b_p^* b_{-p}^* + b_p b_{-p} \big) \end{equation} 
    is a quadratic operator, with coefficients satisfying
   \begin{equation}\label{eq:B-ass} \sum_{p \in \Lambda^*_+} \frac{|O_{p}|^2}{p^2} \leq \frac{C}{N} \,, \quad \|O \|_\infty \leq \frac{C}{N}\,\end{equation}  
    and where 
    \begin{equation}\label{eq:D-def} \mathcal{D}= \frac{1}{N} \sum_{\substack{r,p,q\in \Lambda_+^*\\r+p+q\ne0}} D_{r,p,q}\, b^*_{r+p} b^*_{q} b^*_{-p} b^*_{-r-q}+\mathrm{h.c.} \end{equation} 
    is a quartic operator with coefficients $D_{r,p,q}$ satisfying (\ref{eq:D-ass}). Furthermore, we have 
    \begin{equation}\label{eq:cEcC-bd}
      \big| \langle \xi, \cE_{[\mathcal{C},A]} \xi \rangle \big| \le \frac{C}{\sqrt{N}} \| \cN_+^{1/2} \xi \| \| \cK^{1/2} (\cN_+ + 1) \xi \|  +  \frac{C}{N} \langle \xi , (\mathcal{K}+1)(\mathcal{N}_++1)^2 \xi \rangle  
    \end{equation}
   for every $\xi \in \cF_\perp^{\leq N}$.  
\end{lemma}

\begin{proof}
From (\ref{eq:cC_G_N}), (\ref{eq:def_A}), we can decompose 
\begin{equation*}
    \big[\cC_{\cG_N} , A \big]= \sum_{i=1}^4 \Pi_i
\end{equation*}
with
\begin{equation*}
\begin{split}
\label{eq:Theta_dec}
\Pi_1=\;&\frac{1}{\sqrt{N}}\sum_{\substack{p,q\in \Lambda_+^*\\p+q\ne0}}\widehat{V} (p/N)  \gamma_q \big[b^*_{p+q}b^*_{-p} b_q,A_{\nu}\big]+\mathrm{h.c.}\\
\Pi_2=\;&\frac{1}{\sqrt{N}}\sum_{\substack{p,q\in \Lambda_+^*\\p+q\ne0}}\widehat{V} (p/N) \gamma_q \big[b^*_{p+q}b^*_{-p} b_q,A_{\gamma}\big]+\mathrm{h.c.}\\
\Pi_3=\;&\frac{1}{\sqrt{N}}\sum_{\substack{p,q\in \Lambda_+^*\\p+q\ne0}}\widehat{V} (p/N) \sigma_q \big[b^*_{p+q}b^*_{-p} b^*_{-q},A_{\nu}\big]+\mathrm{h.c.}\\
\Pi_4=\;&\frac{1}{\sqrt{N}}\sum_{\substack{p,q\in \Lambda_+^*\\p+q\ne0}}\widehat{V} (p/N) \sigma_q \big[b^*_{p+q}b^*_{-p} b^*_{-q},A_{\gamma}\big]+\mathrm{h.c.}.
\end{split}
\end{equation*}
We analyze the four terms separately. From (\ref{eq:CCR_b}), we obtain 
\begin{equation}
\label{eq:Pi_1}
\begin{split}
     \Pi_1=\;&\frac{1}{N}\sum_{\substack{p,q,r \in\Lambda_+^*\\p+q,\,r+q\ne0}}\widehat{V}(p/N) \gamma_q\big( \nu_{r,-q-r}+\nu_{q,r}+\nu_{r,q}\big)b^*_{-p-q}b^*_{p}b^*_{r+q}b^*_{-r}+\mathrm{h.c.}\\
     & +\frac{1}{N}\sum_{\substack{p,q,r\in\Lambda_+^*\\p+q,\,r+p\ne0}}\Big[\widehat{V} (p/N) +\widehat{V}((p+q)/N) \Big]\gamma_q\big( \nu_{r,-p-r}+\nu_{p,r}+\nu_{r,p}\big)\\
     &\qquad\qquad\qquad\times \,b^*_{p+q}b_{p+r}b_{q}b_{-r}+\mathrm{h.c.}\\
     & +\frac{2}{N}\sum_{\substack{p,q\in\Lambda_+^*\\p+q\ne0}} \Big[\widehat{V} (p/N)+\widehat{V} ((p+q)/N) \Big] \gamma_q\big(\nu_{p,q}+\nu_{q,p}+\nu_{p,-p-q}\big) b^*_qb^*_{-q}\\ &+ \widetilde{\cE}_{\Pi_1}
     \end{split}
 \end{equation} 
where the error $\wt{\mathcal{E}}_{\Pi_1}$ (and, similarly, the errors $\wt{\cE}_{\Pi_j}$, $j=2,3,4$, below) is due to the fact that the commutation relations (\ref{eq:CCR_b}) are not precisely canonical and can be estimated by 
\[ \pm \widetilde{\cE}_{\Pi_1} \leq \frac{C}{N} (\mathcal{K}+1)(\mathcal{N}_++1)^2. \]
Using the fact that $|\nu_{p,q}|\le C|p|^{-2}|q|^{-2}$, the second term on the r.h.s. of (\ref{eq:Pi_1}) can be bounded by 
\begin{equation}\label{eq:quartic-ex0}
\begin{split}
\Big| \frac{1}{N} &\sum_{\substack{p,q,r\in \Lambda_+^*\\p+q,\,r+p\ne0}} \Big[\widehat{V} (p/N) +\widehat{V}((p+q)/N) \Big] \gamma_q \big(\nu_{r,-p-r}+\nu_{p,r}+\nu_{r,p}\big) \langle \xi, b^*_{p+q}b_{p+r}b_{q}b_{-r} \xi \rangle \Big| \\ 
\leq\;& \frac{C}{N} \left[ \sum_{p,q,r} |r|^{-4}(|p|^{-4}+|p+r|^{-4}) \| b_{p+q} (\cN_+ + 1)^{1/2} \xi \|^2 \right]^{1/2}\\
&\qquad\times\left[ \sum_{p,q,r} \| b_{p+r} b_{q} b_{-r} (\cN_+ + 1)^{-1/2} \xi \|^2 \right]^{1/2} \\
\leq\;& \frac{C}{N} \langle\xi , (\cN_+ + 1)^2 \xi \rangle.
\end{split}
\end{equation} 
As for the third term on the r.h.s. of (\ref{eq:Pi_1}), we notice that $|\nu_{p,-p-q}| \le C|p|^{-2}|p+q|^{-2}$ implies that the quadratic term proportional to $\nu_{p,-p-q}$ is bounded by $CN^{-1} \cK\mathcal{N}_+$. To handle the other contributions, we write 
\begin{equation}
\label{eq:nu+nu}
    \nu_{p,q}+\nu_{q,p}=\eta_p \sigma_q -\frac{2q^2}{p^2+q^2+|p+q|^2}[\eta_p(\sigma_q-\eta_q) -\eta_{q}(\sigma_p-\eta_p)] -\eta_p\sigma_q\frac{2p\cdot q}{p^2+q^2+|p+q|^2}.
\end{equation}
The contributions proportional to $\eta_q(\sigma_p-\eta_p)$ and $\eta_p(\sigma_q-\eta_q)$ are bounded by $CN^{-1}(\mathcal{N}_++1)$, since $|\sigma_p-\eta_p|\le C|p|^{-4}$. To bound the contribution of the last term, we distinguish $|p|>N$ and $|p|\leq N$. For $|p|>N$, we estimate 
\begin{equation}
\label{eq:Pi_1quadratic}
\begin{split}
\Big|\frac{2}{N}\sum_{\substack{p,q\in \Lambda_+^*\\|p|>N}}& \Big(\widehat{V} (p/N)+\widehat{V} ((p+q)/N) \Big) \gamma_q\eta_p\sigma_q\frac{2p\cdot q}{p^2+q^2+|p+q|^2} \langle \xi,b^*_qb^*_{-q}\xi \rangle \Big| \\ 
\leq&\; \frac{C}{N}\sum_{\substack{q\in \Lambda_+^*}}|q|\,|\sigma_q| \|(\cN_++1)^{1/2}\xi\| \|b_q\xi\| \leq \frac{C}{N} \langle \xi,(\cK+1)\xi\rangle
\end{split}
\end{equation}
since $\sum_{|p|>N}|\eta_p| |p|^{-1}\leq C$.
For $|p|\leq N$, we proceed with a change of variable $p\rightarrow -p$ and obtain 
\[
\begin{split}
\frac{2}{N} &\sum_{p,q\in \Lambda_+^* , |p|\leq N} \eta_p\sigma_q(p\cdot q)\bigg[\frac{\widehat{V} (p/N)+\widehat{V} ((p-q)/N) }{p^2+q^2+|p-q|^2}-\frac{\widehat{V} (p/N)+\widehat{V} ((p+q)/N) }{p^2+q^2+|p+q|^2} \bigg] \\ &\hspace{9cm} \times \gamma_q \langle \xi, \big( b^*_qb^*_{-q}  + \text{h.c.} \big) \xi \rangle\\ =&\;\frac{2}{N} \sum_{\substack{p,q\in \Lambda_+^*\\|p|\leq N}}[\widehat{V}((p-q)/N)-\widehat{V}((p+q)/N)]\frac{\eta_p\sigma_q(p\cdot q)}{p^2+q^2+|p+q|^2}\gamma_q \langle \xi, \big( b^*_qb^*_{-q}+\mathrm{h.c.}\big) \xi \rangle \\ &\;+\frac{8}{N}\sum_{\substack{p,q\in \Lambda_+^*\\|p|\leq N}}\frac{\Big[\widehat{V}(p/N)+\widehat{V}((p+q)/N)\Big]\eta_p\sigma_q(p\cdot q)^2}{(p^2+q^2+|p+q|^2)(p^2+q^2+|p-q|^2)}\gamma_q \langle \xi , \big(b^*_qb^*_{-q}+\mathrm{h.c.}\big) \xi \rangle .
\end{split}
\]
Using Lipschitz contiuity of $\widehat{V}$ in the first term on the r.h.s and the bound $\sum_{p}|p+q|^{-2}|p|^{-2}\leq C|q|^{-1}$ in the second term, we obtain
\[
\begin{split}
\Big| \frac{2}{N} &\sum_{p,q\in \Lambda_+^* , |p|\leq N} \eta_p\sigma_q(p\cdot q)\bigg[\frac{\widehat{V} (p/N)+\widehat{V} ((p-q)/N) }{p^2+q^2+|p-q|^2}-\frac{\widehat{V} (p/N)+\widehat{V} ((p+q)/N) }{p^2+q^2+|p+q|^2} \bigg] \\ &\hspace{9cm} \times \gamma_q \langle \xi, \big( b^*_qb^*_{-q}  + \text{h.c.} \big) \xi \rangle \Big| \\ 
&\hspace{2cm} \leq \; \frac{C}{N} \sum_{q \in \Lambda^*_+} \sigma_q |q| \langle \xi, \big( b^*_qb^*_{-q}  + \text{h.c.} \big) \xi \rangle  \leq \frac{C}{N} \langle \xi, (\cK + 1) \xi \rangle
\end{split} \] 
proceeding similarly as in (\ref{eq:Pi_1quadratic}). 

Noticing that the coefficients $D^{(1)}_{q,r,p} = \widehat{V} (p/N) \gamma_q\big( \nu_{r,-q-r}+\nu_{q,r}+\nu_{r,q}\big)$ satisfy the assumption (\ref{eq:D-ass}), we denote by $\cD^{(1)}$ the quartic operator on the first line on the r.h.s. of (\ref{eq:Pi_1}); it will be absorbed in (\ref{eq:D-def}). We conclude that 
\begin{equation}
\label{eq:Pi1-fin} 
    \Pi_1=\frac{1}{N} \sum_{\substack{p,q\in\Lambda_+^*\\p+q\ne0}}\Big[\widehat{V} (p/N)+\widehat{V} ((p+q)/N) \Big] \eta_p\gamma_q \sigma_ q\, b^*_{q} b^*_{-q} +\mathrm{h.c.}+\mathcal{D}^{(1)} +\mathcal{E}_{\Pi_1},
\end{equation}
where 
\[ \pm \cE_{\Pi_1} \leq \frac{C}{N} (\cK+1) (\cN_+ + 1)^2. \]

Let us now consider $\Pi_2$. From (\ref{eq:CCR_b}), we obtain 
 \begin{equation}
 \label{eq:Pi_2} 
     \begin{split}
     \Pi_2=\;&\frac{1}{N}\sum_{\substack{p,r,v\in\Lambda_+^*\\r+v,\,p-v\ne0}}\widehat{V} (p/N) (\eta_v +\eta_{r+v} ){\gamma}_r\gamma_v\,b^*_{r+v}b^*_{p-v}b^*_{-p}b_r+\mathrm{h.c.}\\ 
     &-\frac{1}{N}\sum_{\substack{q,r,v\in\Lambda_+^*\\r+v,\,q+v\ne0}}\bigg[\widehat{V}(v/N)+\widehat{V}((v+q)/N) \bigg]\gamma_q{\gamma}_v\eta_r\, b^*_{-r}b^*_{r+v}b^*_{-v-q}b_{-q}+\mathrm{h.c.}\\
     & +\frac{1}{N}\sum_{\substack{q,r,v\in\Lambda_+^*\\r+v,\,q+v\ne0}} \bigg[\widehat{V} (v/N)+\widehat{V}((v+q)/N) \bigg](\eta_v+\eta_{v+r}) {\gamma}_r\gamma_q\, b^*_{r+v}b^*_{q}b_rb_{q+v}+\mathrm{h.c.} \\
     &-\frac{1}{N}\sum_{\substack{p,r,v\in\Lambda_+^*\\r+v,\,p+v\ne0}}\widehat{V} (p/N) \eta_r {\gamma}_v^2\,b^*_{r+v}b^*_{-r}b_{p+v}b_{-p}+\mathrm{h.c.}\\
     & +\frac{2}{N}\sum_{\substack{q,v\in\Lambda_+^*\\q+v\ne0}}\bigg[\widehat{V} (v/N)+\widehat{V} ((v+q)/N)\bigg]\eta_v{\gamma}_q^2\,b^*_qb_{q}+ \wt{\cE}_{\Pi_2}
     \end{split}
 \end{equation}
 with $\pm \widetilde{\cE}_{\Pi_2} \leq C N^{-1} (\mathcal{K}+1)(\mathcal{N}_++1)^2$. The first term on the r.h.s. can be controlled, using Cauchy-Schwarz, by 
 \begin{equation}\label{eq:quartic-ex} \begin{split} 
 \Big| \frac{1}{N}\sum_{p,r,v} & \widehat{V} (p/N) (\eta_v +\eta_{r+v} ){\gamma}_r\gamma_v\, \langle \xi, b^*_{r+v}b^*_{p-v}b^*_{-p}b_r \xi \rangle \Big| \\ \leq \; &\frac{C}{N} \left[ \sum_{p,r,v} p^2 \| b_{r+v} b_{p-v} b_{-p} \xi \|^2 \right]^{1/2} \left[ \sum_{p,r,v}  \frac{|\widehat{V} (p/N)|^2}{p^2} (\eta_v^2 + \eta_{r+v}^2) \| b_r \xi \|^2 \right]^{1/2} \\ \leq \; &\frac{C}{\sqrt{N}} \| \cK^{1/2} \cN_+ \xi \| \| \cN_+^{1/2} \xi \| \end{split} \end{equation}
The other quartic terms on the r.h.s. of (\ref{eq:Pi_2}) can be bounded similarly. We conclude that 
\begin{equation}\label{eq:Pi2-fin} \Pi_2 = \frac{2}{N}\sum_{\substack{q,v\in\Lambda_+^*\\q+v\ne0}}\bigg[\widehat{V} (v/N)+\widehat{V} ((v+q)/N)\bigg]\eta_v{\gamma}_q^2\,b^*_qb_{q} +  \cE_{\Pi_2} \end{equation} 
where
\[ |\langle \xi, \cE_{\Pi_2} \xi \rangle | \leq \frac{C}{\sqrt{N}}  \| \cN_+^{1/2} \xi \| \| \cK^{1/2} \cN_+ \xi \| + \frac{C}{N} \langle \xi,  (\cK+1) (\cN_+ + 1)^2 \xi \rangle  \]
for all $\xi \in \cF_\perp^{\leq N}$.

 Next, we consider $\Pi_3$. With (\ref{eq:CCR_b}), we find 
\begin{equation} \label{eq:Pi_3}
\begin{split}
    \Pi_3=\;& \frac{1}{N} \sum_{\substack{p,q,r\in\Lambda_+^*\\p+q,\,q+r\ne0}} \widehat{V}(p/N) \sigma_q \big( \nu_{r,-q-r}+\nu_{q,r} + \nu_{r,q}\big) b^*_{p+q} b^*_{-p} b_{-r} b_{r+q} + \mathrm{h.c.}\\
    &+ \frac{1}{N} \sum_{\substack{p,q,r\ne0\\p+q,\,p+r\ne0}} \Big[ \widehat{V}(p/N) + \widehat{V}((p+q)/N) \Big] \sigma_q \big(\nu_{r,-p-r} + \nu_{p,q}+ \nu_{r,p}\big) \\
    &\qquad\qquad\qquad\qquad\qquad\times b^*_{p+q} b^*_{-q} b_{-r} b_{r+p}+\mathrm{h.c.}\\
    &+\frac{2}{N} \sum_{\substack{p,q\in\Lambda_+^*\\p+q\ne0}} \Big[ \widehat{V}(p/N) + \widehat{V}((p+q)/N) \Big] \sigma_q \\
    &\qquad\qquad\times\big( \nu_{-p-q,q}+\nu_{q,-p-q}+\nu_{p,-p-q}+\nu_{-p-q,p}+\nu_{p,q}+\nu_{q,p}\big) b^*_{-p}b_{-p}\\
    &+\frac{2}{N} \sum_{\substack{p,q\in\Lambda_+^*\\p+q\ne0}}  \Big[\widehat{V}(p/N)+\widehat{V}((p+q)/N)\Big] \sigma_q  \big(\nu_{-p-q,p}+\nu_{p,q}+\nu_{q,p}\big) b^*_{-q} b_{-q} \\
    &+\frac{2}{N} \sum_{\substack{p,q\in\Lambda_+^*\\p+q\ne0}}  \Big[\widehat{V}(p/N)+\widehat{V}((p+q)/N)\Big] \sigma_q  \big(\nu_{-p-q,p}+\nu_{p,q}+\nu_{q,p}\big)+\widetilde{\cE}_{\Pi_3} 
\end{split}
\end{equation}
where $\pm \widetilde{\cE}_{\Pi_3} \le CN^{-1}(\mathcal{K}+1)(\mathcal{N}_++1)^2$. The first term on the r.h.s. can be handled similarly to (\ref{eq:quartic-ex}). The second term can be bounded analogously to (\ref{eq:quartic-ex0}). Let us now consider the quadratic terms. 

The term proportional to $b^*_{-p}b_{-p}$ is controlled by $C N^{-1}(\mathcal{N}_++1)$, since the sum over $q$ can be bounded using $|\sigma_q| \leq C |q|^{-2}$, $|\nu_{r,s}|\leq C |r|^{-2} |s|^{-2}$. In the next quadratic term, proportional to $b^*_{-q} b_{-q}$, the bound $|\nu_{-p-q,p}|\le C |p|^{-2}|p+q|^{-2}$ allows us to estimate the corresponding contribution by $C N^{-1}(\mathcal{N}_++1)$. To handle the contributions proportional to $\nu_{p,q}$, $\nu_{q,p}$, we recall (\ref{eq:nu+nu}). The contribution of all terms on the r.h.s. of  (\ref{eq:nu+nu}), with the exception of the term proportional to $\eta_p \sigma_q$, can be estimated by $C N^{-1} \log N \mathcal{N}_+$ and can therefore be neglected. 

As for the constant term on the r.h.s. of (\ref{eq:Pi_3}), we write 
\[
\begin{split}
\sigma_q  \big(\nu_{-p-q,p}+\nu_{p,q}+\nu_{q,p}\big) = &\, \eta_p\s_q^2 + \sigma_q\eta_p\frac{2\eta_{p+q} (p+q)^2 - 2\sigma_q \,(p\cdot q)}{p^2+q^2+(p+q)^2}\\
& \,+ \frac{2q^2 \sigma_q}{p^2+q^2+(p+q)^2}[\eta_q(\sigma_p-\eta_p) + \eta_p (\eta_q-\sigma_q)]\\
& \, + \frac{2(p+q)^2 \sigma_q}{p^2+q^2+(p+q)^2}\eta_{p+q}(\sigma_p-\eta_p) 
\end{split}
\]
and we notice that the contribution associated with the last two lines is of the order $N^{-1}$ (because these terms are all summable over $p,q$). We conclude that 
\begin{equation}\label{eq:Pi3-fin} \begin{split} \Pi_3 = \; &\frac{2}{N}\sum_{\substack{p,q\in\Lambda_+^*\\p+q\ne0}}\bigg[\widehat{V} (p/N)+\widehat{V} ((p+q)/N)\bigg]\eta_p\sigma_q^2\\&+\frac{2}{N}\sum_{\substack{p,q\in\Lambda_+^*\\p+q\ne0}}\bigg[\widehat{V} (p/N)+\widehat{V} ((p+q)/N)\bigg]\sigma_q\eta_p\frac{2\eta_{q+p}(q+p)^2-2\sigma_q(p\cdot q)}{p^2+q^2+(p+q)^2}\\
     &+\frac{2}{N} \sum_{\substack{p,q\in\Lambda_+^*\\p+q\ne0}}\bigg[\widehat{V}(v/N)+\widehat{V} ((v+q)/N) \bigg] \eta_v \sigma_q^2\,b^*_q b_q + \cE_{\Pi_3}
     \end{split} \end{equation} 
where 
\[ |\langle \xi, \cE_{\Pi_3} \xi \rangle| \leq \frac{C}{\sqrt{N}} \| \cN_+^{1/2} \xi \|  \| \cK^{1/2} (\cN_+ +1) \xi \| 
+ \frac{C}{N} \langle \xi, (\cK+1 ) (\cN_+ + 1)^2 \xi \rangle \]
for all $\xi \in \cF_\perp^{\leq N}$.

Finally, we consider $\Pi_4$. Again with (\ref{eq:CCR_b}), we find 
\begin{equation}
 \label{eq:Pi_4}
     \begin{split}
     \Pi_4=\;&-\frac{1}{N}\sum_{\substack{p,r,v\in\Lambda_+^*\\r+v,\,p+v\ne0}}\widehat{V} (p/N) \eta_r {\gamma}_v\sigma_v\, b^*_{r+v}b^*_{-r}b^*_{p-v}b^*_{-p}+\mathrm{h.c.}\\
      & -\frac{1}{N}\sum_{\substack{p,r,v\in\Lambda_+^*\\r+v,\,q+v\ne0}} \bigg[\widehat{V} (v/N)+\widehat{V} ((v+p)/N) \bigg]\eta_r {\gamma}_v \sigma_p\, b^*_{r+v}b^*_{-r}b^*_{p}b^*_{-p-v}+\mathrm{h.c.}\\
    &+\frac{1}{N}\sum_{\substack{p,r,v\in\Lambda_+^*\\r+v,\,p+v\ne0}}\widehat{V} (p/N)\big(\eta_v +\eta_{r+v}){\gamma}_r\sigma_v\,b^*_{r+v}b_{p+v}b_{-p}b_r+\mathrm{h.c.}\\ 
     & +\frac{1}{N}\sum_{\substack{p,r,v\in\Lambda_+^*\\r+v,\,p+v\ne0}} \bigg[\widehat{V} (v/N)+\widehat{V}((v+p)/N) \bigg] \big(\eta_v +\eta_{r+v})\sigma_p {\gamma}_r \,b^*_{r+v}b_{p+v}b_{-p}b_{r}+\mathrm{h.c.}\\
         & +\frac{1}{N}\sum_{\substack{q,v\in\Lambda_+^*\\q+v\ne0}}\bigg[\widehat{V} (v/N)+\widehat{V} ((v+q)/N) \bigg]\eta_v{\gamma}_q\sigma_q\,b^*_qb^*_{-q}+\mathrm{h.c.}\\
     &+\frac{1}{N}\sum_{\substack{r,v\in\Lambda_+^*\\r+v\ne0}}\bigg[\widehat{V} (r/N)+\widehat{V} ((r+v)/N)\bigg](\eta_{r+v}+\eta_{v}){\gamma}_{r}\sigma_v\,b_{r}b_{-r}+\mathrm{h.c.}+\wt{\cE}_{\Pi_4}
     \end{split}
 \end{equation}
with $\pm \wt{\cE}_{\Pi_4} \le CN^{-1}(\mathcal{K}+1)(\mathcal{N}_++1)^2$. Since the coefficients $D^{(2)}_{v,r,p} = - \widehat{V} (p/N) \eta_r \gamma_v \sigma_v$ and $D^{(3)}_{v,r,p} = \big[\widehat{V} (v/N)+\widehat{V} ((v+p)/N) \big] \eta_r {\gamma}_v \sigma_p$ satisfy the assumption (\ref{eq:D-ass}), we denote by $\cD^{(2)}$ and $\cD^{(3)}$ the quartic operators on the first and second line of (\ref{eq:Pi_4}); we will absorb them in (\ref{eq:D-def}). The other quartic contributions to $\Pi_4$ can be bounded similarly to (\ref{eq:quartic-ex0}) (term on the fourth line) and (\ref{eq:quartic-ex}) (term on the third line). As for the quadratic terms, the contribution proportional to $\eta_{r+v}$ in the last line is small, since 
\begin{equation*}
    \begin{split}
        \pm\frac{1}{N}\sum_{\substack{r,v\in\Lambda_+^*\\r+v\ne0}}&\bigg[\widehat{V} (r/N) +\widehat{V}((r+v)/N) \bigg]\eta_{r+v}{\gamma}_{r}\sigma_v\big(\langle \xi,b_{r}b_{-r}\xi\rangle+\mathrm{h.c.}\big)\\
        \leq\;& \frac{C}{N}\sum_{\substack{r,v\in\Lambda_+^*\\r+v\ne0}}\frac{1}{|r+v|^2|v|^2}\|b_r\xi\|\,\big\|(\cN_++1)^{1/2}\xi\big\| \\ 
        \leq\;& \frac{C}{N}\sum_{r\in\Lambda_+^*}\frac{1}{r}\|b_r\xi\|\,\big\|(\cN_++1)^{1/2}\xi\big\|  \leq  \frac{C}{N}\langle \xi, (\cK+ \cN_++1)\xi \rangle.
    \end{split}
\end{equation*}
Similarly, in the contribution proportional to $\eta_v$ in the last line, we can replace $\gamma_r$ by $1$, up to an error bounded by $N^{-1} (\cN_++1)$. Observing that the coefficients  
\[ O_r = \frac{1}{N} \sum_{v \in \L^*_+ :  r+v \not = 0} \bigg[\widehat{V} (r/N)+\widehat{V} ((r+v)/N) \bigg]\eta_{v}\sigma_v \]
satisfy (\ref{eq:B-ass}), we conclude that 
 \begin{equation}
 \label{eq:Pi4-fin}
     \begin{split}
     \Pi_4=\;&\frac{1}{N}\sum_{\substack{q,v\in\Lambda_+^*\\q+v\ne0}}\bigg[\widehat{V} (v/N)+\widehat{V}((v+q)/N) \bigg]\eta_v{\gamma}_q\sigma_q\,b^*_qb^*_{-q}+\mathrm{h.c.}\\
     &+\frac{1}{N}\sum_{\substack{r,v\in\Lambda_+^*\\r+v\ne0}}\bigg[\widehat{V} (r/N) + \widehat{V}((r+v)/N)\bigg]\eta_{v}\sigma_v\,b^*_r b^*_{-r}+\mathrm{h.c.} \\ &+ \cB+ \cD^{(2)} + \cD^{(3)}  + \cE_{\Pi_4}
     \end{split}
 \end{equation}
where 
\[ |\langle \xi, \cE_{\Pi_4} \xi \rangle | \leq \frac{C}{\sqrt{N}} \| \cN^{1/2}_+ \xi \| \| \cK^{1/2} (\cN_+ + 1) \xi \| + \frac{C}{N} \langle \xi, (\cK+1)( \cN_+ + 1)^2 \xi \rangle  \]
for all $\xi \in \cF^{\leq N}_{\perp}$, where $\cB$ has the form (\ref{eq:B-def}), with coefficients satisfying the corresponding conditions (\ref{eq:B-ass}), and where $\cD^{(2)}, \cD^{(3)}$ have the form (\ref{eq:D-def}), with coefficients satisfying (\ref{eq:D-ass}). Combining (\ref{eq:Pi1-fin}), (\ref{eq:Pi2-fin}), (\ref{eq:Pi3-fin}), (\ref{eq:Pi4-fin}), we obtain (\ref{eq:commut_C,A}).
\end{proof}

Making use of the last lemma, we can compute $e^{-A} \cC_{\cG_N} e^A$, up to small errors.
\begin{prop} \label{prop:A_on_cubic}
 Let $\cC_{\cG_N}$ be defined as in \eqref{eq:cC_G_N} and $A$ as in \eqref{eq:def_A}. Under the same assumptions as Theorem \ref{thm:cubic}, we have 
    \begin{equation}\label{eq:eCe-fin} 
    \begin{split}
        &e^{-A} \mathcal{C}_{\mathcal{G}_N}e^{A}\\
         &=  \, \cC_{\cG_N} \\
        &\hspace{.2cm} + \frac{2}{N}\sum_{\substack{p,q\in\Lambda_+^*\\p+q \ne0}}\Big(\widehat{V} (p/N)+\widehat{V} ((p+q)/N) \Big) \eta_p\Big[\sigma_q^2+\big(\gamma_q^2+\sigma_{q}^2\big)b^*_qb_q+\gamma_q\sigma_q\big(b_qb_{-q}+b^*_q b^*_{-q}\big)\Big]\\
    &\hspace{.2cm} +\frac{2}{N}\sum_{\substack{p,q\in\Lambda_+^*\\p+q\ne0}}\bigg[\widehat{V} (p/N)+\widehat{V} ((p+q)/N)\bigg]\sigma_q\eta_p\frac{2\eta_{q+p}(q+p)^2-2\sigma_q(p\cdot q)}{p^2+q^2+(p+q)^2} +\mathcal{E}_\mathcal{C_{\cG_N}},
    \end{split}
    \end{equation}
    where, for every $\eps > 0$ small enough, we have 
   \[ \pm \cE_{\cC_{\cG_N}} \leq \eps \cK + \frac{C}{\eps N}  (\cH_N + 1) ( \cN_+ + 1)^4 \]   
    for all $\xi \in \cF^{\leq N}_\perp$. 
    \end{prop}

\begin{proof} 
We write 
\begin{equation}\label{eq:eCe-1} 
e^{-A} \cC_{\cG_N} e^A = \cC_{\cG_N} + \int_0^1 ds \, e^{-sA} \big[ \cC_{\cG_N} , A \big] e^{sA}. \end{equation} 
From Lemma \ref{lemma:commut_C,A}, we recall the identity  
\[ \begin{split}
    [ \mathcal{C}_{\mathcal{G}_N},A]=\;&\frac{2}{N}\sum_{p,q \in\Lambda_+^* :p+q \ne 0} \big(\widehat{V} (p/N)+ \widehat{V} ((p+q)/N) \big)\eta_p \\ &\hspace{2.5cm} \times \Big[\sigma_q^2+\big(\gamma_q^2+\sigma_{q}^2\big)b^*_qb_q+\gamma_q\sigma_q\big(b_qb_{-q}+b^*_q b^*_{-q}\big)\Big]\\
    &+\frac{2}{N}\sum_{\substack{p,q\in\Lambda_+^*\\p+q\ne0}}\bigg[\widehat{V} (p/N)+\widehat{V} ((p+q)/N)\bigg]\sigma_q\eta_p\frac{2\eta_{q+p}(q+p)^2-2\sigma_q(p\cdot q)}{p^2+q^2+(p+q)^2} \\& +  \cB + \cD + \cE_{[\mathcal{C},A]},
    \end{split} \]
 where $\cB$ is a quadratic operator having the form (\ref{eq:B-def}), $\cD$ is a quartic operator like (\ref{eq:D-def}) and $\cE_{[\cC,A]}$ satisfies the bounds \eqref{eq:cEcC-bd} . From Prop. \ref{prop:quad-A}, we find 
\[ \begin{split}  \frac{2}{N} \sum_{p,q} \big( \widehat{V} &(p/N) + \widehat{V} ((p+q)/N) \big) \eta_p \, e^{-sA}  \big[ (\g_q^2 + \s^2_q) b_q^*b_q + \g_q \s_q (b_q b_{-q} + b_q^* b_{-q}^* ) \big] e^{sA} \\ &= \frac{2}{N} \sum_{p,q} \big( \widehat{V} (p/N) + \widehat{V} ((p+q)/N) \big) \big[ (\g_q^2 + \s^2_q) b_q^*b_q + \g_q \s_q (b_q b_{-q} + b_q^* b_{-q}^* ) \big] + \cE_1 \end{split} \]
where, for any $\eps > 0$ small enough, 
\[  \pm \cE_1  \leq \eps \cN_+ + \frac{C}{\eps N} (\cN_+ + 1)^2 . \]
Using (\ref{eq:B-ass}) and proceeding as in the proof of (\ref{eq:eT2e}) (applying Lemma \ref{lm:quad-A} and using that the coefficients $O_{p}$ in (\ref{eq:B-def}) satisfy (\ref{eq:B-ass})), we obtain that, for any $\eps > 0$, 
\[ \pm  e^{-sA} \cB e^{sA}  \leq \eps \cK + \frac{C}{\eps N}  (\cH_N + \cN_+^2 + 1).  \]
Moreover, proceeding as in Prop. \ref{prop:T4-bd} (we can apply Lemma \ref{lemma:commut_quartic_A_1}, because the coefficients $D_{r,p,q}$ in (\ref{eq:D-def}) satisfy (\ref{eq:D-ass})), we obtain that, for any $\eps > 0$, 
\[ \pm e^{-sA} \cD e^{sA}   \leq  \eps \cK + \frac{C}{\eps N} (\cH_N +1) (\cN_+ + 1)^4 .\] 
From Prop. \ref{prop:a_priori_A}, we also get 
\[ | \langle \xi, e^{-sA} \cE_{[\cC,A]} e^{sA} \xi \rangle | \leq   \frac{C}{\sqrt{N}} \| \cN_+^{1/2} \xi \| \| (\cH_N +1)^{1/2} (\cN_+ + 1)^{3/2} \xi \|  + \frac{C}{N} \langle \xi, (\cH_N +1) (\cN_+ + 1)^3 \xi \rangle. \]
Inserting in (\ref{eq:eCe-1}), we obtain (\ref{eq:eCe-fin}). 
\end{proof}

\subsection{Control of $\cH_N = \cK + \cV_N$} 

Finally, we conjugate the Hamiltonian $\cH_N$. Besides Lemma \ref{lemma:K,V,A}, we will also need estimates for the the commutators of the terms $[\cK, A]_2, [\cV_N ,A]_2$, defined in (\ref{eq:[K,A]2}) and, respectively, (\ref{eq:[VN,A]2}), with $A$.
\begin{lemma} \label{lemma:second_commut_H_A} 
Let $A$ be defined as in \eqref{eq:def_A} and $[\mathcal{K},A]_2$, $[\mathcal{V}_N,A]_2$ be defined as in Lemma~\ref{lemma:K,V,A}. Then
\begin{equation} \label{eq:K,A_2,A}
        \pm \big[[\mathcal{K},A]_2,A\big] + \mathrm{h.c.}\leq \frac{C (\log N)^{1/2}}{N}(\mathcal{K}+1)(\mathcal{N}_++1)^2.
\end{equation}
Moreover, 
\begin{equation} \label{eq:V,A_2,A}
    \begin{split}
        \big[[\mathcal{V}_N,A]_2,A\big] + \mathrm{h.c.}  =&\;\frac{2}{N^2}\sum_{\substack{p,q\in\Lambda_+^*\\p+q\ne0}}\bigg[\big(\widehat{V}(\cdot/N)\ast\eta\big)_p +\big(\widehat{V}(\cdot/N)\ast\eta\big)_{p+q}\bigg] \\&\hspace{3cm}\times\,\eta_p\sigma_q\frac{2\eta_{q+p}(q+p)^2-2\sigma_q(p\cdot q)}{p^2+q^2+(p+q)^2} \\&+\mathcal{E}_{[[\mathcal{V}_N,A]_2,A]}
    \end{split}
\end{equation}
with 
\begin{equation*} \label{eq:err_V,A,A}
    \pm\mathcal{E}_{[[\mathcal{V}_N,A]_2,A]} \le \frac{C}{N}(\cK+1)(\mathcal{N}_++1)^2.
\end{equation*}
\end{lemma}

\begin{proof}

    We start by proving \eqref{eq:K,A_2,A}. Recalling the definitions \eqref{eq:def_A} and (\ref{eq:[K,A]2}) 
    of $A$ and $[\cK,A]_2$, we can split 
     \begin{equation*}
        \big[[\mathcal{K},A]_2,A\big] + \text{h.c.} =\sum_{j=1}^4 T_j
    \end{equation*}
    with
    \begin{equation*}
        \begin{split}
            T_1=\;&\frac{2}{N} \sum_{\substack{r,v,p,q\in \Lambda_+^*\\r+v,p+q\ne0}}r\cdot v\,\eta_r  \gamma_v  \eta_p \gamma_q \big[b^*_{r+v}b^*_{-r} b_v,b^*_{p+q}b^*_{-p}b_q\big]+\mathrm{h.c.}\\            
            T_2=\;&\frac{2}{N} \sum_{\substack{r,v,p,q\in \Lambda_+^*\\r+v,p+q\ne0}}r\cdot v\,\eta_r \gamma_v\nu_{p,q} \big[b^*_{r+v}b^*_{-r} b_v,b^*_
            {p+q}b^*_{-p}b^*_{-q}\big]+\mathrm{h.c.}\\
            T_3=\;&\frac{2}{N} \sum_{\substack{r,v,p,q\in \Lambda_+^*\\r+v,p+q\ne0}}r\cdot v\,\eta_r \gamma_v \eta_p   \gamma_q \big[b^*_v b_{r+v} b_{-r},b^*_{p+q} b^*_{-p} b_q \big]+\mathrm{h.c.}\\   
            T_4=\;&\frac{2}{N} \sum_{\substack{r,v,p,q\in \Lambda_+^*\\r+v,p+q\ne0}}r\cdot v\,\eta_r   \gamma_v \nu_{p,q} \big[b^*_v b_{r+v} b_{-r},b^*_{p+q} b^*_{-p} b^*_{-q}\big]+\mathrm{h.c.} \\      
        \end{split}
    \end{equation*}
 Up to terms of lower order (due to the fact that the operators $b,b^*$ do not exactly satisfy canonical commutation relations), the operators $T_1, T_2$ only contain quartic contributions, satisfying (\ref{eq:K,A_2,A}). As an example, consider the term 
 \[ T_{1,1} = \frac{2}{N} \sum_{r,v,p} r\cdot v \, \eta_r  \g_v\eta_v \g_{v-p}\, b^*_{r+v} b^*_{-r} b^*_{-p} b_{v-p}+\mathrm{h.c.} \]
 contributing to $T_1$, which can be bounded by 
 \begin{equation}\label{eq:T11-bd} \begin{split} |\langle \xi, T_{1,1} \xi \rangle | &\leq \frac{C}{N} \Big[ \sum_{r,v,p} (r+v)^2 \| b_{r+v} b_{-r} b_{-p} (\cN_+ + 1)^{-1/2} \xi \|^2 \Big]^{1/2} \\ &\hspace{3cm} \times  \Big[ \sum_{r,v,p} \frac{(r\cdot v)^2}{(r+v)^2} \eta_r^2 \eta_v^2 \| b_{v-p} (\cN_+ + 1)^{1/2} \xi \|^2 \Big]^{1/2} \\  &\leq \frac{C (\log N)^{1/2}}{N} \| \cK^{1/2} (\cN_+ +1)^{1/2} \xi \| \| (\cN_++1) \xi \| \\ &\leq \frac{C (\log N)^{1/2}}{N} \langle \xi, (\cK + 1) (\cN_+ + 1) \xi \rangle \end{split}  \end{equation} 
where we estimated
\[ \sum_{r,v} \frac{(r\cdot v)^2}{(r+v)^2} \eta_r^2 \eta_v^2 \leq C \sum_{r,v} \frac{1}{(r+v)^2} \frac{1}{v^2} |\eta_r| \leq C \sum_r \frac{1}{r} |\eta_r|  \leq C \log N \]
as can be proven separating $|r| \leq N$ and $|r| > N$ (in the second region, we can apply the last bound in (\ref{eq:norms_eta})). All other contributions to $T_1, T_2$ can be handled similarly. Let us focus on the other terms. With (\ref{eq:CCR_b}), we find 
    \begin{equation}\label{eq:T4}
    \begin{split}
        T_3=\;&\frac{4}{N} \sum_{\substack{r,v\in \Lambda_+^*\\r+v\ne0}} r\cdot v\,\eta_r  \eta_{r+v} \gamma_v^2\, b^*_v b_v \\
        &+\frac{2}{N} \sum_{\substack{r,v,p \in \Lambda_+^*\\r+v,\,r+v-p\ne0}} r\cdot v \, \eta_r  \big( \eta_{r+v-p}  +\eta_{r+v} \big) \gamma_v \gamma_p \,b^*_v b^*_{p-r-v}b_{-r} b_p+\mathrm{h.c.}\\
        &+\frac{2}{N} \sum_{\substack{r,v,p \in \Lambda_+^*\\r+v,\,r-p\ne0}} r\cdot v \,\eta_r  \big( \eta_{r+p}  +\eta_r \big) \gamma_v \gamma_p\, b^*_v b^*_{p+r} b_{r+v} b_{p}+ \mathrm{h.c.}\\
        &-\frac{2}{N} \sum_{\substack{r,v,p\in \Lambda_+^*\\r+v,\,p+r\ne0}} r\cdot v\,\eta_r  \eta_p  \gamma_v^2\,b^*_{p+q} b^*_{-p} b_{r+v} b_{-r}+\mathrm{h.c.}+\wt{\mathcal{E}}_{T_3}
    \end{split}
    \end{equation}
    with  
    \begin{equation}\label{eq:wtcET} 
        \pm \wt{\mathcal{E}}_{T_3} \le \frac{C}{N}(\mathcal{K}+1)(\mathcal{N}_++1)^2.
    \end{equation}
Here, and in the rest of this proof, we will denote by $\wt{\cE}_{T_j}$ contributions due to the fact that the commutation relations (\ref{eq:CCR_b}) are not exactly canonical. All these contributions satisfy an estimate like (\ref{eq:wtcET}) and are therefore negligible. All quartic terms on the r.h.s. of (\ref{eq:T4}) are small. As an example, we can bound 
    \begin{equation} \label{eq:bound_quartic_T_3}
        \begin{split}
            \Big| \frac{2}{N} \sum_{r,v,p} &r\cdot v \, \eta_r^2 \gamma_v \gamma_p \, \langle \xi, b^*_v b^*_{p+r} b_{r+v} b_{p}\xi\rangle \Big| \\
            \le\;& \frac{C}{N} \bigg( \sum_{r,v,p} \frac{r^2 v^2}{|r+v|^2}\eta_r^2\,\| b_v b_{p+r}\xi\| \bigg)^{1/2} \bigg( \sum_{r,v,p} \eta_r^2 \, |r+v|^2 \,\|b_{r+v} b_p \xi\|^2\bigg)^{1/2}\\
            \le\;&\frac{C}{N}\langle\xi,(\mathcal{K}+1)(\mathcal{N}_++1)\xi\rangle.
        \end{split}
    \end{equation}
 All other quartic terms are bounded similarly. To estimate the quadratic term in the first line of \eqref{eq:T4}, we observe that, with the change of variables $-r-v = r'$ and $v = v'$, 
     \begin{equation}
    \label{eq:trickr_r+v}
        \frac{4}{N} \sum_{r,v} r\cdot v\,\eta_r \eta_{r+v}\gamma_v^2\;b^*_{v} b_v=-\frac{2}{N} \sum_{r,v} v^2\,\eta_r \eta_{r+v}\gamma_v^2\;b^*_{v} b_v,
    \end{equation}
which is clearly controlled by  $\cK/N$. Thus, 
    \begin{equation*} \label{eq:final_T_4}
        \big|\langle\xi,T_3\xi\rangle\big| \le  \frac{C (\log N)^{1/2}}{N} \big\langle \xi,(\mathcal{K}+1)(\mathcal{N}_++1)^2\xi\big\rangle.
    \end{equation*}

Next, let us consider the term $T_4$. With (\ref{eq:CCR_b}), we find  
    \begin{equation*}
    \begin{split}
        T_{4}=\;&\frac{2}{N} \sum_{\substack{r,v \in \Lambda_+^*\\r+v\ne0}} r\cdot v\, \eta_r \gamma_v  \big(\nu_{-r-v,r}+\nu_{r,-r-v}+\nu_{-r-v,v}+\nu_{v,-r-v}+ \nu_{r,v}+\nu_{v,r}\big)  \,b^*_v b^*_{-v} +\mathrm{h.c.}\\
        &+\frac{2}{N} \sum_{\substack{r,v,p \in \Lambda_+^*\\r+v,r+p\ne0}}r\cdot v \,\eta_r \gamma_v \big( \nu_{p,-r-p}+\nu_{r,p}+\nu_{p,r} \big)  \, b^*_v b^*_{r+p} b^*_{-p} b_{r+v}+\mathrm{h.c.}\\
        &+\frac{2}{N} \sum_{\substack{r,v,p \in \Lambda_+^*\\r+v,\,p-r-v \ne0}} r\cdot v\, \eta_r \gamma_v  \big( \nu_{p,r+v-p}+
        \nu_{-r-v,p}+\nu_{p,-r-v}\big) b^*_v b^*_{-p} b^*_{p-r-v} b_{-r}+\mathrm{h.c.}\\
        &+\widetilde{\mathcal{E}}_{T_4},
    \end{split}
    \end{equation*}
    with $\pm \widetilde{\mathcal{E}}_{T_4} \le C N^{-1} (\mathcal{K}+1)(\mathcal{N}_++1)^2$. All quartic terms can be bounded similarly as \eqref{eq:bound_quartic_T_3}. For the first two quadratic terms on the first line (the ones proportional to $\nu_{-r-v,r}$ and $\nu_{r,-r-v}$),  we notice that, for any $\varepsilon>0$, $\sum_r |r|^{-3}|r+v|^{-2} \le C |v|^{-2+\varepsilon}$ and therefore, choosing $\varepsilon< 1/2$,
    \begin{equation} \label{eq:-2+eps}
    \begin{split}
        &\Big|\frac{2}{N} \sum_{r,v} r\cdot v\,\eta_r\gamma_v ( \nu_{-r-v,r}+\nu_{r,-r-v}) \langle \xi, b^*_v b^*_{-v} \xi\rangle \Big| \\
        &\qquad \le\; \frac{C}{N} \sum_{v} |v|^{-1+\varepsilon} \|b_v\xi\|\|(\mathcal{N}_++1)^{1/2}\xi\|
        \le \frac{C}{N} \langle \xi, (\mathcal{K}+1)\xi\rangle.     
    \end{split}
    \end{equation}
     Using Eq. 
    \eqref{eq:nu+nu}, we rewrite the rest of the quadratic terms as
    \begin{equation*}
        \begin{split}
        &\frac{2}{N} \sum_{\substack{r,v \in \Lambda_+^*\\r+v\ne0}} r\cdot v\, \eta_r \gamma_v  \big(\nu_{-r-v,v}+\nu_{v,-r-v}+ \nu_{r,v}+\nu_{v,r}\big)  \,b^*_v b^*_{-v} +\mathrm{h.c.}\\
            = \, &\frac{2}{N} \sum_{\substack{r,v }} r\cdot v\, \eta_r \gamma_v\big(\eta_r+\eta_{r+v} \big)\sigma_v \;b^*_v b^*_{-v} \\
            &\;-\frac{4}{N}\sum_{\substack{r,v }} \; \frac{\eta_r^2\sigma_v \;(r\cdot v)^2}{r^2+v^2+|r+v|^2}\gamma_v\,b^*_v b^*_{-v}+\frac{4}{N}\sum_{\substack{r,v }}\;\frac{\eta_r\eta_{r+v}\sigma_v \;((r+v)\cdot v)(r\cdot v)}{r^2+v^2+|r+v|^2}\gamma_v\,b^*_v b^*_{-v} \\&\;-\frac{4}{N}\sum_{\substack{r,v }} \frac{(r\cdot v)v^2}{r^2+v^2+|r+v|^2}\, \eta_r \big((\eta_r+\eta_{r+v})(\sigma_v-\eta_v)\\&\hspace{5cm}+\eta_v(\sigma_{r+v}-\eta_{r+v})+\eta_v(\sigma_{r}-\eta_{r})\big) \gamma_v\, b^*_v b^*_{-v}+\mathrm{h.c.}
        \end{split}
    \end{equation*}
The terms in the last three lines can all be bounded, first summing over $r$ and then proceeding similarly as in (\ref{eq:-2+eps}), by $C N^{-1} (\cK+1)$. As for the first term on the r.h.s., the contribution proportional to $r \cdot v \, \eta^2_r \gamma_v \sigma_v$  vanishes, as can be seen replacing $r \to -r$. Moreover, switching $r \to -r - v$, 
\[ \begin{split} \frac{1}{N} \sum_{r,v} r \cdot v \, \eta_r \eta_{r+v} \gamma_v \sigma_v \big(b_v^* b_{-v}^* + \text{h.c.} \big) &= - \frac{1}{N} \sum_{r,v} (r+v) \cdot v \, \eta_r \eta_{r+v} \gamma_v \sigma_v \big(b_v^* b_{-v}^* + \text{h.c.} \big) \\ &= -\frac{1}{2N} \sum_{r,v} v^2 \eta_r \eta_{r+v} \gamma_v \sigma_v \big( b_v^* b_{-v}^* + \text{h.c.} \big) \end{split} \]
With $\sum_r | \eta_r|  |\eta_{r+v}| \leq C |v|^{-1}$, we can proceed as in (\ref{eq:-2+eps}) to prove that also this contribution is bounded by $C N^{-1} (\cK+ 1)$. We conclude that
    \begin{equation*}
        \pm T_4 \le \frac{C}{N}(\mathcal{K}+1)(\mathcal{N}_++1)^2.
    \end{equation*}

    Next, we show (\ref{eq:V,A_2,A}).
    We write 
    \[
    \big[ [\cV_N,A]_2, A \big] + \text{h.c.}  = \sum_{j=1}^3S_j
    \]
where
\begin{equation*}
    \begin{split}
        S_1 = &\, \frac{1}{N^{3/2}} \sum_{\substack{\;r,v,p,q\in \Lambda_+^*\\ r+v,p+q\ne 0}}\alpha_{r,v}\eta_p\gamma_q \,[b^*_{r+v}b^*_{-r} b^*_{-v}, b^*_{p+q}b^*_{-p}b_q]+\mathrm{h.c.}\\
        S_2 = &\, \frac{1}{N^{3/2}} \sum_{\substack{r,v,p,q\in \Lambda_+^*\\ r+v,p+q\ne 0}}\alpha_{r,v}\eta_p\gamma_q \,[b_{r+v}b_{-r} b_{-v}, b^*_{p+q}b^*_{-p}b_q]+\mathrm{h.c.}\\
        S_3 = &\, \frac{1}{N^{3/2}} \sum_{\substack{\;r,v,p,q\in \Lambda_+^*\\ r+v, p+q\ne 0}}\alpha_{r,v} \nu_{p,q} \,[b_{r+v}b_{-r} b_{-v}, b^*_{p+q}b^*_{-p}b^*_{-q}]+\mathrm{h.c.},\\
    \end{split}
\end{equation*}
and the coefficient $\alpha_{r,v}$ is defined as in \eqref{eq:defalpha}. With (\ref{eq:CCR_b}), we find 
\[\begin{split}
S_1 =&\,  -\frac{1}{N^{3/2}} \sum_{\substack{r,v,p\in \Lambda_+^*\\ r+v, p+r+v\neq 0}} \eta_p\gamma_{r+v} \big(\alpha_{r,v} + \alpha_{-r-v,v}+\alpha_{r,-r-v}\big) b^*_{p+r+v}b^*_{-p}b^*_{-r} b^*_{-v} +\hc + \wt{\cE}_{S_1} 
\end{split}\]
where $\pm \wt\cE_{S_1} \leq C N^{-1} (\cK+1)(\cN_++1)^2$. Let us bound, using Eq. \eqref{eq:boundalpha}, the  contribution proportional to $\alpha_{r,v}$; the other terms can be controlled analogously.
\begin{equation}
\label{eq:S11}
\begin{split}
\Big|\frac{1}{N^{3/2}}&\sum_{\substack{r,v,p}} \eta_p\gamma_{r+v} \alpha_{r,v} \,\langle \xi, b^*_{p+r+v}b^*_{-p}b^*_{-r} b^*_{-v}\xi \rangle\Big|\\&\leq \frac{C}{N^{3/2}}\bigg(\sum_{\substack{r,v,p}}|\eta_p|^2\frac{|\alpha_{r,v}|^2}{r^2}\bigg)^{1/2} \bigg(\sum_{\substack{r,v,p}} r^2 \|b_{p}b_{r}b_{v}(\cN_++1)^{-1/2}\xi\|^2\bigg)^{1/2}\|(\cN_++1)\xi\|\\&\leq \frac{C}{N} \langle \xi,(\cK+1)(\cN_++1) \xi\rangle \,. 
\end{split}
\end{equation}
We conclude that $\pm S_1 \leq C N^{-1} (\cK +1)(\cN_+ + 1)^2$.

Next, we consider the term $S_2$. Using (\ref{eq:CCR_b}), we find 
\[
\begin{split}
    S_2 =&\, \frac{1}{N^{3/2}} \sum_{\substack{ r,v \in \L^*_+\\ r+v\ne0 }} \big(\eta_{r+v} +\eta_v\big)\gamma_r\,\big(\alpha_{r,v}+\alpha_{-r-v,v}+\alpha_{v,r}\big)b^*_rb^*_{-r} +\hc\\
    &+\frac{1}{N^{3/2}} \sum_{\substack{ r,v \in \L^*_+\\r+v\ne 0 }} \big(\eta_p+\eta_v\big)\g_{v+p}\big(\alpha_{r,v}+\alpha_{v,r}+\alpha_{r,-r-v}\big)b^*_{-p}b_{-r}b_{r+v}b_{-v-p} +\hc\\
    & + \wt{\cE}_{S_2}\,,
\end{split}
\]
where $\pm \wt\cE_{S_2} \leq CN^{-1} (\cK + 1) (\cN_+ + 1)^2$.  The quartic terms  can be controlled in a similar way to $S_1$, using the second bound in \eqref{eq:boundalpha}. Also the first two quadratic contributions, proportional to $\alpha_{r,v}$ and $\alpha_{-r-v,v}$, can be estimated using the second bound in \eqref{eq:boundalpha}. As for the last quadratic term, the one proportional to $\alpha_{v,r}$, it can be handled using the third bound of (\ref{eq:boundalpha}). In fact, 
\begin{equation}
\label{eq:S2}
\begin{split}
\Big|\frac{1}{N^{3/2}}& \sum_{r,v } \big(\eta_{r+v} +\eta_v\big)\gamma_v\,\alpha_{v,r}\langle \xi,b^*_rb^*_{-r}\xi\rangle \Big|\\ &\leq \frac{C}{N^2}\sum_{r,v }\frac{|\eta_{r+v}|+|\eta_v|}{|r|}\|b_r\xi\|\,\|(\cN_++1)^{1/2}\xi\|\leq \frac{C}{N}\langle \xi, (\cK+1)\xi \rangle.
\end{split}
\end{equation}
We conclude that $\pm S_2 \leq C N^{-1} (\cK +1)(\cN_+ + 1)^2$. 

Finally, we consider $S_3$. We decompose $S_3 = S_3^{(0)} + S_3^{(2)} +S_3^{(4)} + \wt{\cE}_{S_3}$, where $\pm \wt{\cE}_{S_3} \leq C N^{-1} (\cK + 1)( \cN_+ + 1)^2$, 
\[
\begin{split}
    S_3^{(0)} = &\, \frac{2}{N^{3/2}}\sum_{\substack{ r,v \in \L^*_+\\ r+v\ne0 }} \alpha_{r,v}\,\big( \nu_{r,v}+\nu_{r,-r-v}+\nu_{v,r}+\nu_{v,-r-v}+\nu_{-r-v,r}+\nu_{-r-v,v}\big) \, , 
\end{split}
\]
\[
\begin{split}
S_3^{(2)} = &\frac{2}{N^{3/2}} \sum_{\substack{ r,v \in \L^*_+\\r+v\ne0}}\big(\alpha_{r,v}+2\alpha_{v,r}\big) \big( \nu_{r,v}+\nu_{r,-r-v}+\nu_{v,r}+\nu_{v,-r-v}+\nu_{-r-v,r}+\nu_{-r-v,v}\big) b^*_r b_r
\end{split}
\]
and
\begin{equation*}
\begin{split}
    S_3^{(4)} =&\, \frac{1}{N^{3/2}} \sum_{\substack{r,v,p\in \Lambda_+^*\\r+v,r+p\neq 0}} \big( \alpha_{r,v}+\alpha_{v,r}+\alpha_{-r-v,v} \big)\big(\nu_{p,-r-p}+\nu_{p,r}+\nu_{r,p} \big)b^*_{p+r}b^*_{-p}b_{r+v}b_{-v}+\hc
\end{split}
\end{equation*}
Using (\ref{eq:boundalpha}) (in particular, the bound for the first two terms on the l.h.s. of (\ref{eq:boundalpha})), as well as $|\nu_{r,p}|\leq C|\eta_r|\,|p|^{-2}$, $S_3^{(4)}$ can be estimated by $CN^{-1}(\cK+1)(\cN_++1)$.

Also the quadratic terms are negligible. Indeed,  the terms proportional to $\alpha_{r,v}$ are easily bounded by $N^{-3/2}(\cN_++1)$, using the first estimate in  (\ref{eq:boundalpha}). The quadratic terms proportional to $\alpha_{v,r}$, on the other hand, can be controlled using the third bound in (\ref{eq:boundalpha}) and $|\nu_{r,p}|\leq C\,|\eta_r||\sigma_p| $, similarly to (\ref{eq:S2}). Finally, we consider the constant term in $S^{(0)}_3$. From the definition \eqref{eq:def_wtsigma} of $\nu_{r,v}$, we arrive at 
\begin{equation}\label{eq:S3-dec}
\begin{split}
S_3^{(0)}=&\, \frac{2}{N^{3/2}} \sum_{\substack{r,v\in \Lambda_+^*\\r+v\neq 0}}\alpha_{r,v}\sigma_v\big(\eta_r+\eta_{r+v}\big)\\
&+\frac{2}{N^{3/2}} \sum_{\substack{r,v\in \Lambda_+^*\\r+v\neq 0}}\alpha_{r,v}\, \Big[\frac{\sigma_v(r\cdot v)(\eta_{r+v} -\eta_r)}{r^2+v^2+(r+v)^2} + \frac{v^2\eta_v\sigma_{r+v}}{r^2+v^2 +(r+v)^2}\Big]\\
&+\frac{2}{N^{3/2}} \sum_{\substack{r,v\in \Lambda_+^*\\r+v\neq 0}}\alpha_{r,v}\frac{v^2}{r^2+v^2+(r+v)^2}[\eta_v(\sigma_r-\eta_r) -\eta_r(\sigma_v-\eta_v)] \\
&+\frac{2}{N^{3/2}} \sum_{\substack{r,v\in \Lambda_+^*\\r+v\neq 0}}\alpha_{r,v}\big(\nu_{r,-r-v}+\nu_{-r-v,r}\big)
 \end{split}
\end{equation} 
With (\ref{eq:boundalpha}), $|\nu_{r,p}|\leq C|\eta_r||\sigma_p|$ and $|\sigma_r-\eta_r|\leq C |\tau_r|$, it is easy to show that the terms in the third and fourth lines are small, bounded by $ CN^{-1}$. The terms on the second line are also negligible, since 
\[\begin{split}
& \Big|\frac{2}{N^{3/2}} \sum_{r,v}\alpha_{r,v} \frac{\sigma_v(\eta_{r+v}-\eta_r)(r\cdot v) + v^2\eta_v\sigma_{r+v}}{r^2+v^2+(r+v)^2}\Big|\\
& \qquad \qquad \qquad\, \leq \frac{C}{N^2}\sum_{r,v}\Big(\frac{|\eta_{r+v}||\sigma_v|}{|r|}+ \frac{|\eta_{r}||r|}{(r+v)^2v^2} + \frac{|\eta_{v}||v|}{(r+v)^2r^2}\Big)\leq\frac{C}{N}.
\end{split}\]
Recalling the definition (\ref{eq:defalpha}) of $\alpha_{r,v}$ to rewrite the first term on the r.h.s. of (\ref{eq:S3-dec}), we obtain (\ref{eq:V,A_2,A}).
\end{proof}

Combining Lemma \ref{lemma:K,V,A} with Lemma \ref{lemma:second_commut_H_A}, we can now compute the action of $A$ on the Hamiltonian $\cH_N$. 
\begin{prop} \label{prop:cubic_on_H}
    Let $A$ be defined as in \eqref{eq:def_A}, and $\mathcal{H}_N=\mathcal{K}+\mathcal{V}_N$. Let $\mathcal{C}_{\mathcal{G}_N}$ be defined in \eqref{eq:cC_G_N}, $[\mathcal{K},A]_2$ and $[\mathcal{V}_N,A]_2$ be defined in Lemma \ref{lemma:K,V,A}. Then, under the same assumptions as in Theorem \ref{thm:cubic}, we have 
    \begin{equation} \label{eq:cubic_on_H}
    \begin{split}
        e^{-A}\mathcal{H}_N e^{A} =\;&\mathcal{H}_N- \mathcal{C}_{\mathcal{G}_N} \\ &-\frac{1}{N}\sum_{p,q\in\Lambda_+^* : p+q \ne 0}\Big(\widehat{V} (p/N)+\widehat{V} ((p+q)/N) \Big)\eta_p \\ &\hspace{3cm} \times \Big[\sigma_q^2+\big(\gamma_q^2+\sigma_{q}^2\big)b^*_qb_q+\gamma_q\sigma_q\big(b_qb_{-q}+b^*_q b^*_{-q}\big)\Big]\\
            &+ \frac{1}{N}\sum_{\substack{p,q\in\Lambda_+^*\\p+q\ne0}}\bigg[\big(\widehat{V}(\cdot/N)\ast\widehat f_{N,\ell}\big)_p +\big(\widehat{V}(\cdot/N)\ast\widehat f_{N,\ell}\big)_{p+q}\bigg]\\
            &\hspace{4cm}\times\eta_p\,\sigma_q\frac{2\eta_{q+p}(q+p)^2-2\sigma_q(p\cdot q)}{p^2+q^2+(p+q)^2}
            \\ &+ \cE_{\cH_N} 
                   \end{split}
    \end{equation}
    where, for every $\eps > 0$ (s.t. $\eps > C (\log N) / N$) 
    \[  \pm \cE_{\cH_N} \leq \eps \cK + \frac{C}{N} \left[ (\log N)^{1/2} + \eps^{-1} \right] (\cH_N + 1 ) (\cN_+ + 1)^4 .\]
\end{prop}

\begin{proof}
From Lemma \ref{lemma:K,V,A}, we find, using the  scattering equation (\ref{scattering}) to combine $[\cK,A]_1$ and $[\cV_N, A]_1$, 
\[ \big[ \cH_N , A \big] =  - \cC_{\cG_N} + \big\{ [\cK,A]_2 + [\cV_N , A]_2 + \text{h.c.} \big\}  + \cE_{[\cH_N,A]}   \]
where 
\[ |\langle \xi, \cE_{[\cH_N , A]} \xi \rangle | \leq \frac{C}{\sqrt{N}} \| \cN_+^{1/2} \xi \| \| (\cN_+ + 1) \xi \| + \frac{C}{N} \langle \xi, (\cH_N+ 1) (\cN_+ + 1)^3 \xi \rangle \]
Here we used the estimate $\| \widehat{\chi}_\ell * \widehat{f}_{N,\ell} \|_2 = \| \chi_\ell f_{N,\ell} \|_2 \leq \| \chi_\ell \|_2 \leq C$ (for fixed $\ell > 0$, independent of $N$), to absorb into the error $\cE_{[\cH_N, A]}$ the contribution arising from the r.h.s. of the scattering equation (\ref{scattering}), when combining $[\cK,A]_1$ and $[\cV_N, A]_1$.
Thus, we obtain 
\begin{equation} \label{eq:eHe} \begin{split} 
e^{-A} \cH_N e^A = \; &\cH_N + \int_0^1 ds \, e^{-sA} [\cH_N, A] e^{sA} \\ = \; &\cH_N - \int_0^1 e^{-sA} \cC_{\cG_N} e^{sA} \, ds+ \int_0^1  e^{-sA} \left\{ [\cK,A]_2 + [\cV_N , A]_2 + \text{h.c.} \right\}  e^{sA} \, ds \\ &+ \int_0^1 e^{-sA} \cE_{[\cH_N , A]} e^{sA} \, ds 
\end{split} \end{equation} 
From Prop. \ref{prop:a_priori_A}, we find  
\[ | \langle \xi, e^{-sA} \cE_{[\cH_N,A]} e^{sA} \xi \rangle | \leq \frac{1}{\sqrt{N}} \| \cN_+^{1/2} \xi \| \| (\cN_+ + 1) \xi \|   +  \frac{C}{N} \langle \xi, (\cH_N +1) (\cN_+ + 1)^4 \xi \rangle. \]

With Prop. \ref{prop:A_on_cubic} we obtain, after integration over $s$, 
\begin{equation}\label{eq:intC} \begin{split}  \int_0^1 e^{-sA} & \cC_{\cG_N} e^{sA} ds \\\; =&\,\cC_{\cG_N} + \frac{1}{N} \sum_{p,q \in \Lambda^*_+ : p+q \not=0} \Big(\widehat{V} (p/N)+\widehat{V} ((p+q)/N) \Big) \eta_p \\
    &\qquad\qquad\times\Big[\sigma_q^2+\big(\gamma_q^2+\sigma_{q}^2\big)b^*_qb_q+\gamma_q\sigma_q\big(b_qb_{-q}+b^*_q b^*_{-q}\big)\Big]\\
    &\;+\frac{1}{N}\sum_{\substack{p,q\in\Lambda_+^*\\p+q\ne0}}\bigg[\widehat{V} (p/N)+\widehat{V} ((p+q)/N)\bigg]\eta_p\sigma_q\frac{2\eta_{q+p}(q+p)^2-2\sigma_q(p\cdot q)}{p^2+q^2+(p+q)^2} \\
    &\;+\mathcal{E}_{\mathcal{C}_{\cG_N}}, \end{split}  \end{equation} 
where, for every $\eps > 0$, \[ \pm \cE_{\cC_{\cG_N}}  \leq \eps \cK + \frac{C}{\eps N} (\cH_N +1) (\cN_+ + 1)^4. \]

Finally, we compute 
\[ \begin{split} e^{-sA} &\left\{ [\cK,A]_2 + [\cV_N , A]_2 + \text{h.c.} \right\}  e^{sA} \\ &= [\cK,A]_2 + [\cV_N , A]_2 + \text{h.c.}  + \int_0^s  e^{-t A} \left\{ \big[ [\cK,A]_2 ,A \big] + \big[ [\cV_N , A]_2 , A \big] + \text{h.c.} \right\}  e^{tA} \, dt. \end{split} \]
We estimate the terms in the integral over $t$ using Lemma \ref{lemma:second_commut_H_A}. 
We conclude that 	
\begin{equation}\label{eq:cKA2+cVA2} 
\begin{split} \int_0^1 &e^{-sA} \left\{ [\cK,A]_2 + [\cV_N , A]_2 + \text{h.c.} \right\}  e^{sA} \, ds \\ =\; &
[\cK,A]_2 + [\cV_N , A]_2 + \text{h.c.} \\
&+ \frac{1}{N^2}\sum_{\substack{p,q\in\Lambda_+^*\\p+q\ne0}}\bigg[\big(\widehat{V}(\cdot/N)\ast\eta\big)_p +\big(\widehat{V}(\cdot/N)\ast\eta\big)_{p+q}\bigg]\eta_p\,\sigma_q\frac{2\eta_{q+p}(q+p)^2-2\sigma_q(p\cdot q)}{p^2+q^2+(p+q)^2} \\ & + \cE_{[[\cH_N,A],A]} \end{split} \end{equation} 
where 
\[ \pm \cE_{[[\cH_N,A],A]} \leq \frac{C (\log N)^{1/2}}{N}(\mathcal{K}+1)(\mathcal{N}_++1)^2. \] Furthermore,  from \eqref{eq:bound[VN,A]2} we have  
\[ \begin{split} \big| \langle \xi, &\big[ \cV_N, A \big]_2 \xi \rangle \big| 
\leq \frac{C}{\sqrt{N}} \| \cK^{1/2} \xi \| \| (\cN_+ + 1) \xi \|. \end{split} \] 
As for $[\cK, A]_2$, from \eqref{eq:bound[K,A]2} we also have 
\[ \big| \langle \xi, [\cK,A]_2 \xi \rangle \big| \leq \frac{C}{\sqrt{N}} \| \cK^{1/2} \cN_+^{1/2} \xi \| \| \cK^{1/2} \xi \|. \]
  
Combining  (\ref{eq:cKA2+cVA2}) with (\ref{eq:intC}), and (\ref{eq:eHe}), and using the definition of $\eta_p = -N \delta_{p,0} + N\widehat f_{N,\ell}(p)$ we arrive at (\ref{eq:cubic_on_H}). 
\end{proof}

\subsection{Proof of Theorem \ref{thm:cubic}} 

 We start from (\ref{eq:expansion_G_N}) in Theorem \ref{thm:quadratic}. We use Prop. \ref{prop:a_priori_A} to estimate 
\[ \pm e^{-A} \cE_{\cG_N} e^A \leq  \eps \cN_+ + \frac{C}{\eps N} (\cH_N +1)( \cN_+ + 1)^3 \]

Combining the expansion (\ref{eq:expansion_G_N}) for $\cG_N$ with Prop. \ref{prop:a_priori_A} and with the results of Prop. \ref{prop:quad-A}, Prop. \ref{prop:T4-bd}, Prop. \ref{prop:A_on_cubic} and Prop. \ref{prop:cubic_on_H}, we conclude that  
\begin{equation} \label{eq:cJN1}  \begin{split} \cJ_N = \; &C_{\cG_N} + \frac{1}{N} \sum_{p,q} (\widehat{V} (p/N)  + \widehat{V} ((q+p)/N)) \eta_p \s_q^2  \\
            &+ \frac{1}{N}\sum_{\substack{p,q\in\Lambda_+^*\\p+q\ne0}}\bigg[\big(\widehat{V}(\cdot/N)\ast\widehat f_{N,\ell}\big)_p +\big(\widehat{V}(\cdot/N)\ast\widehat f_{N,\ell}\big)_{p+q}\bigg]\\
            &\hspace{4cm}\times\eta_p\,\sigma_q\frac{2\eta_{q+p}(q+p)^2-2\sigma_q(p\cdot q)}{p^2+q^2+(p+q)^2}\\
&+ \sum_{p \in \Lambda^*_+} \sqrt{|p|^4+2 (\widehat{V} (\cdot /N) * \widehat{f}_{N,\ell})_p \, p^2} \, a_p^* a_p     + \cV_N  + \wt{\cE}_{\cJ_N} \end{split} \end{equation} 
where, for every $\eps > 0$ (s.t. $\eps > C (\log N) / N$), we have 
\[ \pm \wt{\cE}_{\cJ_N} \leq \eps \cK + \frac{C}{N} \left[ (\log N)^{1/2} + \eps^{-1} \right] (\cH_N+1)( \cN_+ + 1)^4\,.  \]
Notice here that the quadratic terms appearing on the second line of (\ref{eq:expansion_G_N}) cancel exactly the quadratic contribution arising from conjugation of the cubic term $\cC_{\cG_N}$ (as determined in Prop. \ref{prop:A_on_cubic}) and of the Hamiltonian $\cH_N$ (as determined in Prop. \ref{eq:cubic_on_H}) \footnote{Roughly speaking, the goal of conjugation with $\exp (A)$ is to eliminate the cubic term $\cC_{\cG_N}$ appearing in the expression (\ref{eq:expansion_G_N}) for $\cG_N$. To reach this goal, we require that, in an appropriate sense, $[\cH_N, A] \simeq - \cC_{\cG_N}$. But then, $[ [ \cH_N, A] , A] \simeq - [ \cC_{\cG_N} , A]$, which implies that $e^{-A} \cH_N e^A \simeq \cH_N - \cC_{\cG_N} - [ \cC_{\cG_N}, A] /2$ and therefore that $e^{-A} (\cH_N + \cC_{\cG_N} ) e^A \simeq \cH_N + [ \cC_{\cG_N}, A] /2$. Since $\cC_{\cG_N}$ and $A$ are both cubic in (modified) creation and annihilation operators, $[\cC_{\cG_N} , A] / 2$ is, to leading order, quartic. Arranging it in normal order, we generate quadratic terms (cancelling the quadratic operator on the second line of (\ref{eq:expansion_G_N})) and constant terms (correcting the ground state energy).} 

To complete the proof of Theorem \ref{thm:cubic}, we observe, first of all, that the terms  on the second line of (\ref{eq:cJN1}) produce, up to smaller errors, the terms on the second line of (\ref{eq:def_C_J}). To this end, we consider the difference 
\[\begin{split}
&\frac{1}{N}\sum_{\substack{p,q\in\Lambda_+^*\\p+q\ne0}}\bigg[\big(\widehat{V}(\cdot/N)\ast\widehat f_{N,\ell}\big)_p +\big(\widehat{V}(\cdot/N)\ast\widehat f_{N,\ell}\big)_{p+q}\bigg]\\
& \hspace{4cm} \times \bigg[\frac{2\eta_{q+p}\eta_p\,(\sigma_q-\eta_q)(q+p)^2}{p^2+q^2+(p+q)^2}-\frac{2\eta_p\,(\sigma_q^2-\eta_q^2)(p\cdot q)}{p^2+q^2+(p+q)^2}\bigg]\,.
\end{split}\]
where we compare $\sigma_q$ and $\sigma_q^2$ with $\eta_q$ and, respectively, $\eta_q^2$. Using $|\sigma_q-\eta_q| \leq C|q|^{-4}$, the first contribution can be shown to be of order $N^{-1}$. Also the second contribution, containing the factor $p \cdot q$, is of order $N^{-1}$; this can be proven proceeding similarly as we did after  \eqref{eq:Pi_1quadratic}, switching $p\rightarrow -p$ and using the Lipschitz continuity of $\widehat{V}$ for $|p| \leq N$.

\medskip

Finally, we claim that the constant terms on the first line on the r.h.s. of (\ref{eq:cJN1}) match, up to negligible errors, the terms on the first line on the r.h.s. of (\ref{eq:def_C_J}). To this end, we set (recall (\ref{eq:C_G_N}) for the definition of  $C_{\cG_N}$) 
\[ \begin{split} C_{\cO (1)} = \; &C_{\cG_N} + \frac{1}{N} \sum_{p,q \in \Lambda^*_+}  (\widehat{V} (p/N)  + \widehat{V} ((q+p)/N)) \eta_p \s_q^2 \\ = \; &\frac{(N-1)}{2} \, \widehat V(0)  + \sum_{p \in \L^*_+} \Big[ p^2\s_p^2  + \widehat V (p/N) \s_p \g_p + \big( \widehat{V} (\cdot/N) *\hat{f}_{N} \big)_p  \s_p^2 \Big] \\
 & + \frac{1}{2N} \sum_{p \in \L^*_+} \widehat V ((p-q)/ N) \s_q\g_q\s_p\g_p + \frac1N \sum_{p \in \L^*_+}\Big( p^2\eta_p^2 + \frac1{2N}\left(\widehat V\left(\frac{\cdot}{N}\right) *\eta\right)_p \eta_p\Big) \end{split}  \]
with $\gamma_p=\cosh(\mu_p)$, $\sigma_p=\sinh(\mu_p)$, and $\mu_p=\eta_p+\tau_p$. We now claim that 
\begin{equation}\label{eq:CO1} C_{\cO (1)} = 4\pi\mathfrak{a}(N-1)+e_\Lambda \mathfrak{a}^2-\frac{1}{2} \sum_{p\in\Lambda_+^*} \Big[p^2+8\pi\mathfrak{a}-\sqrt{|p|^4+16 \pi \mathfrak{a}p^2} - \frac{(8\pi\mathfrak{a})^2}{2p^2}\Big]+\mathcal{O}(N^{-1}) \end{equation} 
To show (\ref{eq:CO1}), we consider the following quantities, defined in terms of the kernel $\eta$:
\begin{equation*}
\begin{split}
    \widetilde{{C}}_{\mathcal{J}_N}=\;&\frac{N-1}{2}\widehat{V}(0)\\
    &+\sum_{p\in\Lambda_+^*}\bigg[p^2\sinh(\eta_p)^2+\widehat{V} (p/N) \cosh(\eta_p)\sinh(\eta_p)+\big(\widehat{V} (\cdot / N)*\widehat{f}_{N,\ell} \big)_p\sinh(\eta_p)^2\bigg]\\
    &+\frac{1}{2N}\sum_{p,q\in\Lambda_+^*} \widehat{V} ((p-q)/N) \cosh(\eta_q)\sinh(\eta_q) \cosh(\eta_p)\sinh(\eta_p)\\
    &+\frac{1}{N}\sum_{p\in\Lambda_+^*}\bigg[p^2\eta_p^2+\frac{1}{2N}\big(\widehat{V} (\cdot /N)*\eta\big)_p\eta_p\bigg] 
\end{split}
\end{equation*}
and 
\begin{equation*}
\begin{split}
    \widetilde{F}_p=\;&p^2\big(\cosh(\eta_p)^2+\sinh(\eta_p)^2\big)+\big(\widehat{V} (\cdot / N) *\widehat{f}_{N,\ell} \big)_p \big( \cosh(\eta_p)+\sinh(\eta_p)\big)^2\\
    \widetilde{G}_p=\;&2p^2\cosh(\eta_p)\sinh(\eta_p)+\big(\widehat{V} (\cdot / N) *\widehat{f}_{N,\ell} \big)_p \big( \cosh(\eta_p)+\sinh(\eta_p)\big)^2.
\end{split}
\end{equation*}
In \cite[Lemma 5.4 (i)]{BBCS} it was proven that
\begin{equation} \label{eq:acta_E_bog}
\begin{split}
    \widetilde{C}_{\mathcal{J}_N} &+ \frac{1}{2} \sum_{p\in\Lambda_+^*}\big(-\widetilde{F}_p+\sqrt{\widetilde{F}_p^2-\widetilde{G}_p^2} \, \big)\\
    =\;&4\pi\mathfrak{a}(N-1)+e_\Lambda \mathfrak{a}^2-\frac{1}{2} \sum_{p\in\Lambda_+^*} \Big[p^2+8\pi\mathfrak{a}-\sqrt{|p|^4+16 \pi \mathfrak{a}p^2} - \frac{(8\pi\mathfrak{a})^2}{2p^2}\Big]+\mathcal{O}(N^{-1}) .
\end{split}
\end{equation}
In fact, the statement in \cite[Lemma 5.4 (i)]{BBCS} contains a larger remainder, of the order $\mathcal{O}( \log N / N)$, emerging when controlling the difference
    \begin{equation}\label{eq:bog-sum} 
    \begin{split}
        \sum_{p\in\Lambda_+^*}\bigg[\;& \big(\widehat{V} (\cdot/N)*\widehat{f}_{N,\ell} \big)_p + p^2-\sqrt{p^4+2p^2\big(\widehat{V}(\cdot /N)*\widehat{f}_{N,\ell} \big)_p}+\frac{1}{2p^2}\big(\widehat{V} (\cdot /N)*\widehat{f}_{N,\ell}\big)_p^2\\
        &-\big(\widehat{V}(\cdot /N)*\widehat{f}_{N,\ell} \big)_0-p^2+\sqrt{p^4+2p^2\big(\widehat{V} (\cdot / N)*\widehat{f}_{N,\ell} \big)_0}+\frac{1}{2p^2}\Big(\widehat{V} (\cdot / N)*\widehat{f}_{N,\ell} \big)_0^2\bigg] .     
    \end{split}
    \end{equation}
    We can however show that this remainder is indeed smaller, of the order $\cO (N^{-1})$, consistently with (\ref{eq:acta_E_bog}). Taylor expansion of the square root easily implies that the whole square bracket is bounded, in absolute value, by $C |p|^{-4}$. Hence, the sum over momenta $|p|>N$ can be estimated by $CN^{-1}$. To treat $|p|\le N$, we notice that the function 
    \begin{equation*}
        g_p(x)=x-p^2\sqrt{1+2x/p^{2}}-x^2 /(2p^2)
    \end{equation*}
    has derivative 
    \begin{equation*}
    \begin{split}
        g_p'(x)=\;&1-(1+2x/p^{2})^{{-1/2}}-x/p^2\\
    \end{split}
    \end{equation*}
satisfying $|g'_p (x)| \leq C x^2 / |p|^4$ (as it follows again by expansion of the square root). This implies that $|g_p (x) - g_p (y)| \leq C |x-y| / |p|^4$, for all $x,y$ varying in a bounded interval. Estimating 
\begin{equation} \label{eq:Vfp-Vf0} \Big| (\widehat{V} (\cdot / N) *\widehat{f}_{N,\ell})_p - (\widehat{V} (\cdot /N) *\widehat{f}_{N,\ell})_0 \Big| \leq \Big| N^3 \int_\Lambda dx V (Nx) f_{N,\ell} (x) (e^{-ip\cdot x} - 1) \Big| \leq C p^2 / N^2  \end{equation} 
since $\int_\Lambda x \, V(x) f_{N,\ell} (x) dx = 0$ by symmetry, we conclude that the bracket in (\ref{eq:bog-sum}) is bounded, in absolute value, by $C / p^2 N^2$. Thus, the sum over all $|p| < N$ can also be estimated by $C/N$. 

From (\ref{eq:acta_E_bog}), it follows that, in order to show (\ref{eq:CO1}), it is enough to prove that 
\begin{equation} \label{eq:reduction_F,G_p}
    C_{\mathcal{O}(1)} = \widetilde{C}_{\mathcal{J}_N}+\frac{1}{2} \sum_{p\in\Lambda_+^*}\Big[ -\widetilde{F}_p+ \sqrt{\widetilde{F}_p^2-\widetilde{G}_p^2} \Big] +\mathcal{O}(N^{-1})
\end{equation}
To this end, we consider 
\begin{equation}\label{eq:reduction_F,G_p_2}
\begin{split}
    C_{\mathcal{O}(1)} & -\widetilde{C}_{\mathcal{J}_N}\\
    =\;&\sum_{p\in\Lambda_+^*} \bigg[ \big( p^2+\big(\widehat{V} (\cdot / N) * \widehat{f}_{N,\ell} \big)_p\big)\big(\sigma_p^2-\sinh(\eta_p)^2\big)\\
        &\qquad\quad+\widehat{V} (p/ N) \big(\gamma_p\sigma_p-\cosh(\eta_p)\sinh(\eta_p)\big)\\
        &\qquad\quad+\frac{1}{2N}\sum_{q\in\Lambda_+^*} \widehat{V} ((p-q)/N) \big(\gamma_q\sigma_q \gamma_p\sigma_p-\cosh(\eta_p)\sinh(\eta_p)\cosh(\eta_q)\sinh(\eta_q)\big)\bigg]\,.
    \end{split}
\end{equation}
Let us separately manipulate the three main contributions to the sum in the right hand side. Through elementary identities we find
\begin{equation*}
\begin{split}
    \sigma_p^2-\sinh(\eta_p)^2=\;&\frac{1}{2}\sinh(2\eta_p)\sinh(2\tau_p)+\sinh^2(\tau_p)\cosh(2\eta_p)\\
    =\;&\frac{1}{2}\sinh(2\eta_p)\sinh(2\tau_p)+\frac{1}{2}\cosh(2\eta_p)(\cosh(2\tau_p)-1),
\end{split}
\end{equation*}
and therefore the first contribution to the right hand side of \eqref{eq:reduction_F,G_p_2} is
\begin{equation} \label{eq:first_term_F,G}
\begin{split}
    \sum_{p\in\Lambda_+^*} \big( &p^2+\big(\widehat{V} (\cdot / N) * \widehat{f}_{N,\ell} \big)_p\big)\big(\sigma_p^2-\sinh(\eta_p)^2\big)\\
    =\;&\frac{1}{2}\sum_{p\in\Lambda_+^*}\big(p^2+\big(\widehat{V} (\cdot / N) *\widehat{f}_{N,\ell} \big)_p\big)\big[\sinh(2\eta_p)\sinh(2\tau_p)+\cosh(2\eta_p)(\cosh(2\tau_p)-1)\big]. 
\end{split}
\end{equation}
In a similar way we rewrite the second contribution as
\begin{equation}\label{eq:second_term_F,G}
    \begin{split}
        \sum_{p\in\Lambda_+^*}\widehat{V}& (p/ N) \big(\gamma_p\sigma_p-\cosh(\eta_p)\sinh(\eta_p)\big) \\
        = \;&\frac{1}{2}\sum_{p\in\Lambda_+^*}\widehat{V} (p/N) \big[ \sinh(2\tau_p)\cosh(2\eta_p)+\sinh(2\eta_p)(\cosh(2\tau_p)-1)\big].
    \end{split}
\end{equation}
In order to rearrange the third contribution to the right hand side of \eqref{eq:reduction_F,G_p_2}, notice that
\begin{equation*}
    \begin{split}
        4 \gamma_q &\sigma_q\gamma_p\sigma_p \\ =\;&\sinh(2\eta_q)\cosh(2\tau_q) \sinh(2\eta_p) \cosh(2\tau_p)+ \sinh(2\tau_q)\cosh(2\eta_q)\cosh(2\eta_p)\sinh(2\tau_p) \\
        &+\sinh(2\eta_q)\cosh(2\tau_q)\cosh(2\eta_p)\sinh(2\tau_p) + \cosh(2\eta_q)\sinh (2\tau_q)\sinh (2\eta_p)\cosh (2\tau_p)  
    \end{split}
\end{equation*}
Since $|\tau_p|\le C|p|^{-4}$, the second term on the r.h.s. (containing $\sinh(2\tau_q) \sinh(2\tau_p)$) yields a contribution bounded by $CN^{-1}$ when inserted in the sum over $p,q$ in \eqref{eq:reduction_F,G_p_2}. Using the estimate $|\sinh(2\eta_p) \cosh (2\tau_p) -2\eta_p |\le C|p|^{-6}$ to handle the last two terms (they give the same contribution), decomposing $\cosh (2\tau_p) = 1 + (\cosh (2\tau_p) - 1)$ and similarly for $\cosh (2\tau_q)$ and subtracting the terms $\cosh (\eta_p) \sinh (\eta_p) \cosh (\eta_q) \sinh (\eta_q) = (1/4) \sinh (2\eta_p) \sinh (2\eta_q)$, we find (adding also a negligible term proportional to $\eta_0$) 
\begin{equation} \label{eq:third_term_F,G}
\begin{split}
    \frac{1}{2N} &\sum_{p,q\in\Lambda_+^*} \widehat{V} ((p-q)/ N) \big(\gamma_q\sigma_q \gamma_p\sigma_p-\cosh(\eta_p)\sinh(\eta_p)\cosh(\eta_q)\sinh(\eta_q)\big)\\
    =\;&\frac{1}{2N}\sum_{p\in\Lambda_+^*} \big(\widehat{V} (\cdot / N) *\eta\big)_p \big[ \sinh(2\tau_p)\cosh(2\eta_p)+\sinh(2\eta_p)(\cosh(2\tau_p)-1)\big]+\mathcal{O}(N^{-1}).
\end{split}
\end{equation}
Plugging \eqref{eq:first_term_F,G}, \eqref{eq:second_term_F,G}, and \eqref{eq:third_term_F,G} into the right hand side of \eqref{eq:reduction_F,G_p_2}, the terms with a minus sign recombine into 
\begin{equation*}
\begin{split}
    -\frac{1}{2}\big( p^2+&\big(\widehat{V}(\cdot / N) *\widehat{f}_{N,\ell} \big)_p\big)\cosh(2\eta_p)-\frac{1}{2}\Big(\widehat{V} (\cdot / N)*\widehat{f}_{N,\ell} \big)_p \sinh(2\eta_p) 
    =-\frac{\widetilde{F}_p}{2}.
\end{split}
\end{equation*}
As already discussed in \eqref{eq:F-compu}, the other terms produce 
\begin{equation*}
    \begin{split}
        \frac{1}{2}\big(p^2+&\big(\widehat{V}(\cdot/ N)*\widehat{f}_{N,\ell} \big)_p\big)\big[\sinh(2\eta_p)\sinh(2\tau_p)+\cosh(2\eta_p)\cosh(2\tau_p)\big]\\
        &+\frac{1}{2}\big(\widehat{V} (\cdot / N) *\widehat{f}_{N,\ell} \big)_p\big[ \sinh(2\tau_p)\cosh(2\eta_p)+\sinh(2\eta_p)\cosh(2\tau_p)\big]\\
        =\;&\frac{1}{2}\sqrt{p^4+2p^2\Big(\widehat{V} (\cdot / N)*\widehat{f}_{N,\ell} \big)_p}=\frac{1}{2}\sqrt{\widetilde{F}_p-\widetilde{G}_p}.
    \end{split}
\end{equation*}
This concludes the proof of \eqref{eq:reduction_F,G_p}.

To conclude the proof of Theorem \ref{thm:cubic}, we observe now that, combining Lemma \ref{lemma:scattering} and (\ref{eq:Vfp-Vf0}), we find 
\[ \Big| (\widehat{V} (\cdot / N) *\widehat{f}_{N,\ell})_p  - 8\pi \frak{a} \Big| \leq C |p| / N \]
This allows us to replace the dispersion of the term on the third line of (\ref{eq:cJN1}) with $(|p|^4 + 16 \pi \frak{a} p^2)^{1/2}$, since  
\[ \pm \sum_{p \in \Lambda^*_+} \Big[ \sqrt{|p|^4 +2 (\widehat{V} (\cdot / N) *\widehat{f}_{N,\ell})_p \, p^2} - \sqrt{|p|^4 +16 \pi\frak{a} p^2} \Big] \, a_p^* a_p \leq \frac{C}{N} \cK .\]
can be absorbed in the error.

\appendix

\section{Bound for $d$-operators}
\label{sec:d-ops} 

In this section we prove Lemma \ref{lemma:d}, without any assumptions on the smallness of the $\ell^2$-norm of the kernel of the Bogoliubov transformation. We start by proving the following Lemma
\begin{lemma}\label{lemma:aprioriboundsa}
Let $B_\mu$ be defined as in \eqref{eq:Bmu}. Then there exists $C>0$ such that the following bounds hold true: 
\begin{align}
\|(\cN_++1)^{n/2} a_p e^{sB_\mu} \xi\|& \leq C \Big[ \|(\cN_++1)^{n/2} a_p  \xi\| +|\mu_p| \|(\cN_++1)^{(n+1)/2}  \xi\| \Big] \label{eq:boundap}\\
   \|(\cN_++1)^{n/2}\check a_x e^{sB_\mu} \xi\| &\leq C \Big[ \|(\cN_++1)^{n/2}\check a_x  \xi\| + \|(\cN_++1)^{n+1/2}  \xi\| \Big] \label{eq:boundax}   \\
  \|(\cN_++1)^{n/2} a_p a_qe^{sB_\mu} \xi\| &\leq C \Big[ \|(\cN_++1)^{n/2}a_pa_q\xi\| + |\mu_q| \|(\cN_++1)^{(n+1)/2}a_p\xi\| \nonumber\\
    &\quad+  |\mu_p|\|(\cN_++1)^{(n+1)/2}a_q\xi\| + \d_{p,-q}|\mu_p| \|(\cN_++1)^{(n+3)/2}\xi\| \Big] \label{eq:boundapaq} \\ 
  \|(\cN_++1)^{n/2}\check a_x \check a_ye^{sB_\mu} \xi\| &\leq C \Big[ \|(\cN_++1)^{n/2}\check a_x\check a_y  \xi\| + |\check \mu(x-y)| \|(\cN_++1)^{n/2}  \xi\| \nonumber\\
 &\quad+\|(\cN_++1)^{n+1/2}\check a_x  \xi\| +\|(\cN_++1)^{n+1/2}\check a_y  \xi\| \Big] \label{eq:boundaxay}
\end{align}
for any $p \in \L^*_+$, any $x,y \in \L$, and $\xi \in \cF_\perp^{\leq N}$. 
\end{lemma}

\begin{proof}
The bounds (\ref{eq:boundap})-(\ref{eq:boundaxay}) can all be shown using Gronwall's Lemma. We focus on (\ref{eq:boundap}); the other estimates can be proven similarly. We consider
    \begin{equation}\label{eq:apeB}
    \begin{split}
        \frac{d}{ds} \|(\cN_++1)^{n/2}a_p e^{sB_\mu}\xi\|^2 = &\, \langle \xi, e^{-sB_\mu} [\cN_+^{n}, B_\mu ]a^*_p a_p e^{sB_\mu}\xi \rangle \\
    &+ \langle \xi, e^{-sB_\mu} \cN_+^{n}[ a^*_p a_p, B_\mu] e^{sB_\mu}\xi \rangle = K_1 + K_2\,.
    \end{split}
    \end{equation} 
 With the commutation relation (\ref{eq:CCR_b}), we find 
    \begin{equation}\label{eq:K2-apeB} 
    \begin{split}
        |K_2| = &\,  |\langle \xi, e^{-sB_\mu}\cN_+^n (\mu_p b^*_p b^*_{-p} +\hc) e^{sB_\mu}\xi\rangle|\\
        \leq &\, C|\mu_p| \|(\cN_++1)^{n/2}a_p e^{sB_\mu}\xi\|\|(\cN_++1)^{(n+1)/2}\xi\| \\ \leq &\, C \| (\cN_+ + 1)^{n/2} a_p e^{sB_\mu} \xi \|^2 + C |\mu_p|^2 \| (\cN_+ + 1)^{(n+1)/2} \xi \|^2  
    \end{split}
    \end{equation} 
    Now we consider $K_1$. We find 
    \begin{equation}\label{eq:commN^nB}
    [\cN_+^n, B_\mu] = \frac{1}{2}\sum_{q \in \L^*_+} \mu_q b^*_q b^*_{-q} [(\cN_++2)^n -\cN_+^n] +\hc  \end{equation}
    Therefore 
    \[
    \begin{split}
        | K_1 |  \leq &\, C \sum_{q \in \Lambda^*_+} |\mu_q |  \big\|  b_{-q} b_q (\cN_+ + 1)^{n/2-1} a_p  e^{sB_\mu} \xi \| \\ &\hspace{1.5cm}  \times \| (\cN_+ + 1)^{-n/2+1} \big[  (\cN_++3)^n -(\cN_++1)^n \big] a_p e^{sB_\mu} \xi \big\| \\ \leq &\; C \|(\cN_+ + 1)^{n/2} a_p e^{sB_\mu}  \xi \|^2 \end{split} \]
        Inserting the last equation and (\ref{eq:K2-apeB}) in (\ref{eq:apeB}) and applying Gronwall's Lemma, we obtain the desired bound. 
        \end{proof} 
        
With Lemma \ref{lemma:aprioriboundsa}, we are now ready to show Lemma \ref{lemma:d}. 
\begin{proof}[Proof of Lemma \ref{lemma:d}]
We start with the first bound in Eq. \eqref{eq:decay_d}. From \eqref{eq:d_detailed_expansion}, we find 
\[
\begin{split}
    \| (\cN_++1)^{n/2} d_p \xi \|  \leq\,&\frac{C}{N}\int_0^1 ds\, \Big[ |\mu_p| \|(\cN_++1)^{(n+3)/2} e^{(1-s)B_\mu} \xi\|  \\
    &\hspace{2cm} + \big\| \sum_{q \in \L^*_+} \mu_q (\cN_++1)^{n/2} b^*_{q}a^*_{-q}a_pe^{(1-s)B_\mu}\xi \big\| \\ 
    &\hspace{2cm} + |\s_p^{(s)}| \Big\| \sum_{q\in \L^*_+}\mu_q (\cN_++1)^{n/2} a^*_{-p}a_{-q}b_qe^{(1-s)B_\mu}\xi\| \Big]\,.
\end{split}
\]
Observing that $|\sigma_p^{(s)}| \leq C |\mu_p|$ and that 
\begin{equation} \label{eq:sum-bds} \begin{split} \big\| \sum_{q \in \L^*_+} \mu_q  &(\cN_++1)^{n/2} b^*_{q}a^*_{-q}a_pe^{(1-s)B_\mu}\xi \big\| 
\leq \| (\cN_+ + 1)^{(n+2)/2} a_p e^{(1-s) B_\mu} \xi \| \\ 
\big\| \sum_{q \in \L^*_+} \mu_q  &(\cN_++1)^{n/2} a_{-p}^* a_{-q} b_q e^{(1-s)B_\mu}\xi \big\| 
\leq \| (\cN_+ + 1)^{(n+3)/2} e^{(1-s) B_\mu} \xi \|
\end{split} \end{equation} 
we conclude, with the help of (\ref{eq:control-eBmu}) and Lemma \ref{lemma:aprioriboundsa}, that 
\[   \| (\cN_++1)^{n/2} d_p \xi \|  \leq \frac{C}{N} \Big[ |\mu_p| \| (\cN_+ + 1)^{(n+3)/2} \xi \| + \| (\cN_+ + 1)^{(n+2)/2} a_p \xi \| \Big] \]
The second estimate in (\ref{eq:decay_d}) follows similarly (it is simpler, because we do not need to extract factors decaying in $p$). Let us now consider the third bound in (\ref{eq:decay_d}). We proceed as above, applying \eqref{eq:d_detailed_expansion}. We find 
\[
\begin{split}
    \|(\cN_++1)^{n/2}a_pd_q\xi\| \leq &\, \frac{C}{N} \int_0^1 ds\, \Big[ |\mu_q| \|(\cN_++1)^{(n+2)/2}a_p b^*_{-q}e^{(1-s)B_\mu}\xi\|  \\
    &+ \big\| \sum_{r \in \L^*_+} \mu_r (\cN_++1)^{n/2} a_p b^*_{r}a^*_{-r} a_qe^{(1-s)B_\mu}\xi \big\| \\ 
    & + |\s_q^{(s)}| |\mu_q| \|(\cN_++1)^{(n+2)/2}a_p b_q e^{(1-s)B_\mu}\xi\| \\
    &+ |\s_q^{(s)}| \big\|\sum_{r\in \L^*_+}\mu_r(\cN_++1)^{n/2}a_pa^*_{-q}a_{-r}b_re^{(1-s)B_\mu}\xi \big\| \Big] 
\end{split}
\]
Commuting $a_p$ to the right of all creation operators, proceeding as in (\ref{eq:sum-bds}) to bound the terms in the second and fourth line and applying Lemma \ref{lemma:aprioriboundsa}, we conclude that  
\[ 
\begin{split}
    \|(\cN_++1)^{n/2}a_pd_q\xi\| \leq &\, \frac{C}{N}\Big[ |\mu_q| \|(\cN_++1)^{(n+3)/2}a_p \xi\| + |\mu_p| 
    \|(\cN_++1)^{(n+3)/2}a_q \xi \| \\
    & + \d_{p,-q}\|(\cN_++1)^{(n+3)/2} \xi\| +\|(\cN_++1)^{(n+2)/2}a_p a_q \xi\| \\
    &+|\mu_p| |\mu_q| \|(\cN_++1)^{(n+4)/2}\xi\|\Big] \,.\end{split} \]
Also the fourth estimate in (\ref{eq:decay_d}) can be shown analogously. To prove (\ref{eq:decay_d_position}), we rewrite  \eqref{eq:d_detailed_expansion} in position space. We obtain 
    \begin{equation*} 
    \begin{split}
        \label{eq:d_detailed_expansion_position}
        \check d_x = &\,  -\frac{1}{N} \int_0^1ds \, e^{-(1-s)B_\mu }\bigg[\int dz \, \check \gamma (x-z) \Big(\cN_+ b^*(\check \eta_z) + \int dr ds \, \check \eta(r-s) \check b^*_s\check a^*_r \check a_z \Big)\\
        & \quad\quad+ \int dz \, \check \sigma(x-z) \Big( b(\check \eta_z) \cN_+  + \int dr ds \, \check \eta(r-s) \check a^*_z\check a_r \check b_s \Big)\bigg]e^{(1-s)B_\mu}\,.
        \end{split}
    \end{equation*}
With the last identity, we can show the first two bounds in (\ref{eq:decay_d_position}) similarly as we did above for the estimates in (\ref{eq:decay_d}). The third bound in (\ref{eq:decay_d_position}) follows directly from the first two (applying the first bound to $\xi' = \check{d}_y \xi$, and then the first and the second bound to the vector $\xi$). 
\end{proof}

\section{A-priori bounds: proof of Proposition \ref{prop:apri}}
\label{app:apri} 

\begin{proof}
It follows from \cite[Proposition 4.1]{BBCS} that for every $k \in \bN$ there exists $C > 0$ such that 
\begin{align}
&\langle \xi'_N, (\cK+1)(\cN_++1)^k \xi'_N \rangle \leq C. \label{eq:boundKN_+'}
\end{align}
for $\xi'_N=e^{-B_{\eta}}U_N\psi_N$. We are going to show that the same bounds hold true, if we replace $\xi'_N$ with $\xi_N = e^{-B_\mu} U_N \psi_N = e^{-B_\mu} e^{B_\eta} \xi'_N$. Since $\cV_N \leq C \cK \cN_+$ (see \cite[Lemma 6.2]{BBCS}), this will conclude the proof of Prop. \ref{prop:apri}. From (\ref{eq:control-eBmu}) and from the corresponding bound for $B_\eta$, it follows immediately that 
\[ \langle \xi_N, (\cN_+ + 1)^k \xi_N \rangle \leq C \]
Thus, we focus on the expectation of $\cK (\cN_+ +1)^k$. For $s \in (0;1)$, we compute 
\begin{equation*}
    \begin{split}
\frac{d}{ds} \langle \xi'_N,e^{-sB_\eta}e^{sB_\mu}&\cK(\cN_++1)^k e^{-sB_\mu}e^{sB_\eta}\xi'_N\rangle \\= &- \langle \xi_s, [\cK(\cN_++1)^k, B_\mu] \xi_s \rangle + \langle \xi_s , [\cK(\cN_++1)^k,e^{-sB_\mu}B_\eta e^{sB_\mu}] \xi_s \rangle 
    \end{split}
\end{equation*}
where we defined $\xi_s = e^{-sB_\mu} e^{s B_\eta} \xi'_N$. Using (\ref{eq:def_d}), we write  
\begin{equation*}
e^{-sB_\mu}B_\eta e^{sB_\mu} = B_{\eta}+J 
\end{equation*}
where 
\begin{equation*}
J=\sum_{q \in \L^*_+} \eta_qd^{*(s)}_{q}\Big(\gamma^{(s)}_{-q} b_{-q}^*+\sigma^{(s)}_q b_{q}+d^{*(s)}_{-q}\Big)+\sum_{q \in \L^*_+} \eta_q\Big(\gamma^{(s)}_q b_q^*+\sigma^{(s)}_q b_{-q}\Big)d^{*(s)}_{-q} - \mathrm{h.c.}.
\end{equation*}
Therefore, 
\begin{equation}\label{eq:ddsKN} 
    \begin{split}
        \frac{d}{ds} \langle \xi_s,\cK(\cN_++1)^k \xi_s\rangle=&\, \langle \xi_s,[\cK,B_\eta-B_\mu](\cN_++1)^k \xi_s\rangle\\&+  \langle \xi_s,\cK[(\cN_++1)^k,B_\eta-B_\mu]\xi_s\rangle+ \langle \xi_s,[\cK(\cN_++1)^k,J]\xi_s\rangle \,. 
    \end{split}
\end{equation}
With (\ref{eq:tau-bds}), we have 
\begin{equation*}
    \begin{split}
|\langle \xi_s,[\cK,B_\eta-B_\mu](\cN_++1)^k \xi_s\rangle| &=2 \Big| \sum_{p \in \L^*_+}p^2\tau_p\langle \xi_s, b^*_{p}b^*_{-p}(\cN_++1)^k \xi_s\rangle \Big| \\ &\leq C\langle \xi_s,(\cN_++1)^{k+1} \xi_s\rangle \leq C \langle \xi'_N, (\cN_+ + 1)^{k+1} \xi'_N \rangle \leq C
    \end{split}
\end{equation*}
where we applied (\ref{eq:control-eBmu}) and then (\ref{eq:boundKN_+'}). 
Using \eqref{eq:commN^nB}, we can bound the second term on the r.h.s. of (\ref{eq:ddsKN}) by 
\begin{equation*}
    \begin{split}
|\langle \xi_s,\cK &[(\cN_++1)^k,B_\eta-B_\mu]\xi_s\rangle| \\ &\leq C\sum_{p \in \L^*_+}p^2|\tau_q|\, \Big|\langle \xi_s,a^*_pa_pb_qb_{-q} \big[ (\cN_++3)^{k} - (\cN_+ + 1)^k \big] \xi_s\rangle \Big| \\&\leq C \langle \xi_s,\cK(\cN_++1)^k \xi_s\rangle\,.      \end{split}
\end{equation*}
We now focus on the last term on the r.h.s. of (\ref{eq:ddsKN}). Using $J^*=-J$, we have 
\begin{equation*}
    \begin{split}
|\langle \xi_s,[\cK(\cN_++1)^k,J]\xi_s\rangle|\leq 2|\langle \xi_s,\cK(\cN_++1)^kJ\,\xi_s\rangle|\leq 2 (|J_1|+|J_2|+|J_3|)
    \end{split}
\end{equation*}
where we defined 
\begin{equation*}
    \begin{split}
        J_1&=\sum_{q \in \L^*_+}\eta_q \Big\langle \xi_s,\cK(\cN_++1)^k\Big(\gamma^{(s)}_q(b_{-q}d^{(s)}_q+d^{(s)}_qb_{-q})+ \sigma^{(s)}_q(b^*_{q}d^{(s)}_q+d^{(s)}_qb^*_{q})\Big)\xi_s\Big \rangle\\
        J_2&=\sum_{q \in \L^*_+}\eta_q \Big\langle \xi_s,\Big(\gamma^{(s)}_q(b_{-q}d^{(s)}_q+d^{(s)}_qb_{-q})+ \sigma^{(s)}_q(b^*_{q}d^{(s)}_q+d^{(s)}_qb^*_{q})\Big)\cK(\cN_++1)^k\xi_s\Big\rangle\\
        J_3&= \sum_{q \in \L^*_+}\eta_q \Big\langle \xi_s,\cK(\cN_++1)^k\Big(d^{(s)}_qd^{(s)}_{-q}+ d^{*(s)}_qd^{*(s)}_{-q}\Big)\xi_s \Big\rangle
    \end{split}
\end{equation*}
We first consider $J_1$. Making use of the third bound in \eqref{eq:decay_d}, the first contribution to $J_1$ is estimated in absolute value by  
\begin{equation*}
    \begin{split}
\sum_{p,q \in \L^*_+}&p^2|\eta_q|\|a_p\xi_s\| \|a_p(\cN_++1)^{k+1/2} d^{(s)}_q\xi_s\|\\ & \leq \frac{C}{N} \sum_{p,q \in \L^*_+}p^2|\eta_q|\|a_p\xi_s\| \Big( |\mu_p| \|(\cN_++1)^{k+2}a_q\xi_s\|+ |\mu_q|\|(\cN_++1)^{k+2}a_p\xi_s\|\\&\qquad \qquad \qquad \qquad+ |\mu_p||\mu_q|\|(\cN_++1)^{k+5/2}\xi_s\|+ \|(\cN_++1)^{k+3/2}a_pa_q\xi_s\|\\&\qquad \qquad \qquad \qquad  +|\eta_p|\delta_{p,q}\|(\cN_++1)^{k+3/2}\xi_s\|\Big) \\&\leq \frac{C}{N}\langle\xi_s, \cK(\cN_++1)^{2k+4}\xi_s\rangle+C\langle\xi_s, (\cN_++1)^{2k+5}\xi_s\rangle \leq C
    \end{split}
    \label{eq:boundJ1}
\end{equation*}
where we used (\ref{eq:boundKN_+'}) after applying (\ref{eq:a_priori_quadratic}), (\ref{eq:control-eBmu}), to control the growth of moments of $\cN_+$ and of $\cK$ w.r.t. the action of $B_\eta$ and $B_\mu$. 
The other contributions to $J_1$ can be bounded similarly. Also $J_3$ can be bounded using the estimates in \eqref{eq:decay_d}; we skip the details. As for the term $J_2$, the contributions proportional to $\sigma_q^{(s)}$ can be bounded using the fourth bound in \eqref{eq:decay_d}. The second contribution proportional to $\gamma^{(s)}_q$, on the other hand, can be estimated after normal ordering by  
\begin{equation*}
    \begin{split}
\sum_{p,q \in \L^*_+}&p^2|\eta_q| \| a_pd^{*(s)}_q\xi_s\|  \| a_pb_{-q}(\cN_++1)^k\xi_s\|+ \sum_{p \in \L^*_+}p^2|\eta_p| \| d^{*(s)}_p\xi_s\|  \| a_p(\cN_++1)^k\xi_s\|\\&\leq \frac{C}{N}\langle\xi_s, \cK(\cN_++1)^{2k+1}\xi_s\rangle+\frac{C}{N}\langle\xi_s, \cK(\cN_++1)^{3}\xi_s\rangle+C\langle\xi_s, (\cN_++1)^{4}\xi_s\rangle \leq C
    \end{split}
\end{equation*}
where we applied again the fourth bound in \eqref{eq:decay_d}. Finally, let us consider the first contribution to $J_2$. Using the expansion \eqref{eq:d_detailed_expansion},
 we find 
 \begin{equation}
    \begin{split}
        \Big| \sum_{q \in \L^*_+}&\eta_q \gamma^{(s)}_q \langle \xi_s, b_{-q}d^{(s)}_q \cK(\cN_++1)^k\xi_s\rangle \Big| \\ &\leq \frac{C}{N}\int_0^1d \tau \sum_{p,q \in \L^*_+}p^2|\eta_q|\, \Big| \Big\langle \xi_s,b_{-q}e^{-(1-\tau) s B_\mu}\bigg[\gamma_q^{(\tau)} \Big( \mu_q \mathcal{N}_+ b^*_{-q}+\sum_{r\in\Lambda_+^*} \mu_r b^*_r a^*_{-r} a_q \Big)\\
            &\qquad\qquad\qquad+\sigma_q^{(\tau)} \Big( \mu_q \mathcal{N}_+ b_q+\sum_{r\in\Lambda_+^*} \mu_r a^*_{-q} a_{-r} b_r \Big)\bigg]e^{(1-\tau)sB_\mu}a^*_pa_p(\cN_++1)^k\xi_s \Big\rangle\Big|\,. 
    \end{split}
    \label{eq:gamma_term}
\end{equation}
With \eqref{eq:boundap}, \eqref{eq:control-eBmu}, (\ref{eq:a_priori_quadratic}) and, again, (\ref{eq:boundKN_+'}), the first term is bounded by  
\begin{equation*}
    \begin{split}
\frac{C}{N}&\int_0^1d \tau \sum_{p,q \in \L^*_+}p^2|\eta_q||\mu_q|\| a_pe^{-(1-\tau) sB_{\mu}}b_{-q}\cN_{+}e^{(1-\tau) sB_{\mu}}b^*_{-q}\xi_s\| \|a_p(\cN_++1)^k\xi_s\|\\ &\leq \frac{C}{N}\sum_{p,q \in \L^*_+}p^2|\eta_q||\mu_q| \|a_p(\cN_++1)^k\xi_s\|\,\Big(|\mu_p| \|(\cN_++1)^{5/2}\xi_s\|+ \|(\cN_++1)^{2} a_p \xi_s\|\\ &\qquad \qquad \qquad \qquad \qquad \qquad \qquad \qquad + \delta_{p,-q}\|(\cN_++1)^{3/2} \xi_s\|\Big)\\ &\leq \frac{C}{N}\langle\xi_s, \cK(\cN_++1)^{2k}\xi_s\rangle+\frac{C}{N}\langle\xi_s, \cK(\cN_++1)^{4}\xi_s\rangle+C\langle\xi_s, (\cN_++1)^{5}\xi_s\rangle \leq C\,.
    \end{split}
\end{equation*}
The contributions on the r.h.s. of (\ref{eq:gamma_term}) that are proportional to $\s_q^{(\tau)}$ can be handled similarly, using also the estimate 
\[ \Big\| (\cN_++1)^n\sum_{r\in \L^*_+}\mu_r b^*_r a^*_{-r}\xi \Big\| \leq C\|(\cN_++1)^{n+1}\xi\| \]
As for the second term proportional to $\gamma^{(\tau)}_q$ on the r.h.s. of (\ref{eq:gamma_term}),  we write $a^*_{-r} a_q= b^*_{-r} b_q + N^{-1}a^*_{-r} (\cN_+-1)a_q$ and we observe that the contribution associated with the error $N^{-1} a_{-r}^* (\cN_+ + 1) a_q$ is negligible (it can be bounded through Cauchy-Schwarz). The contribution associated with $b_{-r}^* b_q$, on the other hand, can be estimated using  \eqref{eq:def_d} by  
\begin{equation*}
\begin{split}
\frac{C}{N}&\int_0^1d\tau \sum_{p,q,r \in \L^*_+}p^2|\eta_q| |\mu_r|\, \Big|\langle \xi_s,b_{-q}e^{-(1-\tau)s B_\mu}  b^*_r b^*_{-r}e^{(1-\tau )sB_\mu}\\&\qquad \qquad \qquad \qquad (\gamma_q^{((1-\tau)s)}b_q+\sigma_q^{((1-\tau)s)}b^*_{-q}+d^{((1-\tau)s)}_q) a^*_pa_p(\cN_++1)^k\xi_s\rangle \Big| \,.
\end{split}
\end{equation*}
The term proportional to $\gamma_q^{((1-\tau)s)}$ is bounded by 
\begin{equation*}
\begin{split}
\frac{C}{N}&\int_0^1d\tau \sum_{p,q,r \in \L^*_+}p^2|\eta_q| |\mu_r|\,\| a_pe^{-(1-\tau) sB_\mu }b_rb_{-r}e^{(1-\tau) sB_\mu} b^*_{-q}\xi_s\| \|b_qa_p(\cN_++1)^k\xi_s\|\\ &\qquad \qquad+\sum_{p,r \in \L^*_+}p^2 |\eta_p| |\mu_r| \| b_rb_{-r}e^{(1-\tau)sB_\mu} b^*_{-p}\xi_s\| \|a_p(\cN_++1)^k\xi_s\|\\
&\leq \frac{C}{N}\langle\xi_s, \cK(\cN_++1)^{2k+1}\xi_s\rangle+\frac{C}{N}\langle\xi_s, \cK(\cN_++1)^3\xi_s\rangle+\langle\xi_s, (\cN_++1)^4\xi_s\rangle\leq C
\end{split}
\end{equation*}
where we used (\ref{eq:boundap}), (\ref{eq:control-eBmu}), (\ref{eq:a_priori_quadratic}) and  (\ref{eq:boundKN_+'}). The term proportional to $\sigma_q^{((1-\tau)s)}$ can be handled analogously. As for the contribution proportional to $d^{((1-\tau)s)}_q$,  we can apply the fourth bound in \eqref{eq:decay_d} to conclude that 
\begin{equation*}
 \begin{split}
\frac{C}{N^2}&\int_0^1d\tau \sum_{p,q,r \in \L^*_+} p^2 |\eta_q| |\mu_r| \|a_p(\cN_++1)^{k}\| \\ &\hspace{1cm} \times \Big[ \|(\cN_++1)^{3/2}a_p e^{-(1-\tau)s B_\mu} b_r b_{-r} e^{(1-\tau) s B_\mu} b^*_{-q}\xi_s\|\\& \hspace{2cm} +  |\mu_p| \|(\cN_++1)^{2} e^{-(1-\tau)s B_\mu} b_r b_{-r} e^{(1-\tau) s B_\mu} b^*_{-q}\xi_s\| \\ &\hspace{2cm} + \delta_{p,-q}\|(\cN_++1)e^{-(1-\tau)s B_\mu} b_r b_{-r} e^{(1-\tau) s B_\mu} b^*_{-q}\xi_s\| \Big] \\ &\leq \frac{C}{N}\langle\xi_s, \cK(\cN_++1)^{2k}\xi_s\rangle+\frac{C}{N}\langle\xi_s, \cK(\cN_++1)^6\xi_s\rangle+\langle\xi_s, (\cN_++1)^7\xi_s\rangle\leq C
    \end{split}
\end{equation*}
using again (\ref{eq:boundap}) (twice, in the first term in the parenthesis, to pass $a_p$ through the unitary operators $e^{\pm (1-\tau) s B_\mu}$), (\ref{eq:control-eBmu}) (in the second and third terms, to pass powers of $(\cN_+ + 1)$ through the unitary operators), (\ref{eq:a_priori_quadratic}) and  (\ref{eq:boundKN_+'}). Combining all bounds, we arrive at 
\[ \Big| \frac{d}{ds} \langle \xi_s,\cK(\cN_++1)^k \xi_s\rangle \Big| \leq C \langle \xi_s, \cK (\cN_+ + 1)^k \xi_s \rangle + C \]
and the claim follows from Gronwall's Lemma.   
\end{proof}


\begin{thebibliography}{55}
 
\bibitem{ABS}
A.~Adhikari, C.~Brennecke, B.~Schlein.
Bose-Einstein Condensation Beyond the Gross-Pitaevskii Regime.
{\it Ann. Henri Poincar\'e } {\bf 22}, (2021) 1163--1233.



\bibitem{BCGOPS} 
G. Basti, S. Cenatiempo, A. Giuliani, A. Olgiati, G. Pasqualetti, B. Schlein. Upper bound for the ground state energy of a dilute Bose gas of hard spheres. Preprint arXiv:2212.04431.

\bibitem{BCOPS} 
G. Basti, S. Cenatiempo, A. Olgiati, G. Pasqualetti, B. Schlein. A second order upper bound for the ground state energy of a hard-sphere gas in the Gross-Pitaevskii regime. {\it Commun. Math. Phys.} {\bf 399} (2023), 1--55. 

\bibitem{BCS} 
G. Basti, S. Cenatiempo, B. Schlein. A new second-order upper bound for the ground state energy of dilute Bose gases. {\it Forum of Mathematics, Sigma} {\bf 9} (2021) , e74.



\bibitem{BBCS}
C. Boccato, C. Brennecke, S. Cenatiempo, B. Schlein. Bogoliubov Theory in the Gross-Pitaevskii limit. {\it Acta Math.} \textbf{222} (2019), no. 2, 219--335.

\bibitem{BBCS4}
C. Boccato, C. Brennecke, S. Cenatiempo, B. Schlein. Optimal Rate for Bose-Einstein Condensation in the Gross-Pitaevskii Regime. {\it Comm. Math. Phys.} {\bf 376}, (2020),  1311–1395.

\bibitem{BDM} 
C. Boccato, A. Deuchert, D. Stocker. Upper bound for the grand canonical free energy of the Bose gas in the Gross-Pitaevskii limit. Preprint arXiv:2305.19173. 


\bibitem{Bo}
N. N. Bogoliubov. On the theory of superfluidity.
{\it Izv. Akad. Nauk. USSR} {\bf 11} (1947), 77. Engl. Transl. {\it J. Phys. (USSR)} {\bf 11} (1947), 23.

\bibitem{BLPR} L. Bossmann, N. Leopold, S. Petrat, S. Rademacher. Ground state of Bose gases interacting through singular potentials. Preprint arXiv:2309.12233.  

\bibitem{BPS} 
L. Bossmann, S. Petrat, R. Seiringer. Asymptotic expansion of the low-energy excitations for weakly interacting bosons. {\it Forum Math. Sigma} {\bf 9} (2021), 1--61. 

\bibitem{BBCO}
C. Brennecke,  M. Brooks, C. Caraci, J. Oldenburg. A Short Proof of Bose–Einstein Condensation in the Gross–Pitaevskii Regime and Beyond.  {\it Ann. Henri Poincaré}, pp. 1-21. (2024).

\bibitem{BCaS}
C. Brennecke, M. Caporaletti, B. Schlein. Excitation spectrum for Bose gases beyond the Gross-Pitaevskii regime. Reviews in Mathematical Physics 34, no. 09 (2022), 2250027.


\bibitem{BS}
C. Brennecke, B. Schlein. Gross-Pitaevskii dynamics for Bose-Einstein condensates. {\it Anal. PDE} {\bf 12} (2019), no. 6, 1513-1596.

\bibitem{BSS1}
 C. Brennecke, B. Schlein, S. Schraven. Bose-Einstein Condensation with Optimal Rate for Trapped Bosons in the Gross-Pitaevskii Regime. {\it Mathematical Physics, Analysis and Geometry} {\bf 25}, no. 2 (2022), 12.
 

\bibitem{BSS2}
 C. Brennecke, B. Schlein, S. Schraven. Bogoliubov Theory for Trapped Bosons in the Gross-Pitaevskii Regime. {\it Ann. Henri Poincaré} (2022). 


\bibitem{B}
M. Brooks. Diagonalizing Bose Gases in the Gross-Pitaevskii Regime and Beyond. Preprint arXiv:2310.11347.

\bibitem{CD} 
M. Caporaletti, A. Deuchert. Upper bound for the grand canonical free energy of the Bose gas in the Gross-Pitaevskii limit for general interaction potentials. Preprint arXiv:2310.12314. 

\bibitem{CCS1}
C. Caraci, S. Cenatiempo, B. Schlein. Bose-Einstein condensation for two dimensional bosons in the Gross-Pitaevskii regime. {\it J. Stat. Phys.} {\bf 183} (2021), no. 3, 39. 


\bibitem{CCS2}
C. Caraci, S. Cenatiempo, B. Schlein. The excitation spectrum of two dimensional Bose gases in the Gross-Pitaevskii regime.  {\it Ann. H. Poincar{\'e}}, (2023), 1-52. 



\bibitem{DS} A. Deuchert, R. Seiringer. Gross--Pitaevskii Limit of a Homogeneous Bose Gas at Positive
Temperature. {\it Arch. Rat. Mech. An.} {\bf 236} (2020), p. 1217.


\bibitem{Dy}
F.J. Dyson. Ground-State Energy of a Hard-Sphere Gas. {\it Phys. Rev.} {\bf 106} (1957), 20--26.


\bibitem{FS1}
S. Fournais, J.P. Solovej. The energy of dilute Bose gases. {\it Ann. Math.} {\bf 192} (2020).

\bibitem{FS2} 
S. Fournais, J.P. Solovej. The energy of dilute Bose gases II: The general case. {\it Inv. Math.}  232, no. 2, (2023), 863-994.


\bibitem{HHNST} 
F. Haberberger, C. Hainzl, P. T. Nam, R. Seiringer, A. Triay. The free energy of dilute Bose gases at low temperatures. Preprint arXiv:2304.02405.

\bibitem{HST} C. Hainzl, B. Schlein, A. Triay. Bogoliubov theory in the Gross-Pitaevskii limit: a simplified approach. {\it Forum of Math. Sigma} {\bf 10} (2022), e90.

\bibitem{HP} 
N. M. Hugenholtz, D. Pines. Ground-State Energy and Excitation Syectrum of a System of Interacting Bosons. {\it Phys. Rev.} {\bf 116} (1959), no. 3, 489--506. 


\bibitem{LHY} T. D. Lee, K. Huang, C. N. Yang. Eigenvalues and Eigenfunctions of a Bose System of Hard Spheres and Its Low-Temperature Properties. {\it Phys. Rev.} \textbf{106} (1957), 6, 1135--1145.

\bibitem{LS1}
E.~H.~Lieb and R.~Seiringer.
\newblock Proof of {B}ose-{E}instein condensation for dilute trapped gases.
\newblock {\it Phys. Rev. Lett.} \textbf{88} (2002), 170409.



\bibitem{LSY}
E.~H.~Lieb, R.~Seiringer, and J.~Yngvason.
\newblock Bosons in a trap: A rigorous derivation of the {G}ross-{P}itaevskii
  energy functional. \newblock {\it Phys. Rev. A} \textbf{61} (2000), 043602.

\bibitem{LS2}
E.~H.~Lieb and R.~Seiringer.
\newblock Derivation of the Gross-Pitaevskii Equation for Rotating Bose Gases
\newblock {\it Comm. Math. Phys.} \textbf{264} (2006), 505--537.

\bibitem{LY}
E.~H.~Lieb, J.~Yngvason. Ground State Energy of the low density Bose Gas. {\it Phys. Rev. Lett.} {\bf 80} (1998), 2504--2507.

\bibitem{NN}
P.~T.~{Nam}, M. Napi{\'o}rkowski. Two-term expansion of the ground state one-body density matrix of a mean-field Bose gas. {\it Calc. Var. and PDE} {\bf 60} (2021), art. 99. 

\bibitem{NNRT}
P.~T.~{Nam}, M. Napi{\'o}rkowski, J. Ricaud, A. Triay. Optimal rate of condensation for trapped bosons in the Gross--Pitaevskii regime. {\it Anal. PDE} 15, no. 6 (2022), 1585-1616.

\bibitem{NRS}
P.~T.~{Nam}, N. ~{Rougerie}, R.~Seiringer. Ground states of large bosonic systems: The Gross-Pitaevskii limit revisited. {\it Anal. PDE} {\bf 9} (2016), no. 2, 459--485.

\bibitem{NT} 
P.~T.~{Nam}, A. Triay. Bogoliubov excitation spectrum of trapped Bose gases in the Gross-Pitaevskii regime. {\it J. de Math. Pur. App.} (2023).


\bibitem{P3}
A. Pizzo. Bose particles in a box III. A convergent expansion of the ground state of the Hamiltonian in the mean field limiting regime. Preprint arxiv:1511.07026. 


\bibitem{Sa}
K. Sawada. Ground-State Energy of Bose-Einstein Gas with Repulsive Interaction. {\it Phys. Rev.} {\bf 116} (1959), no. 6,  1344--1358.



\bibitem{YY}
H.-T. Yau, J. Yin. The second order upper bound for the ground state energy of a Bose gas. {\it J. Stat. Phys.} {\bf 136} (2009), no. 3, 453--503.

\bibitem{W}
T.T. Wu. Ground State of a Bose System of Hard Spheres. {\it Phys. Rev.} {\bf 115} (1959), no. 6, 1390--1404. 

\end{thebibliography}
\end{document}